\newif\iffull
\fulltrue    

\newif\ifanon

\iffull
\documentclass[12pt]{article}
\usepackage[margin=1in]{geometry}
\else
\documentclass[11pt]{article}
\usepackage[margin=.5in]{geometry}
\fi

\usepackage{style/qstyle}

\iffull\else
\usepackage{titling}
\fi

\begin{document} 

\title{\textbf{Rigidity for Monogamy-of-Entanglement Games}}
\ifanon
\author{\vspace{0mm}}
\date{\vspace{-10mm}}
\else
\author{
	Anne Broadbent and Eric Culf
}
\date{\vspace{-10mm}} 
\affil{University of Ottawa,  Department of Mathematics and Statistics
\iffull
\footnote{\texttt{\{abroadbe,eculf019\}@uottawa.ca}}
\else
\fi
}
\setcounter{Maxaffil}{0}
\renewcommand\Affilfont{\small}
\fi

\maketitle
\iffull
\begin{abstract}
	
	In a monogamy-of-entanglement (MoE) game, two players who do not communicate try to simultaneously guess a referee's measurement outcome on a shared quantum state they prepared. We study the prototypical example of a game where the referee measures in either the computational or Hadamard basis and informs the players of her choice.
	
	We show that this game satisfies a rigidity property similar to what is known for some nonlocal games. That is, in order to win optimally, the players' strategy \emph{must} be of a specific form, namely a convex combination of four unentangled optimal strategies generated by the Breidbart state. We extend this to show that strategies that win near-optimally must also be near an optimal state of this form. We also show rigidity for multiple copies of the game played in parallel.
	
	We give three applications:  (1) We construct for the first time a weak string erasure (WSE) scheme where the security does not rely on limitations on the parties' hardware. Instead, we add a prover, which enables security via the rigidity of this MoE game. (2) We show that the WSE scheme can be used to achieve bit commitment in a model where it is impossible classically. (3) We achieve everlasting-secure randomness expansion in the model of trusted but leaky measurement and untrusted preparation and measurements by two isolated devices, while relying only on the temporary assumption of pseudorandom functions. This achieves randomness expansion without the need for shared entanglement.
\end{abstract}

\fi

\iffull
\newpage
\setcounter{tocdepth}{2}
\tableofcontents
\newpage
\fi


\iffull
\section{Introduction}
\fi

Monogamy-of-entanglement (MoE) games provide an intuitive way to understand the strength of quantum multipartite correlations. Such games pit two cooperating players, usually named Bob and Charlie, against an honest referee, Alice. The players try, without communicating, to simultaneously guess the outcome of Alice's measurement on a quantum state provided by the players and with which they may share entanglement freely. Interestingly, any one of the players can always correctly guess the result of any projective measurement Alice makes, by providing her with one register of a maximally entangled state, whereas two players are prohibited from \emph{simultaneously} doing as well since tripartite correlations of the shared state are weaker.

The quintessential MoE game is the original example introduced by Tomamichel, Fehr, Kaniewski, and Wehner \cite{TFKW13}. In this game, Alice's space consists of a single qubit and she measures either in the computational basis or the Hadamard basis with equal probability to get a one-bit answer. As shown there, Bob and Charlie can win with probability at most $\cos^2{\frac{\pi}{8}}\approx0.85$. Further, this game has an exponentially small winning probability when played in parallel, which was shown to yield  applications such as a one-sided device independent quantum key distribution (DI-QKD) scheme. The TFKW game has a particularly simple optimal strategy: Bob and Charlie share no entanglement; they just send Alice a pure Breidbart state $\ket{\beta}\propto\ket{0}+\ket{+}$, that sits directly between the computational zero $\ket{0}$ and the Hadamard zero $\ket{+}$, and always guess $0$ for the measurement outcome. It is straightforward to see that, due to the symmetries of Alice's measurement bases under the action of the Pauli operators, there are at least $4$ optimal unentangled strategies: the Wiesner-Breidbart states $\ket{\beta}$, $X\ket{\beta}$, $Z\ket{\beta}$, $XZ\ket{\beta}$\iffull, illustrated in \cref{fig:bloch}\fi. But the question remains: are these all the possible optimal strategies? Particularly, are there optimal strategies where the players use entanglement? This question is tantamount to asking about the \emph{rigidity} of the TFKW~game.

The idea of rigidity, first formally introduced by Mayers and Yao \cite{MY04}, is that certain games can be used to ``self-test'' quantum states: if such a game is won with high enough probability, then the self-test property tells us that the players must hold some quantum state, up to local isometry. More general are robust self-tests, where even a near-optimal winning probability gives a guarantee that the state is near to this optimal one. Up until now, the study of rigidity has been limited to \emph{nonlocal games}. Nonlocal games are similar to monogamy-of-entanglement games (both belong to the class of extended nonlocal games \cite{JMRW16}), except the referee is classical while the players might be asked different questions and be expected to provide different answers. This area of study grew around the CHSH game. This game, introduced by Clauser, Horne, Shimony, and Holt \cite{CHSH69} as a discrete-variable analogue of a Bell inequality \cite{Bel64}, was known, even before rigidity was formalised, to self-test a maximally-entangled state on two qubits \cite{Tsi93}. This result was later extended to be robust \cite{MYS12} and to hold under parallel repetition \cite{Col17}. \iffull
The rigidity of the CHSH game has found many applications: for example it was used to construct a protocol for quantum delegated computation, which was then used to show the equivalence of complexity classes $\tsf{QMIP}=\tsf{MIP}^\ast$ \cite{RUV13}. Other examples of nonlocal games include the Mermin-Peres magic square game --- which can always be won and self-tests two copies of the maximally entangled state, and was used to show $\tsf{MIP}^\ast=\tsf{RE}$ \cite{JNV+20arxiv} --- and more generally linear constraint games \cite{CMMN20}.
\fi

\iffull
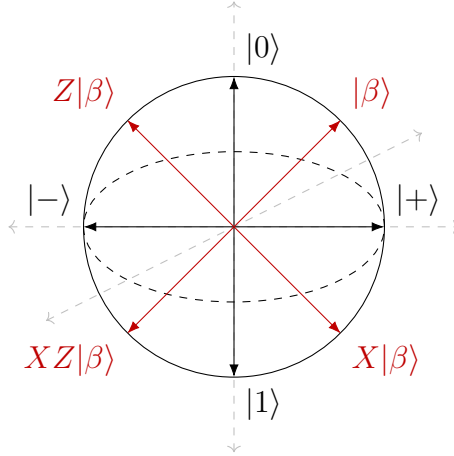
\begin{figure}
	\centering
	\begin{tikzpicture}
	\draw[lightgray, dashed, <->] (-3,0) -- (3,0);
	\draw[lightgray, dashed, <->] (0,-3) -- (0,3);
	\draw[lightgray, dashed, <->] (-2.5,-1.25) -- (2.5,1.25);
	\draw (2,0) arc (0:360:2);
	\draw[dashed] (0,0) ellipse (2cm and 1cm);
	\draw[-Latex] (0,0) -- (0,2) node[above right]{$\ket{0}$};
	\draw[-Latex] (0,0) -- (0,-2) node[below right]{$\ket{1}$};
	\draw[-Latex] (0,0) -- (2,0) node[above right]{$\ket{+}$};
	\draw[-Latex] (0,0) -- (-2,0) node[above left]{$\ket{-}$};
	
	\draw[darkred, -Latex] (0,0) -- (1.42,1.42) node[above right]{$\ket{\beta}$};
	\draw[darkred, -Latex] (0,0) -- (1.42,-1.42) node[below right]{$X\ket{\beta}$};
	\draw[darkred, -Latex] (0,0) -- (-1.42,-1.42) node[below left]{$XZ\ket{\beta}$};
	\draw[darkred, -Latex] (0,0) -- (-1.42,1.42) node[above left]{$Z\ket{\beta}$};
	\end{tikzpicture}
	\caption{Positions of the Wiesner-Breidbart states on the Bloch sphere. They form a pair of bases analogous to the conjugate-coding bases, but rotated by $\frac{\pi}{4}$ so that each vector is located at the midpoint between a vector from the computational basis and one from the Hadamard basis.}
	\label{fig:bloch}
\end{figure}
\newpage\fi
Our main contribution is to prove the first rigidity result for a monogamy-of-entanglement game:

\begin{faketheorem}[informal]
	The state of any optimal strategy for the TFKW game is given as a convex combination of the unentangled optimal states $\ket{\beta}$, $X\ket{\beta}$, $Z\ket{\beta}$, $XZ\ket{\beta}$. This is robust and extends to multiple rounds played in parallel.
\end{faketheorem}

By convex combination, we mean a superposition of tensor product states where the components on Alice's register are the optimal Wiesner-Breidbart states and the components on Bob and Charlie's register have orthogonal supports. That is, Bob and Charlie can simultaneously distinguish which unentangled optimal state Alice receives. Note that a similar notion of rigidity holds for some nonlocal games: for example, Man\v{c}inska, Nielsen, and Prakash \cite{MNP21arxiv} show that the glued magic square game self-tests a convex combination of inequivalent optimal strategies. This requirement on optimal strategies of the TFKW game forces Bob and Charlie to \emph{not} use any of their shared entanglement while playing.

For applications, it is often necessary to extend the rigidity result to be robust and to the scenario where games are played in parallel. This is because playing the game only once gives essentially no information on the winning probability of the strategy used. What Alice can do to remedy this is to get Bob and Charlie to play many games at the same time and use that information to build up statistics about how often they win. As such, she needs the result to be robust --- a guarantee that the state is near-optimal if the winning probability is near-optimal --- as the sampling cannot quite show that the strategy is optimal. Also, she needs the result to hold for games played in parallel, to ensure that there is nothing different and more exotic they may do using their entanglement to win multiple games optimally. We show that the rigidity of the TFKW game holds in this general case.

\paragraph{Application 1: Weak string erasure.}
We construct for the first time a weak string erasure (WSE) scheme that is secure against adversaries with unrestricted quantum systems. WSE is a cryptographic primitive introduced by K\"onig, Wehner, and Wullschleger \cite{KWW12} that allows the sharing of partial information between mistrustful parties, a sender Alice and a receiver Bob. In WSE, Alice receives a random bit string $x$ while Bob receives a substring; Bob knows which bits of $x$ he holds but is unable to determine the remainder, while Alice is unable to determine which substring Bob holds. As shown in~\cite{KWW12}, WSE implies both bit commitment and oblivious transfer from Alice to Bob. Since these are information-theoretically impossible in both the classical and quantum plain model \cite{May96, LC97,BS16}, additional assumptions are necessary to be able to realise WSE. In \cite{KWW12}, they use a noisy-storage model to limit the amount of storage a dishonest party can access.

In our model, we introduce a third party, a prover Charlie who is initially in full collusion with Bob, but who is isolated from Bob once Alice measures. Under the assumption of a public broadcast from Alice to Bob and Charlie, we  are able to exploit the rigidity of the TFKW game to arrive at a secure scheme for WSE. This scheme requires no entanglement for the honest parties and may be run with one round of communication --- in particular it can be realised as a relativistic prepare-and-measure scheme.

\paragraph{Application 2: Bit commitment.}
Two-prover bit commitment was studied before in the classical context \cite{BGKW88}, where it was shown that separating the \emph{sender} into two isolated parties can be used to ensure the binding property (see also \cite{CSST11}). In contrast, the WSE that we achieve implies, using \cite{KWW12}, a bit commitment with two isolated \emph{receivers} and a single sender. To the best of our knowledge, this is the first such scheme; furthermore, we show that, with classical communication only, our model reduces to the single-receiver model --- where unconditionally secure bit commitment is impossible --- meaning that we have identified a new qualitative advantage for quantum communication in cryptography.

\paragraph{Application 3: Everlasting randomness expansion.}
\iffull
Randomness is a precious resource for computation and cryptography. Pseudorandom generators are functions that produce large amounts of randomness from a small random seed, but the quality of this randomness is inherently based on a computational assumption, \emph{e.g.}~the existence of one-way functions. Thus, given sufficient computational power or time, an adversary can eventually break the scheme.
\fi

Quantum entanglement has long been known to provide an advantage in creating unconditionally secure randomness \cite{Col06,AM16}. By verifying that two isolated parties violate a Bell inequality, a verifier is able to guarantee, due to the randomness inherent in quantum mechanics, that the players' outputs provide intrinsic, fresh randomness. Such schemes are able to yield exponential randomness expansion \cite{VV12}. Further, using the rigidity of the CHSH game, it is possible to guarantee that the randomness is secure against side information, and thus allow composition, providing arbitrarily large randomness expansion \cite{CY14}.
The technical difficulty with these schemes is that they require entanglement between isolated parties, which remains difficult to generate in sufficient quantities. Based on the experimental demonstration of a loophole-free Bell inequality violation \cite{HBD+15}, recent work has been able to achieve a randomness expansion of 24\% over a period of 91 hours; however, the new randomness is only secure against classical and not quantum side information \cite{SZB+21}.

Here, we give a protocol where entanglement between isolated parties is not required in order to expand randomness. In order to achieve this, we make use of an adapted version of the WSE protocol as described above. First, the questions Alice asks are pseudorandom rather than uniformly random; this allows Alice to start with only a small random seed.   With polynomial overhead, we can extract statistically near-uniform randomness using the rigidity of the TFKW game. To do this, Alice uses many of the bits to verify that the shared state is near the state of an optimal strategy, and then extracts randomness using her knowledge of the remainder of the state. We thus require the computational assumption to hold during the interaction of the protocol, after which the output randomness becomes nearly indistinguishable from uniform, even to an unbounded adversary --- this concept is called \emph{everlasting} security and was previously studied in the context of quantum key distribution \cite{SML10} and multi-party computation \cite{Unr13}. Furthermore, we note that in our model, all of the measurement settings Alice uses can be leaked as she measures, without compromising the security or uniformity of the randomness.

\iffull
\subsection{Summary of Techniques}
\fi

\iffull
In this section, we summarise the techniques used to show our results. First, we mention our interpretation of MoE games, and then go through the general method we follow to prove rigidity of the TFKW game and apply it to achieve weak string erasure and everlasting randomness expansion.
\fi

\paragraph{Monogamy-of-entanglement games.}
We give an expression of a two-answer MoE game, such as the TFKW game, in terms of a game polynomial where the variables are Bob and Charlie's observables. In this way, we may study the strategies of a game by studying the positivity of this operator-valued polynomial. This technique expands upon one that has been used previously to study nonlocal games~\cite{CMMN20}.

\paragraph{Rigidity.}
We present a sum-of-squares (SOS) decomposition of the game polynomial for the TFKW game. The state $\ket{\psi}$ of any optimal strategy is an eigenspace of the game polynomial in terms of the observables of that strategy, which provides a selection of relations for the observables. There are two types of relations that come out: one allows to exchange Bob and Charlie's observables and the other gives that $\ket{\psi}$ is an eigenvector of a particular sum of observables. In particular, these imply that either of the players' observables must commute with respect to the state, generating a $\ket{\psi}$-representation of $\Z_2^2$. As such, we invoke the Gowers-Hatami theorem as in \cite{Vid18} to locally dilate the players' space isometrically and transform this into a bona fide representation. These observables are simultaneously diagonalisable, so the dilated shared space can be decomposed as a direct sum of orthogonal subspaces on which they act as scalars. Returning to the relations from the SOS decomposition using the dilated observables allows us to constrain where the shared state lives in this orthogonal sum and show that the components on Alice's space must take the form $X^{s_0}Z^{s_1}\ket{\beta}$.

\iffull
We then build on this technique to show the rigidity in the robust case. Here, however, since the winning probability is assumed to be some $\varepsilon>0$ smaller than optimal, the value of each of the terms in the SOS decomposition are not zero when acting on the state, but rather in $O(\varepsilon)$. Nevertheless, we can use the relations to get an approximate representation, which we dilate similarly with Gowers-Hatami. Of course, this cannot give that the state is exactly a convex combination as above, but rather that its projection onto the unwanted subspaces is small, giving that this is $O(\sqrt{\varepsilon})$ close to an optimal state.

We show the exact rigidity for a parallel repetition of $n$ TFKW games by extracting many optimal strategies for a single game, assuming Bob and Charlie can guess each of the answer bits for the repeated games with optimal probability. We show first that the observables related to each copy of the TFKW game must act in the same way on $\ket{\psi}$ by using the rigidity decomposition, and then use this as tool to show that all of the observables commute. This induces, again with Gowers-Hatami, a representation of $(\Z_2^2)^n$ and lets us conclude in a similar way as for the single-game case that the state must be a convex combination of tensor products of states of the form $X^{s_0}Z^{s_1}\ket{\beta}$.

The most general rigidity result we prove is the robust case of the parallel repetition of TFKW games. To generalise the exact-case method, we use a technique of \cite{Col17} to extract sufficiently many strategies for TFKW that win near-optimally. Proceeding similarly as before, we get that the state is $O(n^3\sqrt{\varepsilon})$ away from an optimal state.
\fi

Finally, we adapt a technique of \cite{RUV13} to be able to pass from winning statistics Alice may observe when playing TFKW games in parallel to a guarantee on the winning probability of a large subset of the games. Knowing upper bounds on the winning probability of each of the games, we can couple independent Bernoulli random variables to each game, and use Hoeffding's inequality to show that there is but a low probability that the players win most games while the winning probability for too many of them is more than $\varepsilon$ away from optimal.

\paragraph{Weak string erasure.}
We construct a WSE scheme whose security is based upon the rigidity of the TFKW game. The receiver Bob prepares a state $\rho_{ABC}$ shared between Alice, Bob, and Charlie, where Alice holds $N\in\ttt{poly}(n)$ qubits. In the honest case, this has the form of an unentangled optimal strategy for the parallel-repeated TFKW game. Then, Alice verifies that the state must be near an optimal state for the TFKW game by playing the game with Bob and Charlie using $N-n$ of her qubits. This check fails with exponentially small probability in $n$. On the remaining $n$ qubits, however, she measures in a random Wiesner-Breidbart basis, \emph{i.e.}~either the basis $\ket{\beta},XZ\ket{\beta}$ or the basis $Z\ket{\beta},X\ket{\beta}$. Giving Bob the information about which basis she chose for these $n$ qubits, he may guess on average half of the bits and have no information about the rest. This provides security against a dishonest Bob. For security against a dishonest Alice, we note that the rigidity still gives Bob the freedom to choose the Wiesner-Breidbart state on the register he gives to Alice. \iffull It can be seen from \cref{fig:bloch} that t\fi hese states constitute a pair of mutually-unbiased bases. Therefore, if Bob chooses the state randomly, this eliminates Alice's chance of guessing which bits he knows. The isolation requirement between Bob and Charlie is necessary to prevent an attack where they jointly share a maximally entangled state with Alice and then can always measure each bit in the correct basis. The requirement that Alice broadcast publicly which $n$ bits are used to generate the output string is to prevent an attack where she asks Bob and Charlie to play the TFKW game on different bits, and uses Charlie's replies to extract information about Bob's prepared conjugate-coding~basis.

\paragraph{Everlasting Randomness Expansion.} We use the rigidity of the TFKW game, as well as a computational assumption on the existence of pseudorandom generators, to construct a randomness expansion scheme that is everlasting, in the sense that the output randomness is guaranteed to be near-uniform in trace norm, as long as the computational assumption is not broken during the execution of the protocol.
 As in the previous protocol, Alice interacts with a pair of adversaries, Bob and Charlie, and they all share an adversarially-prepared state $\rho_{ABC}$, where Alice holds $N\in\ttt{poly}(n)$ qubits. Alice plays the TFKW game on $N-n$ of the qubits to verify that the shared state is near an optimal state. However, rather than choosing the locations and questions for the TFKW game rounds uniformly at random, she chooses them by sampling the output of a pseudorandom generator, given a random seed. Bob and Charlie, who are assumed to be computationally bounded, have only a negligible probability of distinguishing this from the uniformly random case, and thus this check has only a negligibly small probability of failure. Alice measures each of the remaining~$n$ qubits in the basis $\ket{0_\circlearrowleft},\ket{1_\circlearrowleft}$ that diagonalises the Pauli $Y$ operator. Since this basis is mutually unbiased with both of the Wiesner-Breidbart bases, the outcome is nearly uniformly random, and neither Bob nor Charlie have information on what this outcome is, as long as they stay isolated.

\iffull
\subsection{Further Related Work}
The study of monogamy-of-entanglement games is a burgeoning field in quantum information, with several applications to cryptography. As mentioned earlier, these games were introduced in~\cite{TFKW13}, where they also introduced an important technique using overlaps of measurement operators to upper bound the winning probabilities. Johnston, Mittal, Russo, and Watrous \cite{JMRW16} adapted the overlap technique to show that all MoE games with two questions can be won using an unentangled strategy and satisfy perfect parallel repetition, and gave a generalisation of the NPA hierarchy \cite{NPA08} that can be used on MoE games. Broadbent and Lord \cite{BL20} used the TFKW game to study uncloneable encryption in the quantum random oracle model. Most recently, Coladangelo, Liu, Liu, and Zhandry \cite{CLLZ21} defined a new MoE game of a slightly different style where Alice measures in a basis of subspace coset states, and Bob and Charlie try to guess different strings, for which an upper bound on the winning probability was shown by Culf and Vidick \cite{CV21arxiv} using overlaps. In \cite{CLLZ21}, they use this game along with some computational assumptions to construct schemes for uncloneable decryption and copy-protection of pseudorandom functions.

The idea of using the probabilistic nature of quantum mechanics to create randomness is well-established. Colbeck \cite{Col06} pointed out that violations of Bell inequalities can be used to generate randomness, which, as mentioned, was expanded by Vazirani and Vidick \cite{VV12}, and then Coudron and Yuen \cite{CY14}, to give very powerful device-independent randomness expansion protocols. Our scheme may be contrasted with the work of Brakerski, Christiano, Mahadev, Vazirani, and Vidick \cite{BCM+18}, where they also use a short-term computational assumption to achieve everlasting randomness expansion. They use the learning with errors (LWE) assumption to construct noisy trapdoor claw-free functions, using which it is possible to verify that an untrusted quantum device approximately prepares states in the Hadamard basis and then measures in the computational basis. This protocol does not require a communication assumption or a trusted measurement, as ours does. On the other hand, it requires a specific computational assumption, far less general than existence of pseudorandom generators, and requires a full fault-tolerant quantum computer in the honest case, rather than simply preparation and measurement of single-qubit states. Less demanding models, where some aspects of the devices are trusted, have also been considered. In a semi-device-independent model, it is assumed that the dimensions of the devices' Hilbert spaces are constrained \cite{PB11}. Randomness expansion schemes in this model do not require entanglement, but make use of strong assumptions on the devices: a finite distribution of states and measurements, and no entanglement with another system \cite{LYW+11}. There are also more asymmetric models, like quantum steering, where one of the devices may be completely trusted while the other is untrusted \cite{BCWSW12}.

\subsection{Acknowledgements}
We would like to thank Arthur Mehta for introducing us to sum-of-squares decompositions, and S\'ebastien Lord for many insightful discussions.

This work was supported by the Air Force Office of Scientific Research under award number FA9550-20-1-0375, Canada’s NFRF and NSERC, an Ontario ERA, and the University of Ottawa’s Research Chairs program.

\subsection{Outline}

In \cref{sec:prelims} we present the notation and technical facts from the theories of quantum information, probability, and approximate representation of finite groups that we use throughout the paper. Next, in \cref{sec:monogamy}, we formally define the concept of a monogamy-of-entanglement game and present different ways of understanding the winning probabilities of strategies for these games. In \cref{sec:rigidity}, we prove rigidity for the TFKW game. Our most general rigidity results are given by \cref{thm:rob-par} and \cref{thm:prob-gen-rob-par}. Lastly, in \cref{sec:application} we apply the rigidity result to construct a weak string erasure scheme, which we relate to a construction of bit commitment; and combine it with a computational assumption to construct a everlasting randomness expansion scheme.

\fi

\iffull
\section{Preliminaries} \label{sec:prelims}

In this section, we go over the basic technical facts needed in the remainder of the paper. First, in \cref{sec:notation}, we introduce the general notation we use, which is largely standard. Next, in \cref{sec:quant-inf}, we go over the basic objects from quantum information theory we need, including the definitions and properties of some important states and operators on the space of a qubit we see throughout. In \cref{sec:probability}, we touch on some notation and results from probability theory. Finally, in \cref{sec:representation}, we recall some results from the representation theory of finite groups and its generalisation to approximate representations. We also prove that operators we encounter later generate approximate representations.

\subsection{Notation} \label{sec:notation}

A Hilbert space is a $\C$-vector space with an inner product that is complete as a metric space. Here, we only consider finite-dimensional Hilbert spaces so the completeness is always guaranteed. As is customary, we use Dirac bra-ket notation. That is, a vector in a Hilbert space $\tsf{H}$ are written as a \emph{ket} $\ket{v}\in\tsf{H}$; the inner product of two vectors $\ket{v},\ket{w}\in\tsf{H}$ is written as a \emph{braket} $\braket{v}{w}$; for every ket $\ket{v}\in\tsf{H}$, there is a corresponding \emph{bra} $\bra{v}\in\tsf{H}^\ast$, which is the unique element of the dual space such that $\bra{v}(\ket{w})=\braket{v}{w}$ for every $\ket{w}\in\tsf{H}$; and for any $\ket{v},\ket{w}\in\tsf{H}$, the \emph{ketbra} $\ketbra{v}{w}$ is the linear map $\tsf{H}\rightarrow\tsf{H}$ such that $\ketbra{v}{w}(\ket{u})=\braket{w}{u}\ket{v}$ for all $\ket{u}\in\tsf{H}$. We denote the adjoint of a linear map $L:\tsf{H}\rightarrow\tsf{K}$ with respect to the brakets on the two spaces by $L^\dag:\tsf{K}\rightarrow\tsf{H}$. For a one-dimensional Hilbert space, we write $\ket{v}$ to mean $\spn_\C\{\ket{v}\}$ when there is little chance of confusion.

Let $\tsf{H}$ and $\tsf{K}$ be Hilbert spaces. We use the following important operator spaces. We denote the space of all linear operators $\tsf{H}\rightarrow\tsf{K}$ as $\mc{L}(\tsf{H},\tsf{K})$, and write $\mc{L}(\tsf{H}):=\mc{L}(\tsf{H},\tsf{H})$. The set of invertible linear maps is denoted $\mc{GL}(\tsf{H})\subseteq\mc{L}(\tsf{H})$. Next, the \emph{Hermitian} operators on $\tsf{H}$ are the $L\in\mc{L}(\tsf{H})$ such that $L=L^\dag$, and we write the set of Hermitian operators $\mc{H}(\tsf{H})$. The \emph{positive} (semidefinite) operators on $\tsf{H}$ are the operators $P\in\mc{H}(\tsf{H})$ such that $\braket{v}{P}{v}\geq 0$ for all $\ket{v}\in\tsf{H}$, and we denote the set of these operators $\mc{P}(\tsf{H})$. We often write $P\geq 0$ to mean $P\in\mc{P}(\tsf{H})$. An \emph{isometry} from $\tsf{H}$ to $\tsf{K}$ is an operator $V\in\mc{L}(\tsf{H},\tsf{K})$ such that $V^\dag V=\Id_H$, and the set of isometries is $\mc{U}(\tsf{H},\tsf{K})$. Note that isometries may only exist if $\dim\tsf{K}\geq\dim\tsf{H}$. Finally, a \emph{unitary} operator is an isometry $U:\tsf{H}\rightarrow\tsf{H}$, i.e. $U^\dag U=UU^\dag=\Id_H$, and the set of unitaries is $\mc{U}(\tsf{H})=\mc{U}(\tsf{H},\tsf{H})$. Operators $A,B\in\mc{L}(\tsf{H})$ are said to \emph{commute} if $AB=BA$ and \emph{anticommute} if $AB=-BA$. The \emph{commutator} of two operators is $[A,B]=AB-BA$.

We consider the natural numbers to be $\N=\{1,2,3,\ldots\}$, and for $n\in\N$ write the subset $[n]=\{1,\ldots,n\}\subseteq\N$. We see elements of the vector space $\Z_2^n$ for $n\in\N$ as bit strings, so we write them as a concatenation $x=x_1x_2\ldots x_n$. We define, for $i\in[n]$, $1^i\in\Z_2^n$ as the bit string that is $1$ in position $i$ and $0$ elsewhere. Analogously, for a subset $I\subseteq[n]$, write $x_I=x_{i_1}x_{i_2}\ldots x_{i_k}$ for $I=\{i_1,\ldots,i_k\}$ with $i_1<i_2<\ldots<i_k$; and $1^I\in\Z_2^n$ the string that is $1$ at indices in $I$ and $0$ elsewhere.

\subsection{Quantum Information}\label{sec:quant-inf}

	The classical states of a system are represented by a finite set $H$ called a \emph{register}. The pure quantum states are represented by superpositions of elements of the register, so vectors with norm $1$ in the Hilbert space $\tsf{H}=\spn_\C\set*{\ket{h}}{h\in H}\cong\C^{|H|}$, where the spanning set is an orthonormal basis. More generally, quantum states may be seen as mixed states, which are probability distributions over pure quantum states. Every mixed state may be represented as a \emph{density operator} of the form $\rho=\sum_{i}p_i\ketbra{\psi_i}$, where $p_i\geq 0$, $\sum_{i}p_i=1$, and $\ket{\psi_i}\in\tsf{H}$ is a pure quantum state. The set of mixed states corresponds to the set of positive operators with trace $1$, and we write this set $\mc{D}(\tsf{H})$. Every mixed state may be purified by appending some auxiliary register (\cref{lem:pure-state}). We call a state \emph{classical} if it is diagonal in the canonical basis of $\tsf{H}$ -- it corresponds exactly to a probability distribution on $H$.
	
	A quantum measurement is represented by a positive operator-valued measurement (POVM), which is a map \begin{align}
		\begin{matrix}P:&X&\rightarrow&\mc{P}(\tsf{H})\\&x&\mapsto&P_x\end{matrix}
	\end{align}
	where $X$ is the (finite) set of possible measurement outcomes and $\sum_xP_x=\Id_H$. The probability of measuring outcome $x$ given a state $\rho\in\mc{D}(\tsf{H})$ is given by Born's rule as $\Tr(P_x\rho)$. Again, any measurement may be purified, by adding auxiliary registers, to a projector-valued measurement (PVM), which is a measurement where the $P_x$ are orthogonal projectors $P_xP_y=\delta_{x,y}P_x$ (\cref{lem:pure-pvm}). This then corresponds to a measurement in some basis followed by a deterministic classical function.
	
	Given two registers $H$ and $K$, the corresponding joint quantum system is given by the tensor product Hilbert space $\tsf{H}\otimes\tsf{K}$. If necessary, we (perhaps inconsistently) add the name of the register as a subscript onto a state/operator to distinguish which register it acts on/belongs to. A state $\ket{\psi}\in\tsf{H}\otimes\tsf{K}$ is \emph{separable} if it can be written as a pure tensor of the form $\ket{\psi}=\ket{v}_H\otimes\ket{w}_K$. Otherwise, the state is \emph{entangled}. Nevertheless, any pure state $\ket{\psi}\in\tsf{H}\otimes\tsf{K}$ admits a \emph{Schmidt decomposition} $\ket{\psi}=\sum_i\sqrt{p_i}\ket{i}_H\otimes\ket{i}_K$ where the $p_i\geq 0$, and $\set{\ket{i}_H}$ and $\set{\ket{i}_K}$ are sets of orthonormal vectors. An operation is \emph{local} if it acts as a pure tensor. The \emph{partial trace} of a register $H$ is the linear map $\Tr_H:\mc{L}(\tsf{H}\otimes\tsf{K})\rightarrow\mc{L}(\tsf{K})$ defined on pure tensors as $\Tr_H(A\otimes B)=\Tr(A)B$ (and extended linearly), corresponding to making a measurement on the space $\tsf{H}$ and then forgetting the result. For a state $\rho_{HK}\in\mc{D}(\tsf{H}\otimes\tsf{K})$, we write the state on $\tsf{H}$ as $\rho_{H}=\Tr_K(\rho_{HK})$.
	
	The Euclidean norm $\norm{\ket{v}}=\sqrt{\!\braket{v}}$ gives the appropriate distance metric between pure states. For mixed states, we use the \emph{trace distance}
	\begin{align}
		d_{\Tr}(\rho,\sigma)=\norm{\rho-\sigma}_{\Tr}=\frac{1}{2}\Tr\abs{\rho-\sigma},
	\end{align}
	where the absolute value of an operator is $\abs{L}=\sqrt{L^\dag L}$. For other operators, we use the \emph{operator norm}
	\begin{align}
		\norm{L}=\sup\set*{\norm{L\ket{v}}}{\vphantom{\big|}\braket{v}=1}.
	\end{align}
	Important properties of these norms are given in \cref{sec:a-prelimlems}.
	
	An important system is the \emph{bit} $Q=\Z_2=\{0,1\}$, and its corresponding Hilbert space, the \emph{qubit} $\tsf{Q}\cong\C^2$. The basis $\set{\ket{0},\ket{1}}$ is the \emph{computational basis} and the basis $\set*{\ket{+},\ket{-}}$ where $\ket{+}=\frac{1}{\sqrt{2}}(\ket{0}+\ket{1})$ and $\ket{-}=\frac{1}{\sqrt{2}}(\ket{0}-\ket{1})$ is the \emph{Hadamard basis}. The \emph{Hadamard operator} is the Hermitian unitary $H:\tsf{Q}\rightarrow\tsf{Q}$ that maps the computational basis to the Hadamard basis, which is expressed in either basis as
	\begin{align}
	H=\frac{1}{\sqrt{2}}\begin{bmatrix}1&1\\1&-1\end{bmatrix}.
	\end{align}
	In the computational basis, the \emph{Pauli operators} are
	\begin{align}
		Z=\begin{bmatrix}1&0\\0&-1\end{bmatrix}\qquad X=\begin{bmatrix}0&1\\1&0\end{bmatrix}\qquad Y=\begin{bmatrix}0&-i\\i&0\end{bmatrix}.
	\end{align}
	It is direct to check that $Z$ and $X$ anticommute, and that the Hadamard diagonalises $X$, so that $X=HZH$. We define the \emph{Breidbart} operator $\be:\tsf{Q}\rightarrow\tsf{Q}$ as the Hermitian unitary that diagonalises $H$, so $H=\be Z\be$ and in the computational basis
	\begin{align}
		\be=\begin{bmatrix}\cos\frac{\pi}{8}&\sin\frac{\pi}{8}\\\sin\frac{\pi}{8}&-\cos\frac{\pi}{8}\end{bmatrix};
	\end{align}
	and the \emph{Breidbart state} $\ket{\beta}=\be\ket{0}$. Important relations that follow from the definition are  $H\ket{\beta}=\ket{\beta}$, $Z\ket{\beta}=\be\ket{+}$, $X\ket{\beta}=\be\ket{-}$, and $ZX\ket{\beta}=\be\ket{1}$. Finally, we define the \emph{conjugate-coding/Wiesner/BB84} states on $n\in\N$ qubits for $x,\theta\in\Z_2^n$ as $\ket{x^\theta}=H^\theta\ket{x}=H^{\theta_1}\ket{x_1}\otimes\cdots\otimes H^{\theta_n}\ket{x_n}\in\tsf{Q}^{\otimes n}$; and we also call the states $\be^{\otimes n}\ket{x^\theta}$ the \emph{Wiesner-Breidbart states}.
	
	More details are given in any of the many good resources for quantum information, such as \cite{NC00,Wat18}.	

\subsection{Probability}\label{sec:probability}

	Any probability distribution on a finite set $X$ may be represented by a function $\pi:X\rightarrow[0,1]$ such that $\sum_{x\in X}\pi(x)=1$. Then, the probability of an event $S\subseteq X$ is $\Pr(S)=\sum_{x\in S}\pi(x)$. For any function $f:X\rightarrow\tsf{V}$, where $\tsf{V}$ is a $\C$-vector space, we write the expectation value with respect to this distribution as
	\begin{align}
		\expec{x\leftarrow \pi}f(x)=\sum_{x\in X}\pi(x)f(x).
	\end{align}
	We distinguish the uniform probability distribution $\mbb{u}:X\rightarrow[0,1]$, which is $\mbb{u}(x)=\frac{1}{|X|}$; and we write $\mbb{E}_{x\in X}$ to mean $\mbb{E}_{x\leftarrow \mbb{u}}$. Any probability distribution $\pi$ on $X$ can be represented as a classical state $\mu_\pi=\mbb{E}_{x\leftarrow\pi}\ketbra{x}\in\mc{D}(\tsf{X})$, where we write in particular the maximally mixed state as the classical state of the uniform distribution $\mu_X:=\mu_{\mbb{u}}$. For a random variable $\Gamma$, we use the same notattion $\expec{x\leftarrow\Gamma}$ to denote $x$ sampled from the image of $\Gamma$ with repsect to its distribution. If $\Gamma$ has image in a vector space, we write its expectation as $\expec{x\leftarrow\Gamma}x=\mathbb{E}\Gamma$.
	
	An important bound we make use of is Hoeffding's inequality. Let $\Gamma_1,\ldots,\Gamma_n$ be independent random variables with image in $[0,1]$, and write their sum $\Gamma=\Gamma_1+\ldots+\Gamma_n$. The inequality states that for any $t\geq 0$,
	\begin{align}
		\Pr\parens*{\Gamma-\mathbb{E}\Gamma\geq t}\leq e^{-2\frac{t^2}{n}}
	\end{align}

\subsection{Exact and Approximate Representation Theory}\label{sec:representation}

	Throughout this section, let $G$ be a finite group.
	
	A \emph{representation} of $G$ over $\C$ is a group homomorphism $\gamma:G\rightarrow\mc{GL}(\tsf{V})$, where $\tsf{V}$ is a finite-dimensional $\C$-vector space. Two representations $\gamma_i:G\rightarrow\mc{GL}(\tsf{V}_i)$ for $i=0,1$ are \emph{isomorphic} if there exists an invertible linear map $U:\tsf{V}_0\rightarrow\tsf{V}_1$ such that $U\gamma_0(g)=\gamma_1(g)U$ $\forall\,g\in G$ (intertwining operator). Every representation is isomorphic to a unitary representation, i.e. where $\gamma(G)\subseteq\mc{U}(\tsf{V})$. A representation is \emph{irreducible} if the only subspaces invariant under the action of $G$ are $\tsf{V}$ and $0$. By Maschke's theorem, every representation of $G$ decomposes as a direct sum of irreducible representations. Let $\tsf{Irr}(G)$ be a set of representatives for the isomorphism classes of the irreducible representations; $\tsf{Irr}(G)$ has finitely many elements and the sum $\sum_{\gamma\in\tsf{Irr}(G)}d_\gamma\Tr(\gamma(g))=|G|\delta_{g,1}$, where $d_\gamma$ is the dimension of the representation $\gamma$. The only irreducible representations of an Abelian group are $1$-dimensional. The important example we see in this paper is $\Z_2^n$ under addition. The irreducible representations are indexed by the elements $s\in\Z_2^n$, and they take the form $\gamma_s(x)=(-1)^{x\cdot s}$, where $x\cdot s=x_1s_1+x_2s_2+\ldots+s_nx_n$.
	
	This concludes our whirlwind summary of some of the representation theory of finite groups; complete explanations are available from many perspectives, such as \cite{Wei03,Ser77}. We now go into rather more detail about the theory of approximate representations, which hinges on a result of Gowers and Hatami \cite{GH17}.
	
	\begin{definition}
		Let $\tsf{V}$ and $\tsf{W}$ be Hilbert spaces and let $\ket{\psi}\in\tsf{V}\otimes\tsf{W}$. For $\varepsilon\geq 0$, an \emph{$(\varepsilon,\ket{\psi})$-representation} of $G$ is a map $f:G\rightarrow\mc{U}(\tsf{V})$ such that, for every $y\in G$,
		\begin{align}
		\expec{x\in G}\norm{\parens*{f(x)f(y)-f(xy)}\ket{\psi}}^2\leq\varepsilon^2.
		\end{align}
	\end{definition}

	The following theorem characterises how close an approximate representation is to a true representation.
	
	\begin{theorem}[Gowers-Hatami]\label{thm:gowers-hatami}
		Let $f:G\rightarrow\mc{U}(\tsf{V})$ be a $(\varepsilon,\ket{\psi})$-representation. Then, there exists a Hilbert space $\tsf{V}'$, an isometry $V:\tsf{V}\rightarrow\tsf{V}'$, and a representation $g:G\rightarrow\mc{U}(\tsf{V}')$ such that, for any $x\in G$,
		\begin{align}
		\norm{\parens*{Vf(x)-g(x) V}\ket{\psi}}\leq\varepsilon.
		\end{align}
	\end{theorem}

	Note that the above definition and theorem have a slightly different form from how they were presented in previous work \vphantom{\cite{NV17}}\cite{Vid18, CMMN20}.
	
	The proof given here is almost identical to the proof of \cite{Vid18}, which uses the notion of the Fourier transform of a function acting on a group. Given a function $f:G\rightarrow\tsf{V}$ for $\tsf{V}$ a $\C$-vector space, the Fourier transform is the map acting on $\tsf{Irr}(G)$ defined as
	$$\hat{f}(\gamma)=\frac{1}{|G|}\sum_{g\in G}f(g)\otimes\gamma(g).$$
	It is straightforward to check that the inverse transform is
	$$f(g)=\sum_{\gamma\in\tsf{Irr}(G)} d_\gamma\Tr_{V_\gamma}\parens*{(\Id_V\otimes\gamma(g^{-1}))\hat{f}(\gamma)}.$$
	
	\begin{proof}[Proof of \cref{thm:gowers-hatami}]
		First, we construct the necessary objects. The dilated space is
		\begin{align}
		\tsf{V}'=\bigoplus_{\gamma\in\tsf{Irr}(G)}\tsf{V}\otimes\tsf{V}_\gamma\otimes\tsf{V}_\gamma,
		\end{align}
		and the representation is taken to be
		\begin{align}
		g(x)=\bigoplus_{\gamma\in\tsf{Irr}(G)}\Id\otimes\Id_\gamma\otimes\overline{\gamma(x)},
		\end{align}
		where the complex conjugate on $\tsf{V}_\gamma$ is taken with respect to a fixed basis $\ket{i}_\gamma$ for $i=1,...,d_\gamma$. Then, we take the isometry to be
		\begin{align}
		V\ket{v}=\bigoplus_{\gamma\in\tsf{Irr}(G)} \sqrt{d_\gamma}\sum_{i=1}^{d_\gamma}\hat{f}(\gamma)(\ket{v}\otimes\ket{i}_\gamma)\otimes\ket{i}_\gamma.
		\end{align}
		This is in fact an isometry as
		\begin{align}
		\begin{split}
		V^\dag V&=\sum_{\gamma\in\tsf{Irr}(G)} d_\gamma\sum_{i=1}^{d_\gamma}(\Id\otimes\bra{i}_\gamma)\hat{f}(\gamma)^\dag\otimes\bra{i}_\gamma\sum_{i=1}^{d_\gamma}\hat{f}(\gamma)(\Id\otimes\ket{i}_\gamma)\otimes\ket{i}_\gamma\\
		&=\sum_{\gamma\in\tsf{Irr}(G)} d_\gamma\sum_{i=1}^{d_\gamma}(\Id\otimes\bra{i}_\gamma)\hat{f}(\gamma)^\dag\hat{f}(\gamma)(\Id\otimes\ket{i}_\gamma)\\
		&=\expec{x,y\in G}\sum_{\gamma\in\tsf{Irr}(G)} d_\gamma\sum_{i=1}^{d_\gamma}f(x)^\dag f(y)\braket{i}{\gamma(x)^\dag\gamma(y)}{i}_\gamma\\
		&=\expec{x,y\in G}f(x)^\dag f(y)\sum_{\gamma\in\tsf{Irr}(G)} d_\gamma\Tr\parens{\gamma(x^{-1}y)}\\
		&=\expec{x,y\in G}f(x)^\dag f(y)|G|\delta_{x,y}=\expec{x\in G}f(x)^\dag f(x)\\
		&=\Id.
		\end{split}
		\end{align}
		Thus, with a similarly long equation we may simplify
		\begin{align}
		\begin{split}
		V^\dag g(x) V&=\sum_{\gamma\in\tsf{Irr}(G)} d_\gamma\sum_{i,j=1}^{d_\gamma}(\Id\otimes\bra{i}_\gamma)\hat{f}(\gamma)^\dag\hat{f}(\gamma)(\Id\otimes\ket{j}_\gamma)\braket{i}{\overline{\gamma(x)}}{j}_\gamma\\
		&=\expec{y,z\in G}f(y)^\dag f(z)\sum_{\gamma\in\tsf{Irr}(G)} d_\gamma\sum_{i,j=1}^{d_\gamma}\braket{i}{\gamma(y)^\dag\gamma(z)}{j}\!\!\!\braket{j}{\gamma(x)^\dag}{i}_\gamma\\
		&=\expec{y,z\in G}f(y)^\dag f(z)\sum_{\gamma\in\tsf{Irr}(G)} d_\gamma\Tr\parens*{\gamma(y^{-1}zx^{-1})}\\
		&=\expec{y,z\in G}f(y)^\dag f(z)|G|\delta_{z,yx}\\
		&=\expec{y\in G}f(y)^\dag f(yx).
		\end{split}
		\end{align}
		Noting that $\norm{\parens*{f(x)f(y)-f(xy)}\ket{\psi}}^2=2\braket{\psi}-2\latRe\braket{\psi}{f(y)^\dag f(x)^\dag f(xy)}{\psi}$, we can use the above and the hypothesis to get
		\begin{align}
		\begin{split}
		\norm{(Vf(x)-g(x)V)\ket{\psi}}^2&=2\braket{\psi}-2\latRe\braket{\psi}{f(x)^\dag V^\dag g(x) V}{\psi}\\
		&=\expec{y\in G}\parens*{2\braket{\psi}-2\latRe\braket{\psi}{f(x)^\dag f(y)^\dag f(yx)}{\psi}}\\
		&=\expec{y\in G}\norm{\parens*{f(y)f(x)-f(yx)}\ket{\psi}}^2\leq\varepsilon^2.
		\end{split}
		\end{align}
		which completes the proof.
	\end{proof}

	Later, we naturally come across approximate representations of $\Z_2^n$. These representations are induced by approximate commutation relations of the generators. To show they are in fact approximate representations, we need to relate approximate commutation of the generators to approximate commutation of all the elements. First, we tackle the case that needs no extra assumptions, $G=\Z_2^2$.
	
	\begin{lemma}\label{lem:z22-gh}
		Let $\tsf{V}$ and $\tsf{W}$ be Hilbert spaces, let $\ket{\psi}\in\tsf{V}\otimes\tsf{W}$, and let $U_0,U_1\in\mc{U}(\tsf{V})$ be self-inverse such that
		\begin{align}
			\norm{[U_0,U_1]\ket{\psi}}\leq\delta,		
		\end{align}
		for some $\delta\geq 0$. Then, the function $f:\Z_2^2\rightarrow\mc{U}(\tsf{V})$ defined by $f(00)=\Id$, $f(01)=U_0$, $f(10)=U_1$, and $f(11)=U_0U_1$ is an $\parens{\tfrac{\delta}{\sqrt{2}},\ket{\psi}}$-representation of $\Z_2^2$.
	\end{lemma}
	
	\begin{proof}
		This is straightforward to check using the hypothesis and the fact that the action by a unitary does not change the Euclidean norm. For $y=00$,
		\begin{align}
			\expec{x\in G}\norm{\parens*{f(x)f(00)-f(x+00)}\ket{\psi}}^2=\expec{x\in G}\norm{\parens*{f(x)-f(x)}\ket{\psi}}^2=0\leq\frac{\delta^2}{2}.
		\end{align}
		For $y=01$,
		\begin{align}
		\begin{split}
		\expec{x\in G}\norm{\parens*{f(x)f(01)-f(x+01)}\ket{\psi}}^2&=\frac{1}{4}\big(\norm{(U_0-U_0)\ket{\psi}}^2+\norm{(U_0^2-\Id)\ket{\psi}}^2\\
		&+\norm{(U_1U_0-U_0U_1)\ket{\psi}}^2+\norm{(U_0U_1U_0-U_1)\ket{\psi}}^2\big)\\
		&\leq\frac{\delta^2}{2}.
		\end{split}
		\end{align}
		For $y=10$,
		\begin{align}
		\begin{split}
		\expec{x\in G}\norm{\parens*{f(x)f(10)-f(x+10)}\ket{\psi}}^2&=\frac{1}{4}\big(\norm{(U_1-U_1)\ket{\psi}}^2+\norm{(U_0U_1-U_0U_1)\ket{\psi}}^2\\
		&+\norm{(U_1^2-\Id)\ket{\psi}}^2+\norm{(U_0U_1U_1-U_0)\ket{\psi}}^2\big)\\
		&=0\leq\frac{\delta^2}{2}.
		\end{split}
		\end{align}
		And finally, for $y=11$,
		\begin{align}
		\begin{split}
		\expec{x\in G}\norm{\parens*{f(x)f(11)-f(x+11)}\ket{\psi}}^2&=\frac{1}{4}\big(\norm{(U_0U_1-U_0U_1)\ket{\psi}}^2+\norm{(U_0U_0U_1-U_1)\ket{\psi}}^2\\
		&+\norm{(U_1U_0U_1-U_0)\ket{\psi}}^2+\norm{(U_0U_1U_0U_1-\Id)\ket{\psi}}^2\big)\\
		&\leq\frac{\delta^2}{2}.
		\end{split}
		\end{align}
	\end{proof}
	
	Extending a result of this form to $\Z_2^n$ for $n>2$ requires another condition on the unitaries, in order to be able to use the commutation with respect to $\ket{\psi}$ even when there are operators sitting between the state and the unitaries. To do this, we impose an additional relation, arising from our sum-of-squares decomposition, which allows to swap operators onto another register while incurring only a small error.
	
	\begin{lemma}\label{lem:z2n-gh}
		Let $\tsf{V},\tsf{W}$ be Hilbert spaces, let $\ket{\psi}\in\tsf{V}\otimes\tsf{W}$, and let $U_1,...,U_n,V_1,...,V_n\in\mc{U}(\tsf{V})$ be a collection of self-inverse unitaries such that
		\begin{align}
		&\norm{[U_i,U_j]\ket{\psi}}\leq\delta\label{eq:z2n-comm}\\
		&\norm{U_i\ket{\psi}-V_i\ket{\psi}}\leq\epsilon\label{eq:z2n-type}\\
		&[U_i,V_j]=0\label{eq:z2n-reg}
		\end{align}
		for some $\delta,\epsilon\geq 0$. Then, the map
		\begin{align}
			\begin{matrix}f:&\Z_2^n&\rightarrow&\mc{U}(\tsf{V})\\&x&\mapsto&U^x,\end{matrix}
		\end{align}
		where $U^x:=U_1^{x_1}\cdots U_n^{x_n}$, is an $(n^2(3\epsilon+\delta),\ket{\psi})$-representation of $\Z_2^n$.
	\end{lemma}

	\begin{proof}
		Let $x,y\in\Z_2^n$. Then $f(x)f(y)-f(x+y)=U^xU^y-U^{x+y}$. Write $V^x=V_n^{x_n}\cdots V_1^{x_1}$ Suppose the first nonzero term of $y$ is at position $i_0$. Write $x^1=x_1\ldots x_{i_0-1}0\ldots0$ and $x^2=0\ldots0x_{i_0+1}\ldots x_n$ and similarly for $y$. By hypothesis, this gives via \cref{eq:z2n-type}
		\begin{align}
		\begin{split}
		\norm{(U^xU^y-U^{x+y})\ket*{\psi}}&=\norm{(U^{x^1}U_{i_0}^{x_{i_0}}U^{x^2}U_{i_0}U^{y^2}-U^{x+y})\ket*{\psi}}\\
		&\leq \norm{(U^{x^1}U_{i_0}^{x_{i_0}}U^{x^2}U_{i_0} V^{y^2}-U^{x+y})\ket*{\psi}}+|y^2|\epsilon.
		\end{split}
		\end{align}
		Now, we can shift $U_{i_0}$ up through $U^{x^2}$ by using the commutation relations \cref{eq:z2n-comm,eq:z2n-reg} and then replacing that term of $U^{x^2}$ with the corresponding $V$ term, and continuing recursively. This adds an error
		\begin{align}
		\norm{(U^{x^1}U_{i_0}^{x_{i_0}}U^{x^2}U_{i_0} V^{y^2}-U^{x+y})\ket*{\psi}}\leq \norm{(U^{x^1}U_{i_0}^{x_{i_0}+y_{i_0}}V^{y^2}V^{x^2}-U^{x+y})\ket*{\psi}}+|x^2|(\epsilon+\delta).
		\end{align}
		We can then shift $V^{x^2}$ and the first term, $i_1$, of $y^2$ back:
		\begin{align}
		\begin{split}
		\norm{(U^{x^1}U_{i_0}^{x_{i_0}}U^{x^2}U_{i_0} V^{y^2}-U^{x+y})\ket*{\psi}}&\leq \norm{(U^{x^1}U_{i_0}^{x_{i_0}+y_{i_0}}U^{x^2}U^{y_{i_1}}V^{{y^2}^2}-U^{x+y})\ket*{\psi}}+|x^2|(2\epsilon+\delta)+\epsilon\\
		&\leq \norm{(U^{x^1}U_{i_0}^{x_{i_0}+y_{i_0}}U^{x^2}U^{y_{i_1}}V^{{y^2}^2}-U^{x+y})\ket*{\psi}}+n(2\epsilon+\delta)
		\end{split}
		\end{align}
		Note that the above estimate is relatively crude. This process can be repeated another $|y^2|$ times to get
		\begin{align}
		\norm{(U^xU^y-U^{x+y})\ket*{\psi}}\leq n|y|(2\epsilon+\delta)+|y^2|\epsilon\leq n^2(2\epsilon+\delta)+n\epsilon,
		\end{align}
		which gives the result.
	\end{proof}

\section{Monogamy-of-Entanglement Games} \label{sec:monogamy}

In this section, we formally introduce the concept of a monogamy-of-entanglement game. In \cref{sec:moeg-defs}, we define monogamy-of-entanglement games and how to play them, and introduce the game from \cite{TFKW13} we study in this paper. In \cref{sec:bias}, we introduce a different way to look at winning a game, and use this to get an algebraic approach (sum-of-squares decomposition of the game polynomial) to upper bounding the winning probability. This method is adapted from what has been used before for nonlocal games \cite{BaP15}.

\begin{figure}[h!]
	\centering
	\begin{tikzpicture}
	\draw[rounded corners] (0,-3) rectangle (2,3);
	\draw (3.5,-0.75) rectangle (5.5,0.75);
	\draw (3.5,1.5) rectangle (5.5,3);
	\draw (3.5,-1.5) rectangle (5.5,-3);
	\draw[dashed] (3.25,-1.15) -- (5.75,-1.15);
	\draw (2,0) -- (3.5,0);
	\draw (2,2.25) -- (3.5,2.25);
	\draw (2,-2.25) -- (3.5,-2.25);
	\draw[gray] (2.25,3) rectangle (3.25,4);
	\draw[gray] (2.75,3) -- (2.75,-1.875) -- (3.5,-1.875);
	\draw[gray] (2.75,0.375) -- (3.5,0.375);
	\draw[gray] (2.75,2.675) -- (3.5,2.675);
	\draw[gray, -Latex] (5.5,0) -- (6.5,0) node[right]{\color{black}$y_B$};
	\draw[gray, -Latex] (5.5,2.25) -- (6.5,2.25) node[right]{\color{black}$y$};
	\draw[gray, -Latex] (5.5,-2.25) -- (6.5,-2.25) node[right]{\color{black}$y_C$};
	
	\node at (1,0){$\rho$};
	\node at (4.5,0){Bob};
	\node at (4.5,2.25){Alice};
	\node at (4.5,-2.25){Charlie};
	\node at (2.75,3.5){$\Theta$};
	\end{tikzpicture}
	\caption{Scenario of a monogamy-of-entanglement game. Note that Alice's measurements, though not included in the diagram, are fixed by the description of the game.}
	\label{fig:moe-game}
\end{figure}
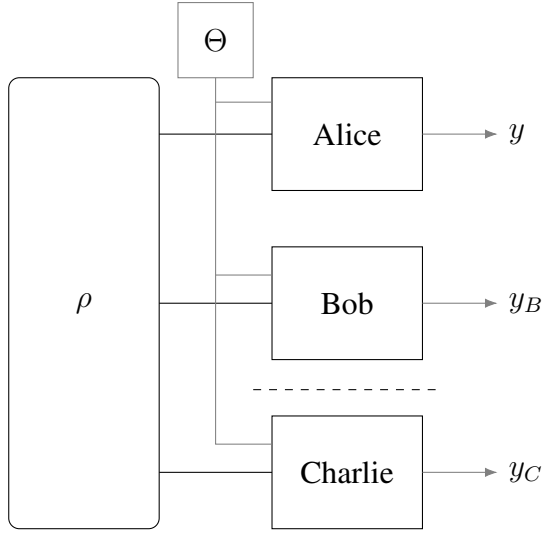

\subsection{Definitions}\label{sec:moeg-defs}

Informally, a monogamy-of-entanglement (MoE) game is a game played by three quantum parties: a trusted referee, Alice, against two collaborating adversaries, Bob and Charlie, who may agree on a strategy but do not communicate while the game is in play. Such a game is played as follows:
\begin{enumerate}[(1)]
	\item The adversaries prepare a quantum state $\rho_{ABC}$ shared between the three players. After this, they may no longer communicate.
	
	\item Alice chooses a measurement to make on her space and provides Bob and Charlie the information about what measurement she chose.
	
	\item Alice measures, and Bob and Charlie both try to guess her measurement outcome using their parts of the state.
	
	\item The adversaries win the game if they simultaneously guessed Alice's outcome correctly.
\end{enumerate}

The setup for a generic MoE game is given in \cref{fig:moe-game}. Note that if there were only one adversary, they would always be able to guess Alice's measurement (as long as it is projective) by sharing a maximally entangled state. However, this is not in general true for MoE games because there is no maximal tripartite entanglement. We can define such a game more formally as follows.


\begin{definition}
	An \emph{monogamy-of-entanglement (MoE) game} is a tuple $\ttt{G}=\parens*{\Theta,Y,\tsf{A},\pi,A}$, where
	\begin{itemize}
		\item $\Theta$ is a finite set representing the possible questions;
		\item $Y$ is a finite set representing the possible answers;
		\item $\tsf{A}$ is the complex Hilbert space that Alice holds;
		\item $\pi:\Theta\rightarrow[0,1]$ is a function representing the probability that Alice chooses each question;
		\item and $A$ is a positive operator-valued function
		\begin{align}
		\begin{matrix}A:&\Theta\times Y&\rightarrow&\mc{P}(\tsf{A})\\&(\theta,y)&\mapsto&A^{\theta}_{y}\end{matrix}
		\end{align}
		such that, for each $\theta$, $A^\theta:Y\rightarrow\mc{P}(\tsf{A})$ is a POVM.
	\end{itemize}
\end{definition}

The MoE game we study here is the original game of this kind introduced in \cite{TFKW13}, where Alice measures a single qubit in a conjugate-coding basis.

\begin{definition}
	The \emph{TFKW game} is the MoE game $\ttt{TFKW}=\parens*{\Z_2,\Z_2,\tsf{Q},\mbb{u},A}$, where $\mbb{u}(\theta)=\frac{1}{2}$ is the uniform distribution and $A^\theta_y=\ketbra{y^\theta}$.
\end{definition}

The strategies Bob and Charlie may use are constrained only by the laws of quantum mechanics. There are other classes of strategies based on other resource theories \cite{JMRW16} that are not studied here.


\begin{definition}
	A quantum \emph{strategy} for an MoE game $\ttt{G}=\parens*{\Theta,Y,\tsf{A},\pi,A}$ is a tuple $\ttt{S}=\parens*{\tsf{B},\tsf{C},B,C,\rho}$, where
	\begin{itemize}
		\item $\tsf{B}$ and $\tsf{C}$ are the complex Hilbert spaces that Bob and Charlie hold, respectively;
		\item $B$ and $C$ are Bob and Charlie's quantum measurements, so positive operator-valued functions
		\begin{align}
		\begin{matrix}B:&\Theta\times Y&\rightarrow&\mc{P}(\tsf{B})\\&(\theta,y)&\mapsto&B^{\theta}_{y}\end{matrix}\qquad\qquad\qquad\begin{matrix}C:&\Theta\times Y&\rightarrow&\mc{P}(\tsf{C})\\&(\theta,y)&\mapsto&C^{\theta}_{y}\end{matrix},
		\end{align}
		such that $B^{\theta}$ and $C^{\theta}$ are POVMs.
		\item and $\rho\in\mc{D}(\tsf{A}\otimes\tsf{B}\otimes\tsf{C})$ is a shared quantum state.
	\end{itemize}
\end{definition}

\begin{definition}
	The \emph{winning probability} of a strategy $\ttt{S}$ for a game $\ttt{G}$ is
	\begin{align}
		\mfk{w}_{\ttt{G}}(\ttt{S})&=\expec{\theta\leftarrow\pi}\sum_{y\in Y}\Tr\squ*{\parens*{A^{\theta}_{y}\otimes B^{\theta}_{y}\otimes C^{\theta}_{y}}\rho}.
	\end{align}
	The optimal winning probability of the game is the supremum over strategies $\mfk{w}_\ttt{G}=\sup_{\ttt{S}}\mfk{w}_\ttt{G}(\ttt{S}).$
\end{definition}

Note that there may not necessarily be a strategy that wins with probability $\mfk{w}_{\ttt{G}}$ if the set of winning probabilities is not closed.

In \cite{TFKW13}, the winning probability of the TFKW game was found to be about $0.85$.

\begin{theorem}[\cite{TFKW13}]
	\begin{align}
		\mfk{w}_{\ttt{TFKW}}=\cos^2\parens*{\frac{\pi}{8}}=\frac{1}{2}+\frac{1}{2\sqrt{2}}.
	\end{align}
\end{theorem}

The canonical strategy for this game is unentangled, i.e. Bob and Charlie share no entanglement: they simply provide Alice with a copy of the Breidbart state $\ket{\beta}$ and always guess measurement outcome $0$. Note that there are optimal strategies using any of the single-qubit Wiesner-Breidbart states due to the symmetries of Alice's measurement operators. The behaviour of these strategies is given in \cref{tab:strats}.

\begin{table}[h!]
	\centering
	\begin{tabular}{c||c|c}
		State & $\theta=0$ & $\theta=1$\\\hline\hline
		$\ket{\beta}$ & $0$ & $0$\\\hline 
			$Z\ket{\beta}$ & $0$ & $1$ \\\hline 
			$X\ket{\beta}$ & $1$ & $0$ \\\hline 
			$XZ\ket{\beta}$ & $1$ & $1$
	\end{tabular}
	\caption{Answers $y$ for the unentangled optimal strategies of the TFKW game. Bob and Charlie reply with the same answer, depending on both $\theta$ and the Wiesner-Breidbart state they chose.}
	\label{tab:strats}
\end{table}

A result of \cite{TFKW13} gives that a strategy for an MoE game may be assumed to be pure, i.e. the shared state is pure and Bob and Charlie's measurements are projective.

\begin{theorem}(\cite{TFKW13})\label{lem:purify}
	A strategy $\ttt{S}=\parens*{\tsf{B},\tsf{C},B,C,\rho}$ for an MoE game $\ttt{G}$ may be purified to a pure strategy $\tilde{\ttt{S}}=\parens*{\tilde{\tsf{B}},\tilde{\tsf{C}},\tilde{B},\tilde{C},\tilde{\rho}}$, where $\tilde{B}$ and $\tilde{C}$ are projective and $\tilde{\rho}=\ketbra{\psi}$ is pure, that wins with the same probability.
\end{theorem}

\begin{proof}
	Using \cref{lem:pure-pvm}, all the POVMs $B^\theta$, $C^\theta$ can be purified. This maps the state isometrically and locally to $\rho\otimes\ketbra{\mathrm{aux}}$ for some auxiliary state in the extension of Bob and Charlie's spaces. Finally, we can purify the state by appending another auxiliary register using \cref{lem:pure-state}.
\end{proof}

We called a strategy \emph{purified} if it is pure as in the above lemma, but there additionally exists an auxiliary register $R$ to which none of the players have access, such that $\ket{\psi}\in\tsf{A}\otimes\tsf{B}\otimes\tsf{C}\otimes\tsf{R}$. In this way, we may reach the state of any general strategy, up to local isometry, simply by tracing out this register, which does not affect the gameplay.

One way to construct new MoE games is using parallel repetition. Given an MoE game $\ttt{G}$, a parallel repetition is the game where $\ttt{G}$ is played some fixed number of times $n$ simultaneously. To win the parallel repetition, the adversaries must win all $n$ copies of $\ttt{G}$.

\begin{definition}
	Let $\ttt{G}=\parens*{\Theta,Y,\tsf{A},\pi,A}$ be an MoE game and let $n\in\N$. The \emph{$n$-fold parallel repetition} of $\ttt{G}$ is the MoE game $\ttt{G}^n=\parens*{\Theta^n,Y^n,\tsf{A}^{\otimes n}, \pi^n, A^n}$ where $\pi^n(\theta_1,\ldots,\theta_n)=\pi(\theta_1)\cdots\pi(\theta_n)$ and $(A^n)^{(\theta_1,\ldots,\theta_n)}_{(y_1,\ldots,y_n)}=A^{\theta_1}_{y_1}\otimes\cdots\otimes A^{\theta_n}_{y_1}$.
\end{definition}

For convenience, we write in general $\tsf{A}^{\otimes n}=\tsf{A}_1\otimes\cdots\otimes\tsf{A}_n$ where $\tsf{A}_i=\tsf{A}$, in order to be able to distinguish terms in different positions. The major result of \cite{TFKW13} is that they show that the adversaries cannot do better at the parallel-repeated TFKW game than by just playing a separate optimal strategy of the single game on each copy. This leads to an exponentially-decreasing bound on the winning probability.

\begin{theorem}[\cite{TFKW13}]
	\begin{align}
		\mfk{w}_{\ttt{TFKW}^n}=\parens*{\cos^2\parens*{\frac{\pi}{8}}}^n.
	\end{align}
\end{theorem}

We make use of a different notion of winning probability for parallel repeated games. Instead of considering the probability of winning all the games at the same time, we consider the probabilities of winning each of the games using the same strategy.

\begin{definition}
	Let $\ttt{G}=\parens*{\Theta,Y,\tsf{A},\pi,A}$ be an MoE game, let $i\in[n]$, and let $\ttt{S}=\parens*{\tsf{B},\tsf{C},B,C,\rho}$ be a strategy for $\ttt{G}^n$. Then, the \emph{$i$-th winning probability} of $\ttt{G}^n$ is
	\begin{align}
		\mfk{w}^i_{\ttt{G}^n}(\ttt{S})=\expec{\theta\leftarrow\pi^n}\sum_{y\in Y}\Tr\squ*{\parens*{A^\theta_{y,i}\otimes B^\theta_{y,i}\otimes C^\theta_{y,i}}\rho},
	\end{align}
	where $A^\theta_{y,i}=\sum_{\substack{x\in Y^n\\x_i=y}}(A^n)^\theta_x$ and $B^\theta_{y,i}=\sum_{\substack{x\in Y^n\\x_i=y}}B^\theta_x$ with analogous definition for $C^\theta_{y,i}$.
\end{definition}

Due to the tensor product structure of $A^n$, $A^{\theta}_{y,i}$ depends only on the $i$-th element of $\theta$. Explicitly,
\begin{align}
	A^\theta_{y,i}=\Id_1\otimes\cdots\otimes\Id_{i-1}\otimes A^{\theta_i}_{y}\otimes\Id_{i+1}\otimes\cdots\otimes\Id_n.
\end{align}

The operators $B^\theta_{y,i}$ and $C^\theta_{y,i}$ depend in general on all the elements of $\theta$. Nevertheless, some important properties of the $A^{\theta}_{y,i}$ still hold: if the adversaries' measurements are projective, the operators commute for the same value of $\theta$, i.e.
\begin{align}\label{eq:commutation}
	[B^\theta_{y,i},B^\theta_{y',j}]=0,
\end{align}
and satisfy the product relation $B^\theta_y=B^\theta_{y_1,1}B^\theta_{y_2,2}\cdots B^\theta_{y_n,n}$; these hold identically for the $C^\theta_{y,i}$. The commutation and the product relation follow directly from the definition.

\subsection{Observables, Bias, and Positivity}\label{sec:bias}

In this section, we assume that we are working with an MoE game that has only two answers, in which case we may identify $Y$ with $\Z_2$, so $\ttt{G}=\parens*{\Theta,\Z_2,\tsf{A},\pi,A}$.

Similarly to what is often done for nonlocal games \cite{CMMN20}, we transform the expression for the winning probability into an expression in terms of observables rather than measurements.

\begin{definition}
	Let $\tsf{H}$ be a Hilbert space and let $P:\Z_2\rightarrow\mc{P}(\tsf{H})$ be a POVM. Then, the \emph{observable} of this POVM is $\overline{P}=P_0-P_1$.
\end{definition}
The observable completely characterises the measurement as $P_y=\frac{1}{2}\parens*{\Id+(-1)^y\overline{P}}$; and $\overline{P}$ is unitary if and only if $P$ is projective. For the measurements of an MoE game, we write the observables $A_\theta=\overline{A^\theta}$ for simplicity, and similarly for the adversaries' observables. It is a direct calculation to express the winning probability in terms of the observables:
\begin{align}
\begin{split}
	\mathfrak{w}_{\ttt{G}}(\ttt{S})&=\frac{1}{8}\expec{\theta\leftarrow\pi}\sum_{y\in\Z_2}\Tr\squ*{\parens*{(\Id_A+(-1)^yA_\theta)\otimes (\Id_B+(-1)^yB_\theta)\otimes (\Id_C+(-1)^yC_\theta)}\rho}\\
	&=\frac{1}{4}\expec{\theta\leftarrow\pi}\Tr\squ*{\parens*{A_\theta\otimes (B_\theta\otimes\Id_C+\Id_B\otimes C_\theta)+\Id_{A}\otimes(\Id_{BC}+B_\theta\otimes C_\theta)}\rho}.
\end{split}
\end{align}
As in the case of a nonlocal game, we study the bias of a strategy rather than the winning probability, since it quantifies how much better or worse a strategy does than a random but coordinated guess.

\begin{definition}
	The \emph{bias} of a strategy $\ttt{S}$ for an MoE game $\ttt{G}$ is
	\begin{align}
		\mfk{b}_{\ttt{G}}(\ttt{S})=4\mfk{w}_{\ttt{G}}(\ttt{S})-2=\expec{\theta\leftarrow\pi}\Tr\squ*{\parens*{A_\theta\otimes (B_\theta\otimes\Id_C+\Id_B\otimes C_\theta)-\Id_{A}\otimes(\Id_{BC}-B_\theta\otimes C_\theta)}\rho}.
	\end{align}
	The optimal bias of the game is $\mfk{b}_{\ttt{G}}=\sup_{\ttt{S}}\mfk{b}_{\ttt{G}}(\ttt{S})$.
\end{definition}

The bias lives in the range [-2,2] and the bias of a strategy is $0$ if its winning probability is $\frac{1}{2}$. The optimal bias of the TFKW game is $\mfk{b}_{\ttt{TFKW}}=\sqrt{2}$.

To shorten expressions, we define $b_\theta=B_\theta\otimes\Id_C$, $c_\theta=\Id_B\otimes C_\theta$, and omit identities as much as possible, replacing them with $1$. Then, the bias is
\begin{align}
	\mfk{b}_{\ttt{G}}(\ttt{S})=\expec{\theta\leftarrow\pi}\Tr\squ*{\parens*{A_\theta\otimes(b_\theta+c_\theta)-1\otimes(1-b_\theta c_\theta)}\rho}.
\end{align}
For any strategy, we call the operator $\sum_{\theta\in\Theta}\pi(\theta)\parens*{A_\theta\otimes(b_\theta+c_\theta)-1\otimes(1-b_\theta c_\theta)}$ the \emph{game polynomial}.

For the TFKW game, the observables take the form of Pauli operators $A_0=Z$ and $A_1=X$, so the game polynomial is
\begin{align}
	\frac{1}{2}\parens*{Z\otimes(b_0+c_0)+X\otimes(b_1+c_1)-1\otimes(1-b_0c_0)-1\otimes(1-b_1c_1)}.\label{eq:game-poly}
\end{align}

A simple but powerful observation is that a value $\beta\in\R$ upper bounds the bias $\beta\geq\mfk{b}_{\ttt{G}}(\ttt{S})$ if
\begin{align}
	\beta-\expec{\theta\leftarrow\pi}\parens*{A_\theta\otimes(b_\theta+c_\theta)-1\otimes(1-b_\theta c_\theta)}\geq 0\label{eq:poly-pos}
\end{align}
as operators. It follows from consideration of the eigenvalues that the smallest value of $\beta$ for which this holds for any valid choice of the $b_\theta,c_\theta$ is the optimal bias $\mfk{b}_{\ttt{G}}$. Conversely, checking whether \cref{eq:poly-pos} holds for some fixed $\beta$ with any choice of observables provides a way to show that $\beta\geq\mfk{b}_{\ttt{G}}$.

This provides a way to upper bound the winning probability of an MoE game using a positivity argument. In particular, we consider whether a polynomial in a certain noncommutative algebra is positive under the matrix representations of the algebra. In language closer to \cite{Oza13}, the algebra we consider is the semi-pre-$\mathrm{C}^\ast$-algebra $\mc{L}(\tsf{A})\otimes\C\squ{\mbb{F}^2_{|\Theta|}\!\times\!\mbb{F}^2_{|\Theta|}}$, where $\mbb{F}_n^k$ is the free group with $n$ generators of order $k$. The first copy of the free group corresponds to Bob's observables, since the only relation we need impose on them is that they are self-inverse. Similarly, the second free group corresponds to Charlie's observables, and since Bob's observables commute with Charlie's, this is in Cartesian product with Bob's free group. The algebra is constructed as a matrix algebra over a semi-pre-$\mathrm{C}^\ast$-algebra by taking the group algebra $\C\squ{\mbb{F}^2_{|\Theta|}\!\times\!\mbb{F}^2_{|\Theta|}}$ and then extending the scalars to an algebra containing all of Alice's observables. Therefore, an element $P$ corresponding to the game polynomial belongs to this algebra, and a unitary representation where Bob and Charlie's observables are in tensor product corresponds to a strategy.

As highlighted in \cite{Oza13}, one way to approach positivity of elements in such an algebra is to use a sum-of-squares (SOS) argument. That is, if $\beta-P$, corresponding to the left hand side of \cref{eq:poly-pos}, admits a decomposition as a \emph{sum of Hermitian squares} of the form $\beta-P=\sum_iS_i^\dag S_i$, then it must be positive under any matrix representation simply because a Hermitian square is always a positive matrix. In fact, an SOS decomposition is guaranteed to exist for $\beta=\mfk{b}_{\ttt{G}}+\varepsilon$ for every $\varepsilon>0$ \cite{Oza13}. We will make use of an SOS decomposition with $\beta=\mfk{b}_{\ttt{G}}$ for the TFKW game.

\section{Rigidity of the TFKW Game}\label{sec:rigidity}

In this section, we prove the main result of the paper, that the TFKW game satisfies a rigidity condition. In \cref{sec:sos}, we give a sum-of-squares decomposition for the game polynomial of the TFKW game, which is used throughout the rigidity proofs.  We proceed progressively to show the rigidity. In \cref{sec:exact}, we show rigidity for TFKW game strategies that win optimally. In \cref{sec:robust}, we show rigidity for strategies that win nearly optimally. In \cref{sec:parallel}, we generalise the rigidity in the case of a single game to rigidity for a collection of games played in parallel that win optimally. In \cref{sec:robust-parallel}, we show rigidity for strategies for a collection of games that win nearly optimally. The main rigidity result is given by \cref{thm:rob-par}. Finally, in \cref{sec:obs-correlations}, we relate the winning probabilities of a collection of games to the winning statistics that the referee observes, which is used for applications of rigidity.

\subsection{Sum-of-Squares Decomposition}\label{sec:sos}

Let $P$ be game polynomial for $\ttt{TFKW}$. The polynomial $\mfk{b}_{\ttt{TFKKW}}-P=\sqrt{2}-P$ admits the following SOS decomposition:
\begin{align}
	\frac{1}{4\sqrt{2}}\squ*{\parens{Z\otimes b_0+X\otimes c_1-\sqrt{2}}^2+\parens{Z\otimes c_0+X\otimes b_1-\sqrt{2}}^2}+\frac{1}{4}\squ*{(b_0-c_0)^2+(b_1-c_1)^2}.\label{eq:sos}
\end{align}
The form of the decomposition takes inspiration from the SOS decomposition used to prove Tsirelson's bound for the CHSH game \cite{BaP15,CMMN20}. First, this directly implies that $\sqrt{2}$ upper bounds the bias of $\ttt{TFKW}$, giving an alternate proof of the winning probability to that of \cite{TFKW13}. Conversely, the state of an optimal strategy must be in the $0$ eigenspace of this operator, and therefore it must be in the $0$ eigenspace of each of the squared terms. We use this idea to work out the rigidity for this game.

\subsection{Exact Rigidity} \label{sec:exact}
	Before dealing with the more involved robust and eventually parallel-repeated rigidity, we can get a lot of intuition from working with the exact case, where we assume the strategy wins with exactly optimal probability.
	
	\begin{theorem}[exact rigidity]\label{thm:exact}
		Let $\ttt{S}=\parens*{\tsf{B},\tsf{C},B,C,\ketbra{\psi}}$ be a purified strategy for $\ttt{TFKW}$. If this strategy is optimal, then there exist Hilbert spaces $\tsf{B}'$, $\tsf{C}'$ and isometries $V:\tsf{B}\rightarrow\tsf{B}'$ and $W:\tsf{C}\rightarrow\tsf{C}'$ such that we have a decomposition of the state
		\begin{align}
		(V\otimes W)\ket{\psi}=\sum_{s\in\Z_2\times\Z_2}X^{s_0}Z^{s_1}\ket{\beta}\otimes\ket{\psi_s},
		\end{align}
		where the supports of the $\ket{\psi_s}\in\tsf{B}'
		\otimes\tsf{C}'\otimes\tsf{R}$ on both $\tsf{B}'$ and $\tsf{C}'$ are orthogonal; and there exist commuting operators $B_\theta'\in\mc{U}(\tsf{B}')$ and $C_\theta'\in\mc{U}(\tsf{C}')$ such that
		\begin{align}
		\begin{split}
			&VB_\theta\ket{\psi}=B_\theta'V\ket{\psi}\\
			&WC_\theta\ket{\psi}=C_\theta'W\ket{\psi},
		\end{split}\\
			&B_\theta'\ket{\psi_s}=C_\theta'\ket{\psi_s}=(-1)^{s_\theta}\ket{\psi_s}.
		\end{align}
	\end{theorem}

	It is a straightforward computation to show that a strategy of this form wins in fact optimally. Intuitively, the result says that what the players must do in order to win optimally is to agree on a Wiesner-Breidbart state to give Alice, which they can do without communicating using the simultaneous distinguishability of their parts of the state, and then guess accordingly. The Wiesner-Breidbart states can be seen as a family of states corresponding to the conjugate-coding states rotated by a Breidbart operator, as seen in \cref{fig:bloch}.
	
	\begin{proof}
		Letting $P$ be the game polynomial (\cref{eq:game-poly}), we know $\braket{\psi}{\sqrt{2}-P}{\psi}=0$. Then, as noted in the previous section, each of the terms in the sum of squares (\cref{eq:sos}) is positive so they must all be zero, giving four relations
		\begin{align}
		\label{eq:rels1}&(Z\otimes b_0+X\otimes c_1)\ket{\psi}=\sqrt{2}\ket{\psi}\\
		\label{eq:rels2}&(Z\otimes c_0+X\otimes b_1)\ket{\psi}=\sqrt{2}\ket{\psi}\\
		\label{eq:rels3}&b_0\ket{\psi}=c_0\ket{\psi}\\
		\label{eq:rels4}&b_1\ket{\psi}=c_1\ket{\psi}.
		\end{align}
		We can combine \cref{eq:rels1} and \cref{eq:rels4} to get a relation solely in terms of Alice and Bob's observables, $(Z\otimes b_0+X\otimes b_1)\ket{\psi}=\sqrt{2}\ket{\psi}$. Squaring this
		\begin{align}
			2\ket{\psi}=(Z\otimes b_0+X\otimes b_1)^2\ket{\psi}=2\ket{\psi}+2ZX\otimes[b_0,b_1]\ket{\psi},
		\end{align}
		we get that the commutator $[b_0,b_1]\ket{\psi}=0$, that is $b_0$ and $b_1$ commute with respect to $\ket{\psi}$. The commutation means that the group generated by $B_0$ and $B_1$ is a $(0,\ket{\psi})$-representation $f$ of $\Z_2^2$. We can alternately use \cref{lem:z22-gh} with $U_0=B_0$ and $U_1=B_1$ to get that these operators generate such a representation. By the Gowers-Hatami theorem (\cref{thm:gowers-hatami}), there exists an isometry $V:\tsf{B}\rightarrow\tsf{B}'$ and a representation $g:\Z_2^2\rightarrow\mc{U}(\tsf{B}')$ such that $Vf(x)\ket{\psi}=g(x)V\ket{\psi}$. Defining $B_0'=g(01)$ and $B_1'=g(10)$, these are commuting unitaries such that $VB_\theta\ket{\psi}=B_\theta'V\ket{\psi}$. Further, as $g$ is a representation, the dilated space decomposes orthogonally as a direct sum of irreducible representations
		\begin{align}
			\tsf{B}'=\bigoplus_{s\in\Z_2^2}\tsf{B}_s,
		\end{align}
		such that the operators act as $B_\theta'=\sum_{s\in\Z_2^2}(-1)^{s_\theta}\Id_{B,s}$, where $\Id_{B,s}$ is the projection onto $\tsf{B}_s$. Following an identical line of reasoning for Charlie's observables, there exists an isometry $W:\tsf{C}\rightarrow\tsf{C}'$ and commuting unitaries $C_0',C_1'\in\mc{U}(\tsf{C}')$ such that $WC_\theta\ket{\psi}=C_\theta'W\ket{\psi}$; and the space decomposes as $\tsf{C}'=\bigoplus_{s\in\Z_2^2}\tsf{C}_s$ so that $C_\theta'=\sum_{s\in\Z_2^2}(-1)^{s_\theta}\Id_{C,s}$. Defining the dilated state $\ket{\psi'}=(V\otimes W)\ket{\psi}$, we have that  \cref{eq:rels1,eq:rels2,eq:rels3,eq:rels4} extend to the dilated spaces:
		\begin{align}
		\label{eq:2rels1}&(Z\otimes B_0'+X\otimes B_1')\ket{\psi'}=\sqrt{2}\ket{\psi'}\\
		\label{eq:2rels3}&B_0'\ket{\psi'}=C_0'\ket{\psi'}\\
		\label{eq:2rels4}&B_1'\ket{\psi'}=C_1'\ket{\psi'}.
		\end{align}
		Now, since $\ket{\psi'}\in\tsf{A}\otimes\tsf{B}'\otimes\tsf{C}'\otimes\tsf{R}=\bigoplus_{s,s'\in\Z_2^2}\tsf{A}\otimes\tsf{B}_s\otimes\tsf{C}_{s'}\otimes\tsf{R}$, we can decompose it accordingly as $\ket{\psi'}=\sum_{s,s'\in\Z_2^2}\ket*{v_{s,s'}}$. Then, \cref{eq:2rels3} gives that $\sum_{s,s'\in\Z_2^2}(-1)^{s_0}\ket{v_{s,s'}}=\sum_{s,s'\in\Z_2^2}(-1)^{s'_0}\ket{v_{s,s'}}$, so $\ket{v_{s,s'}}=0$ if $s_0\neq s_0'$. Doing the same with \cref{eq:2rels4} gives that $\ket{v_{s,s'}}=0$ if $s\neq s'$ so
		\begin{align}
			\ket{\psi'}=\sum_{s\in\Z_2^2}\ket{v_{s,s}}\in\bigoplus_{s\in\Z_2^2}\tsf{A}\otimes\tsf{B}_s\otimes\tsf{C}_s\otimes\tsf{R}.
		\end{align}
		Next, the decomposition of the spaces means that
		\begin{align}
		\parens*{Z\otimes B_0+X\otimes B_1}\otimes\Id_{CR}=\sum_{s\in\Z_2^2}\parens*{(-1)^{s_0}Z+(-1)^{s_1}X}\otimes\Id_{B,s}\otimes\Id_{CR};
		\end{align}
		and \cref{eq:2rels1} says that $\ket{\psi'}$ must belong to the $\sqrt{2}$-eigenspace of this operator. Since
		\begin{align}
			(-1)^{s_0}Z+(-1)^{s_1}X=\sqrt{2}X^{s_0}Z^{s_1}HZ^{s_1}X^{s_0},
		\end{align}
		the $\sqrt{2}$-eigenspace is simply the span of $X^{s_0}Z^{s_1}\ket{\beta}$. Thus,
		\begin{align}
			\ket{\psi'}\in\bigoplus_{s,s'\in\Z_2^2}X^{s_0}Z^{s_1}\ket{\beta}\otimes\tsf{B}_s\otimes\tsf{C}_{s'}\otimes\tsf{R}.
		\end{align}
		Taking the intersection of the spaces $\ket{\psi'}$ belongs to, we have that
		\begin{align}
		\ket{\psi'}\in\bigoplus_{s\in\Z_2^2}X^{s_0}Z^{s_1}\ket{\beta}\otimes\tsf{B}_s\otimes\tsf{C}_s\otimes\tsf{R},
		\end{align}
		which gives the result.
	\end{proof}

\subsection{Robust Rigidity} \label{sec:robust}

Now, we move on to the study of the robust rigidity, where we assume that the winning probability is in some small neighbourhood of the optimal probability. We can approach the proof in about the same way as the exact case, while keeping track of the error.

\begin{theorem}[robust rigidity]\label{thm:robust}
	Let $\ttt{S}=\parens*{\tsf{B},\tsf{C},B,C,\ketbra{\psi}}$ be a purified strategy for $\ttt{TFKW}$ that wins with probability $\mathfrak{w}_{\ttt{TFKW}}(\ttt{S})\geq\cos^2\tfrac{\pi}{8}-\varepsilon$ for some $\varepsilon\geq 0$. Then there exists a constant $K\geq 0$ and isometries $V:\tsf{B}\rightarrow\tsf{B}'$ and $W:\tsf{C}\rightarrow\tsf{C}'$ such that the distance between quantum states
	\begin{align}
		\norm[\Big]{(V\otimes W)\ket{\psi}-\sum_{s\in\Z_2^2}X^{s_0}Z^{s_1}\ket{\beta}\otimes\ket{\psi_s}}\leq K\sqrt{\varepsilon},
	\end{align}
	where the $\ket{\psi_s}\in\tsf{B}'\otimes\tsf{C}'\otimes\tsf{R}$ have orthogonal supports on both $\tsf{B}'$ and $\tsf{C}'$; and there exists a constant $L\geq 0$, and commuting observables $B_\theta'\in\mc{U}(\tsf{B}')$ and $C_\theta'\in\mc{U}(\tsf{C}')$ such that
	\begin{align}
	\begin{split}
		&\norm*{VB_\theta\ket{\psi}-B_\theta'V\ket{\psi}}\leq L\sqrt{\varepsilon}\\
		&\norm*{WC_\theta\ket{\psi}-C_\theta'W\ket{\psi}}\leq L\sqrt{\varepsilon},
	\end{split}\\
		&B_\theta'\ket{\psi_s}=C_\theta'\ket{\psi_s}=(-1)^{s_\theta}\ket{\psi_s}.
	\end{align}
\end{theorem}

The proof below allows us to take $K=110$ and $L=18$ as the necessary constants. Note also that, as seen for the CHSH game in \cite{RUV13}, the order $\sqrt{\varepsilon}$ dependence of this upper bound is in fact necessary, though it may be possible to improve the constants: if we take an unentangled optimal strategy for $\ttt{TFKW}$ and perturb by a vector of length $\delta$ in an orthogonal direction, the winning probability decreases on the order of $\delta^2$.

\begin{proof}
	By hypothesis, $\mfk{w}_{\ttt{TFKW}}(\ttt{S})\geq\cos^2\tfrac{\pi}{8}-\varepsilon$, so the bias $\mfk{b}_{\ttt{TFKW}}(\ttt{S})\geq\sqrt{2}-4\varepsilon$, giving that $\braket{\psi}{\sqrt{2}-P}{\psi}\leq4\varepsilon$, which, using the sum-of-squares decomposition, is
	\begin{align}
	\begin{split}
	16\sqrt{2}\varepsilon\geq&\braket{\psi}{\parens{Z\otimes b_0+X\otimes c_1-\sqrt{2}}^2}{\psi}+\braket{\psi}{\parens{Z\otimes c_0+X\otimes b_1-\sqrt{2}}^2}{\psi}\\
	&+\sqrt{2}\squ*{\braket{\psi}{(b_0-c_0)^2}{\psi}+\braket{\psi}{(b_1-c_1)^2}{\psi}}.
	\end{split}
	\end{align}
	Since each of the terms is positive, we must have that $16\varepsilon\geq\braket{\psi}{(b_\theta-c_\theta)^2}{\psi}$ and $8\sqrt{2}\varepsilon\geq\min\curly*{\braket{\psi}{\parens{Z\otimes b_0+X\otimes c_1-\sqrt{2}}^2}{\psi},\braket{\psi}{\parens{Z\otimes c_0+X\otimes b_1-\sqrt{2}}^2}{\psi}}$. This can be converted to Euclidean norm conditions by taking square roots: \begin{align}
	\label{eq:axrel1}&2(8)^{1/4}\sqrt{\varepsilon}\geq\min\curly*{\norm*{\parens{Z\otimes b_0+X\otimes c_1-\sqrt{2}}\ket{\psi}},\norm*{\parens{Z\otimes c_0+X\otimes b_1-\sqrt{2}}\ket{\psi}}}\\
	\label{eq:axrel2}&4\sqrt{\varepsilon}\geq\norm*{(b_\theta-c_\theta)\ket{\psi}}.
	\end{align}
	Using \cref{eq:axrel2} in \cref{eq:axrel1}, we get
	\begin{align}
		\norm*{\parens{Z\otimes b_0+X\otimes b_1-\sqrt{2}}\ket{\psi}}\leq\norm*{\parens{Z\otimes b_0+X\otimes c_1-\sqrt{2}}\ket{\psi}}+\norm*{X\otimes(b_1-c_1)\ket{\psi}}
	\end{align}
	and
	\begin{align}
		\norm*{\parens{Z\otimes b_0+X\otimes b_1-\sqrt{2}}\ket{\psi}}\leq\norm*{\parens{Z\otimes c_0+X\otimes b_1-\sqrt{2}}\ket{\psi}}+\norm*{Z\otimes(b_0-c_0)\ket{\psi}}.
	\end{align}
	which gives $\norm*{\parens{Z\otimes b_0+X\otimes b_1-\sqrt{2}}\ket{\psi}}\leq 2(2+8^{1/4})\sqrt{\varepsilon}$.
	Noting that
	\begin{align}
	\begin{split}
		\parens{Z\otimes b_0+X\otimes b_1+\sqrt{2}}\parens{Z\otimes b_0+X\otimes b_1-\sqrt{2}}&=\parens{Z\otimes b_0+X\otimes b_1}^2-2\\
		&=ZX\otimes[b_0,b_1],
	\end{split}
	\end{align}
	we have that
	\begin{align}
	\begin{split}
	\norm{[b_0,b_1]\ket{\psi}}&=\norm*{ZX\otimes[b_0,b_1]\ket{\psi}}\\
	&\leq\norm*{Z\otimes b_0+X\otimes b_1+\sqrt{2}}\norm*{\parens{Z\otimes b_0+X\otimes b_1-\sqrt{2}}\ket{\psi}}\\
	&\leq 2(2+\sqrt{2})(2+8^{1/4})\sqrt{\varepsilon},
	\end{split}
	\end{align}
	that is, Bob's operators almost commute with respect to $\ket{\psi}$. As in the exact case, we use \cref{lem:z22-gh} with $U_0=B_0$ and $U_1=B_1$ to generate a $\parens*{\sqrt{2}(2+\sqrt{2})(2+8^{1/4})\sqrt{\varepsilon},\ket{\psi}}$-representation $f$ of $\Z_2^2$. By Gowers-Hatami, there exists an isometry $V:\tsf{B}\rightarrow\tsf{B}'$ to some Hilbert space and a representation $g:\Z_2^2\rightarrow\mc{U}(\tsf{B}')$ such that
	\begin{align}
		\norm*{\parens{Vf(x)-g(x)V}\ket{\psi}}\leq\sqrt{2}(2+\sqrt{2})(2+8^{1/4})\sqrt{\varepsilon}.
	\end{align}
	Defining $B_0'=g(01)$ and $B_1'=g(10)$, they are commuting observables such that
	\begin{align}
		\norm{\parens*{VB_\theta-B_\theta'V}\ket*{\psi}}\leq\sqrt{2}(2+\sqrt{2})(2+8^{1/4})\sqrt{\varepsilon};
	\end{align}
	and since $g$ is a representation, there exists an orthogonal decomposition $\tsf{B}'=\bigoplus_{s\in\Z_2^2}\tsf{B}_s$ where the observables decompose accordingly as $B_\theta'=\sum_{s\in\Z_2^2}(-1)^{s_\theta}\Id_{B,s}$. Applying the same reasoning for Charlie's observables gives that there exists a Hilbert space with orthogonal decomposition $\tsf{C}'=\bigoplus_{s\in\Z_2^2}\tsf{C}_s$, commuting observables $C_\theta'=\sum_{s\in\Z_2^2}(-1)^{s_\theta}\Id_{C,s}$, and an isometry $W:\tsf{C}\rightarrow\tsf{C}'$ such that $\norm*{\parens{WC_\theta-C_\theta'W}\ket*{\psi}}\leq\sqrt{2}(2+\sqrt{2})(2+8^{1/4})\sqrt{\varepsilon}$. Defining $\ket{\psi'}=(V\otimes W)\ket{\psi}$, we can extend \cref{eq:axrel1} and \cref{eq:axrel2} to the dilated spaces as
	\begin{align}
	&\label{eq:b2axrel1}\norm*{\parens{Z\otimes B_0'+X\otimes B_1'-\sqrt{2}}\ket{\psi'}}\leq2(3+2\sqrt{2})(2+8^{1/4})\sqrt{\varepsilon}\\
	&\label{eq:b2axrel2}\norm*{B_\theta'\ket{\psi'}-C_\theta'\ket{\psi'}}\leq 4\parens{(1+\sqrt{2})(2+8^{1/4})+1}\sqrt{\varepsilon}.
	\end{align}
	From the decomposition of Bob and Charlie's spaces, we have that the shared space is $\tsf{A}\otimes\tsf{B}'\otimes\tsf{C}'\otimes\tsf{R}=\bigoplus_{s,s'\in\Z_2^2}\tsf{A}\otimes\tsf{B}_s\otimes\tsf{C}_{s'}\otimes\tsf{R}$, thus the state decomposes accordingly as $\ket{\psi'}=\sum_{s,s'\in\Z_2^2}\ket{v_{s,s'}}$. Using this in \cref{eq:b2axrel2} gives
	\begin{align}
	4\parens{(1+\sqrt{2})(2+8^{1/4})+1}\sqrt{\varepsilon}\geq\norm[\Big]{\sum_{s,s'\in\Z_2^2}\parens*{(-1)^{s_\theta}-(-1)^{s'_\theta}}\ket{v_{s,s'}}}=2\norm[\Big]{\sum_{s_\theta\neq s_\theta'}\ket{v_{s,s'}}}
	\end{align}
	We write $\ket{v_0}=\sum_{s}\ket{v_{s,s}}$ and $\ket{v_1}=\sum_{s\neq s'}\ket{v_{s,s'}}$, so that $\ket{\psi'}=\ket{v_0}+\ket{v_1}$ and
	\begin{align}
		&\norm*{\ket{v_1}}\leq\norm[\Big]{\sum_{s_0\neq s_0'}\ket{v_{s,s'}}}+\norm[\Big]{\sum_{s_1\neq s_1'}\ket{v_{s,s'}}}\leq 4\parens{(1+\sqrt{2})(2+8^{1/4})+1}\sqrt{\varepsilon}.
	\end{align}
	Writing $\ket{\beta_s}=X^{s_0}Z^{s_1}\ket{\beta}$, we can decompose
	\begin{align}
	\begin{split}
		Z\otimes B_0'+X\otimes B_1'-\sqrt{2}&=\sum_{s}\parens{(-1)^{s_0}Z+(-1)^{s_1}X-\sqrt{2}}\otimes\Id_{B,s}\\
		&=2\sqrt{2}\sum_{s}(\ketbra{\beta_s}-\Id)\otimes\Id_{B,s}.
	\end{split}
	\end{align}
	Also, define the projection $\ket{v_\beta}=\sum_s(\ketbra*{\beta_s}\otimes\Id)\ket{v_{s,s}}$, so that \cref{eq:b2axrel1} implies
	\begin{align}
	\begin{split}
	\norm*{\ket{\psi'}-\ket{v_\beta}}&\leq\norm[\Big]{\sum_{s,s'}\parens*{\ketbra{\beta_s}-\Id}\ket{v_{s,s'}}}+\norm[\Big]{\sum_{s\neq s'}\parens*{\ketbra{\beta_s}\otimes\Id}\ket*{v_{s,s'}}}\\
	&\leq\frac{1}{2\sqrt{2}}\norm*{\parens{Z\otimes B_0'+X\otimes B_1'-\sqrt{2}}\ket{\psi'}}+\norm*{\ket{v_1}}\\
	&\leq (\tfrac{3}{2}\sqrt{2}+2)(2+8^{1/4})\sqrt{\varepsilon}+4\parens{(1+\sqrt{2})(2+8^{1/4})+1}\sqrt{\varepsilon}\\
	&=\squ{(6+\tfrac{11}{2}\sqrt{2})(2+8^{1/4})+4}\sqrt{\varepsilon}.
	\end{split}
	\end{align}
	Note that although $\ket*{v_\beta}$ is not necessarily normalised, it must be subnormalised and the above implies that
	\begin{align}
	\norm*{\frac{\ket{v_\beta}}{\norm*{\ket{v_\beta}}}-\ket*{v_\beta}}=1-\norm*{\ket{v_\beta}}\leq\norm*{\ket{\psi'}-\ket{v_\beta}}\leq \squ{(6+\tfrac{11}{2}\sqrt{2})(2+8^{1/4})+4}\sqrt{\varepsilon}.
	\end{align}
	Defining $\ket{\phi}=\frac{\ket{v_\beta}}{\norm*{\ket{v_\beta}}}$, we have by construction that $\ket{\phi}=\sum_s\ket{\beta_s}\otimes\ket{\psi_s}$, where
	\begin{align}
		\ket{\psi_s}=\frac{1}{\norm*{v_\beta}}(\bra{\beta_s}\otimes\Id)\ket{v_{s,s}}\in\tsf{B}_s\otimes\tsf{C}_s\otimes\tsf{R},
	\end{align}
	so simultaneously distinguishable by Bob and Charlie. Thus, to complete the proof, note that
	\begin{align}
	\norm{\ket*{\psi'}-\ket*{\phi}}\leq\norm{\ket*{\psi'}-\ket*{v_\beta}}+\norm{\ket*{v_\beta}-\ket*{\phi}}\leq 2\squ{(6+\tfrac{11}{2}\sqrt{2})(2+8^{1/4})+4}\sqrt{\varepsilon}.
	\end{align}
\end{proof}

We can use the properties of purified strategies and the trace norm to directly extend this result to a general strategy.

\begin{corollary}\label{cor:gen-rob}
	Let $\ttt{S}=\parens*{\tsf{B},\tsf{C},B,C,\rho}$ be an arbitrary strategy for $\ttt{TFKW}$ that wins with probability $\mfk{w}_{\ttt{TFKW}}(\ttt{S})\geq\cos^2\tfrac{\pi}{8}-\varepsilon$ for some $\varepsilon\geq 0$. Then there exists a constant $K\geq 0$ and isometries $V:\tsf{B}\rightarrow\tsf{B}'$ and $W:\tsf{C}\rightarrow\tsf{C}'$ such that
	\begin{align}
	\norm*{(V\otimes W)\rho(V\otimes W)^\dag-\Tr_{R}(\ketbra{\phi})}_{\Tr}\leq K\sqrt{\varepsilon},
	\end{align}
	where $R$ is an auxiliary register such that $\ket{\phi}=\sum_{s\in\Z_2^2}X^{s_0}Z^{s_1}\ket{\beta}\otimes\ket{\psi_s}$ for some vectors $\ket{\psi_s}\in\tsf{B}'\otimes\tsf{C}'\otimes\tsf{R}$ with orthogonal supports on both $\tsf{B}'$ and $\tsf{C}'$.
\end{corollary}

The proof of \cref{cor:gen-rob}
follows directly by using the inequality between the Euclidean distance and the trace distance (\cref{lem:tr-dist}), tracing out the auxiliary register $R$, and finally using the fact that the purification of the measurements only requires an isometric extension of the state space (\cref{lem:pure-pvm}).

\subsection{Rigidity under Parallel Repetition}\label{sec:parallel}

	Similarly to the case of a single game, we begin with the parallel repetition in the exact case. That is, we assume $n$ copies of the TFKW game are played and the adversaries win each of the copies with optimal probability. We aim to show that, in this case, Bob and Charlie must behave as for a single game on each of the copies, i.e. they agree upon a Wiesner-Breidbart state and guess accordingly.
	
	\begin{theorem}[parallel-repeated exact rigidity]\label{thm:exact-parallel}
		Let $n\in\N$ and let $\ttt{S}=\parens*{\tsf{B},\tsf{C},B,C,\rho=\ketbra{\psi}}$ be a purified strategy for $\ttt{TFKW}^n$ that guesses each bit optimally, that is for each $i\in[n]$, $\mfk{w}_{\ttt{TFKW}^n}^i(\ttt{S})=\cos^2\parens*{\frac{\pi}{8}}$. Then, there exist Hilbert spaces $\tsf{B}'$ and $\tsf{C}'$, and isometries $V:\tsf{B}\rightarrow\tsf{B}'$ and $W:\tsf{C}\rightarrow\tsf{C}'$ such that
		\begin{align}
			&(V\otimes W)\ket{\psi}=\sum_{t\in(\Z_2^2)^n}X^{t_{10}}Z^{t_{11}}\ket{\beta}\otimes\cdots\otimes X^{t_{n0}}Z^{t_{n1}}\ket{\beta}\otimes\ket{\psi_t},
		\end{align}
		where the supports of the $\ket{\psi_t}\in\tsf{B}'\otimes\tsf{C}'\otimes\tsf{R}$ on both $\tsf{B}'$ and $\tsf{C}'$ are orthogonal.
	\end{theorem}
	Note that we are writing strings $t\in(\Z_2^2)^n$ as $t=t_{10}t_{11}t_{20}t_{21}\ldots t_{n0}t_{n1}$. To prove this theorem, we want to reduce to the single-game case as much as possible and use the rigidity we know there. As such, we extract a collection of optimal strategies for a single TFKW game. In fact, we may express the $i$-th winning probability as
	\begin{align}
		\mfk{w}^i_{\ttt{TFKW}^n}(\ttt{S})&=\expec{\substack{\varphi\in\Z_2^n\\\varphi_i=0}}\parens*{\frac{1}{2}\sum_{\substack{\theta\in\Z_2^n\\\theta_j=\varphi_j\forall j\neq i}}\sum_{y\in\Z_2}\Tr\squ*{\parens*{(A^n)^\theta_{y,i}\otimes B^\theta_{y,i}\otimes C^\theta_{y,i}}\rho}}.
	\end{align}
	In order for this average to be $\cos^2\parens*{\tfrac{\pi}{8}}$, each of the inner terms must also be $\cos^2\parens*{\tfrac{\pi}{8}}$, and thus they must correspond to an optimal strategy of $\ttt{TFKW}$. Then we get $n2^{n-1}$ optimal strategies: for every $i\in[n]$ and $\varphi\in\Z_2^n$ such that $\varphi_i=0$, the strategy ${^{\varphi,i}\ttt{S}}=\parens*{\tsf{B},\tsf{C},{^{\varphi,i\!}B},{^{\varphi,i}C},\rho}$ where ${^{\varphi,i\!}B}^{\theta}_y=B^{\varphi+\theta1^i}_{y,i}$ and ${^{\varphi,i}C}^{\theta}_y=C^{\varphi+\theta1^i}_{y,i}$ is an optimal strategy for $\ttt{TFKW}$, assuming that Alice measures on her $i$-th qubit, i.e. $A^\theta_y=(A^n)^{\theta1^i}_{y,i}$. Before going ahead to the proof, we prove an important lemma that allows us to relate strategies of this form.
	
	\begin{lemma}\label{lem:2-strat}
		Let ${^{0}\ttt{S}}$ and ${^{1}\ttt{S}}$ be two purified optimal strategies for $\ttt{TFKW}$. Suppose their shared states are equal, $\ket{\psi}=\ket{{^{0}\psi}}=\ket{{^{1}\psi}}$. Then we can choose that the local dilation operations be the same for both strategies and, in that case, the rigidity decompositions of the two states must be identical.
	\end{lemma}
	
	As before, we write $\ket{\beta_s}=X^{s_0}Z^{s_1}\ket{\beta}$.
	
	\begin{proof}
		Using \cref{thm:exact}, for each $i=0,1$ there exist Hilbert spaces with orthogonal decompositions ${^i\tsf{B}'}=\bigoplus_{s\in\Z_2^2}{^i\tsf{B}_s}$ and ${^i\tsf{C}'}=\bigoplus_{s\in\Z_2^2}{^i\tsf{C}_s}$; isometries ${^iV}:{^i\tsf{B}}\rightarrow{^i\tsf{B}'}$ and ${^iW}:{^i\tsf{C}}\rightarrow{^i\tsf{C}'}$; and for each $s\in\Z_2^2$ vectors $\ket{\psi_s^i}\in{^i\tsf{B}_s}\otimes{^i\tsf{C}_s}\otimes\tsf{R}$ such that
		\begin{align}
		({^iV}\otimes{^iW})\ket{\psi}=\sum_{s\in\Z_2^2}\ket{\beta_s}\otimes\ket{\psi_s^i}.
		\end{align}
		Further, again following from the exact rigidity, for each $\theta\in\Z_2$ there exist PVMs ${^{i}\!B'}^\theta:\Z_2\rightarrow\mc{P}(^i\tsf{B}')$ and ${^{i}C'}^\theta:\Z_2\rightarrow\mc{P}(^i\tsf{C}')$ such that ${^iV}{\,^i\!B}^\theta_y\ket{\psi}={^i\!B'}^\theta_y{^iV}\ket{\psi}$, ${^i\!B'}^\theta_y\ket{\psi_s^i}=\delta_{y,s_\theta}\ket{\psi_s^i}$, and $[{^i\!B'}^0_y,{^i\!B'}^1_{y'}]=0$; and identically for the ${^iC'}^\theta_y$. First we show that the dilation unitaries can be constructed so that they are identical for $i=0,1$. Let $\tsf{B}'={^0\tsf{B}'}\oplus{^1\tsf{B}'}$ and $\tsf{C}'={^0\tsf{C}'}\oplus{^1\tsf{C}'}$, so the isometries $^iV$ and $^iW$ can be seen as isometries into $\tsf{B}'$ and $\tsf{C}'$ respectively. Since the images of $^0V$ and $^1V$ have the same dimension in $\tsf{B}'$, there exist unitaries $^0U,{^1U}\in\mc{U}(\tsf{B}')$ such that ${^0U}{^0V}={^1U}{^1V}=:V$. Thus, we may redefine $^0V$ and $^1V$ to be this. Doing the same for $C$, to get $W$, we may assume that the dilation operators are the same for both strategies.
		
		Define $\ket{\psi'}=(V\otimes W)\ket{\psi}$, and $^{i}\Pi^B_{s}={^i\!B'}^0_{s_0}{^i\!B'}^1_{s_1}$ and $^{i}\Pi^C_{s}={^iC'}^0_{s_0}{^iC'}^1_{s_1}$, the projectors onto ${^i\tsf{B}_s}$ and ${^i\tsf{C}_s}$, respectively. Expanding the two expressions of $\ket{\psi'}$ in the basis $\{\ket{\beta_{00}},\ket{\beta_{11}}\}$ of $\tsf{A}$, we get the relations
		\begin{align}
		\begin{split}
		&\ket{\psi_{00}^0}+\tfrac{1}{\sqrt{2}}\parens*{\ket{\psi_{01}^0}+\ket{\psi_{10}^0}}=\ket{\psi_{00}^1}+\tfrac{1}{\sqrt{2}}\parens*{\ket{\psi_{01}^1}+\ket{\psi_{10}^1}}\\
		&\ket{\psi_{11}^0}+\tfrac{1}{\sqrt{2}}\parens*{\ket{\psi_{01}^0}-\ket{\psi_{10}^0}}=\ket{\psi_{11}^1}-\tfrac{1}{\sqrt{2}}\parens*{\ket{\psi_{01}^1}+\ket{\psi_{10}^1}}.
		\end{split}
		\end{align}
		Projecting the second relation onto ${^0\tsf{B}_{00}}\otimes\tsf{C}'$ gives $0={^0\Pi^B_{00}}\ket{\psi_{11}^1}-\tfrac{1}{\sqrt{2}}\parens*{{^0\Pi^B_{00}}\ket{\psi_{01}^1}+{^0\Pi^B_{00}}\ket{\psi_{10}^1}}$, and projecting this onto ${^0\tsf{B}_{00}}\otimes{^1\tsf{C}_{s}}$ for $s=01,10,11$ gives ${^0\Pi^B_{00}}\ket{\psi_s^1}=0$. Thus, projecting the first relation onto ${^0\tsf{B}_{00}}\otimes\tsf{C}'$ gives $\ket{\psi_{00}^0}={^0\Pi^B_{00}}\ket{\psi_{00}^1}$. Repeating a similar procedure for each $s\in\Z_2^2$ gives $\ket{\psi_s^0}={^0\Pi^B_{s}}\ket{\psi_s^1}$. It remains to show that the projectors act as the identity on these states. Suppose there exists $s$ such that ${^0\Pi^B_{s}}$ does not preserve $\ket{\psi_s^1}$. Then, $\braket{\psi_s^0}=\braket{\psi_s^1}{{^0\Pi^B_{s}}}{\psi_s^1}<\braket{\psi_s^1}$. However, we have then
		\begin{align}
			1=\braket{\psi'}=\sum_{t\in\Z_2^2}\braket{\psi_t^0}<\sum_{t\in\Z_2^2}\braket{\psi_t^1}=1,
		\end{align}
		which is a contradiction. Thus, $\ket{\psi_s^0}=\ket{\psi_s^1}$, so the rigidity decompositions are identical.
	\end{proof}

	\begin{proof}[Proof of \cref{thm:exact-parallel}]
		Knowing that the strategies $^{\varphi,i}\ttt{S}$ are optimal, we can use the exact rigidity of \cref{thm:exact} to get that there exist Hilbert spaces with orthogonal decompositions $^{\varphi,i}\tsf{B}'=\bigoplus_{s\in\Z_2^2}{^{\varphi,i}\tsf{B}_s}$ and $^{\varphi,i}\tsf{C}'=\bigoplus_{s\in\Z_2^2}{^{\varphi,i}\tsf{C}_s}$; isometries $^{\varphi,i}V:\tsf{B}\rightarrow{^{\varphi,i}\tsf{B}'}$ and $^{\varphi,i}W:\tsf{C}\rightarrow{^{\varphi,i}\tsf{C}'}$; and vectors $\ket{\psi^{\varphi,i}_s}\in\tsf{A}_1\otimes\cdots\otimes\tsf{A}_{i-1}\otimes\tsf{A}_{i+1}\otimes\cdots\otimes\tsf{A}_n\otimes{^{\varphi,i}\tsf{B}}_s\otimes{^{\varphi,i}\tsf{C}}_s\otimes\tsf{R}$ such that
		\begin{align}
		({^{\varphi,i}V}\otimes{^{\varphi,i}W})\ket{\psi}=\sum_{s\in\Z_2^2}\ket{\beta_s}_i\otimes\ket{\psi^{\varphi,i}_s},
		\end{align}
		where we use the subscript $i$ to indicate that the state lives in the space $\tsf{A}_i$. Using the same construction as the first part of \cref{lem:2-strat}, we can assume that the spaces $^{\varphi,i}\tsf{B}'=\tsf{B}'$ and the isometries $^{\varphi,i}V=V$ for all $\varphi,i$; and similarly for Charlie's. Then, by the lemma again, $\ket{\psi^{\varphi,i}_s}$ is constant over all values of $\varphi$, so we write $\ket*{\psi^i_s}$ for this state. Still using the rigidity, there exist unitary observables $^{\varphi,i\!}B_{\theta}'\in\mc{U}(\tsf{B'})$ and $^{\varphi,i}C_{\theta}'\in\mc{U}(\tsf{C'})$ such that
		\begin{align}
			&V\;{^{\varphi,i\!}B}_\theta\ket{\psi}={^{\varphi,i\!}B'_{\theta}}V\ket{\psi}\\ &{^{\varphi,i\!}B_{\theta}'}\ket{\psi^i_s}=(-1)^{s_{\theta}}\ket{\psi^i_s}\\
			&[{^{\varphi,i\!}B_{\theta}'},{^{\varphi,i\!}B_{\theta+1}'}]=0;
		\end{align}
		and identically for Charlie's observables. Note that, writing $\ket{\psi'}=(V\otimes W)\ket{\psi}$, these relations imply that like the original observables
		\begin{align}
			{^{\varphi,i\!}B'_\theta}\ket{\psi'}={^{\varphi,i}C'_\theta}\ket{\psi'}.
		\end{align}
		The rigidity relations also imply that
		\begin{align}
			{^{\varphi,i}B_\theta'}\ket{\psi'}={^{\varphi',i}B_\theta'}\ket{\psi'}
		\end{align}
		for valid values of $\varphi,\varphi'$.
		
		The first goal is to show that all of the ${^{\varphi,i\!}B'_\theta}$ commute with respect to $\ket{\psi'}$. For $\chi\in\Z_2^n$, write ${^{i\!}B_\chi}={^{\chi+\chi_i1^i,i\!}B_{\chi_i}}$ and similarly for the dilated operators, which simplifies the work a bit. From \cref{eq:commutation} in \cref{sec:moeg-defs}, the operator ${^{i\!}B_\chi}$ commutes with ${^{j\!}B_{\chi}}$. This extends directly to the dilated operators as
		\begin{align}
		\begin{split}
			{^{i\!}B'_\chi}{^{j\!}B'_{\chi}}\ket{\psi'}&={^{i\!}B'_\chi}\otimes{^{j}C'_{\chi}}\ket{\psi'}=(V\otimes W){^{i\!}B_\chi}\otimes{^{j}C_{\chi}}\ket{\psi}\\
			&=(V\otimes W){^{i\!}B_\chi}{^{j\!}B_{\chi}}\ket{\psi}=(V\otimes W){^{j\!}B_{\chi}}{^{i\!}B_{\chi}}\ket{\psi}\\
			&={^{j\!}B_{\chi}'}{^{i\!}B_{\chi}'}\ket{\psi'},
		\end{split}
		\end{align}
		so ${^{i\!}B'_\chi}$ and ${^{j}B'_{\chi}}$ commute with respect to $\ket{\psi'}$. We extend this to all the observables using \cref{lem:2-strat}. Take any $i,j\in[n]$, $\chi,\chi'\in\Z_2^n$. If $i=j$, then let $\xi=\chi+(\chi_i+\chi'_i)1^i$. We have that ${^iB'_\chi}$ and ${^iB'_{\xi}}$ commute as they are the observables from the same game and ${^{i\!}B'_{\chi'}}\ket*{\psi'}={^{i\!}B'_\xi}\ket*{\psi'}$ as they are equal on the $i$-th bit, so
		\begin{align}
			{^{i\!}B'_\chi}{^{i\!}B'_{\chi'}}\ket{\psi'}={^{i\!}B'_\chi}{^{i\!}B'_{\xi}}\ket{\psi'}={^{i\!}B'_{\xi}}{^{i\!}B'_\chi}\ket{\psi'}={^{i\!}B'_{\xi}}\otimes{^iC'_\chi}\ket{\psi'}={^{i\!}B'_{	\chi'}}\otimes{^iC'_\chi}\ket{\psi'}={^{i\!}B'_{\chi'}}{^{i\!}B'_\chi}\ket{\psi'}.
		\end{align}
		If $i\neq j$, there exists a $\xi$ such that $\xi_i=\chi_i$ and $\xi_j=\chi'_j$. Then, we have that
		\begin{align}
			{^{i\!}B'_\chi}{^{j\!}B'_{\chi'}}\ket{\psi'}={^{i\!}B'_\chi}\otimes{^{j\!}C'_{\chi'}}\ket{\psi'}={^{i\!}B'_{\xi}}\otimes{^{j\!}C'_{\xi}}\ket{\psi'}={^{i\!}B'_{\xi}}{^{j\!}B'_{\xi}}\ket{\psi'}={^{j\!}B'_{\xi}}{^{i\!}B'_{\xi}}\ket*{\psi'}={^{j\!}B'_{\chi'}}{^{i\!}B'_{\chi}}\ket{\psi'}.
		\end{align}
		Thus, all of the observables commute.
		
		Consider the group generated by the observables ${^{i\!}B_{\theta1^i}'}$ for $\theta\in\Z_2$ and $i\in[n]$. The commutation implies that this is a $(0,\ket{\psi})$-representation $f$ of $(\Z_2^2)^n$. This holds in the same way for Charlie's observables. Applying Gowers-Hatami, there exist Hilbert spaces with orthogonal decompositions $\tsf{B}''=\bigoplus_{t\in(\Z_2^2)^n}\tsf{B}_t$ and $\tsf{C}''=\bigoplus_{t\in(\Z_2^2)^n}\tsf{C}_t$; isometries $V':\tsf{B}'\rightarrow\tsf{B}''$ and $W':\tsf{C}'\rightarrow\tsf{C}''$; and observables that align with the decomposition ${^{i\!}B_{\theta}''}=\sum_{t\in(\Z_2^2)^n}(-1)^{t_{i\theta}}\Id_{B,t}$ and ${^{i\!}C_{\theta}''}=\sum_{t\in(\Z_2^2)^n}(-1)^{t_{i\theta}}\Id_{C,t}$ such that $V'\,{^{i\!}B_{\theta1^i}'}\ket{\psi'}={^{i\!}B''_{\theta}}V'\ket{\psi'}$ and $W'\,{^iC_{\theta1^i}'}\ket{\psi'}={^iC''_{\theta}}W'\ket{\psi'}$. By construction, ${^{i\!}B''_{\theta}}V'\ket{\psi^i_s}=(-1)^{s_\theta}V'\ket{\psi^i_s}$. Thus, the support of $V'\ket{\psi^i_s}$ on $\tsf{B}''$ is contained in the span of the subspaces $\tsf{B}_t$ such that $t_i=s$. Since an analogous inclusion holds for Charlie's space, we get that
		\begin{align}
			(V'\otimes W')\ket{\psi^i_s}\in\bigoplus_{t_i={t_i}'=s}\tsf{A}_1\otimes\cdots\otimes\tsf{A}_{i-1}\otimes\tsf{A}_{i+1}\otimes\cdots\otimes\tsf{A}_n\otimes\tsf{B}_t\otimes\tsf{C}_{t'}\otimes\tsf{R}.
		\end{align}
		Defining $\ket*{\psi''}=(V'\otimes W')\ket{\psi}$, this gives that
		\begin{align}
			\ket{\psi''}\in\bigoplus_{\substack{t,t'\in(\Z_2^2)^n\\t_i={t_i}'}}\tsf{A}_1\otimes\cdots\otimes\tsf{A}_{i-1}\otimes\ket{\beta_{t_i}}\otimes\tsf{A}_{i+1}\otimes\cdots\otimes\tsf{A}_n\otimes\tsf{B}_t\otimes\tsf{C}_{t'}\otimes\tsf{R}
		\end{align}
		for all $i$. Taking the intersection of all these spaces, which is easy as the $\tsf{B}_t$ and $\tsf{C}_{t'}$ are orthogonal, we end up with
		\begin{align}
		\ket{\psi''}\in\bigoplus_{t\in(\Z_2^2)^n}\ket{\beta_{t_1}}\otimes\cdots\otimes \ket{\beta_{t_n}}\otimes\tsf{B}_t\otimes\tsf{C}_{t}\otimes\tsf{R}
		\end{align}
	\end{proof}
	
	We see also from the proof that ${^{\varphi,i\!}B_\theta'}\ket{\psi_t}={^{\varphi,i}C_\theta'}\ket{\psi_t}=(-1)^{t_{i\theta}}\ket{\psi_t}$.

\subsection{Parallel-Repeated Robust Rigidity}\label{sec:robust-parallel}

	Now, we consider the robust rigidity of the parallel-repeated game. We want to approach it in about the same way in the exact case, so first we need a generalisation of \cref{lem:2-strat} to the approximate case.
	
	\begin{lemma}\label{lem:2-strat-rob}
		Let ${^0\ttt{S}}$ and ${^1\ttt{S}}$ be purified strategies for $\ttt{TFKW}$ that both win with probability $\mfk{w}_{\ttt{TFKW}}({^i\ttt{S}})\geq\cos^2\parens*{\frac{\pi}{8}}-\delta$ for some \mbox{$\delta\geq 0$}. If we suppose their shared states are equal, $\ket{\psi}=\ket{{^0\psi}}=\ket{{^1\psi}}$, then there is a constant $Q\geq 0$ such that for every $\theta\in\Z_2$,
		\begin{align}
		\begin{split}
			&\norm*{{^{0\!}B_\theta}\ket{\psi}-{^{1\!}B_\theta}\ket{\psi}}\leq Q\sqrt{\delta}\\
			&\norm*{{^0C_\theta}\ket{\psi}-{^1C_\theta}\ket{\psi}}\leq Q\sqrt{\delta}.
		\end{split}
		\end{align}
	\end{lemma}

	The proof below lets us take $Q=6300$.

	\begin{proof}
		For each of the strategies, we use robust rigidity \cref{thm:robust} where by the method of \cref{lem:2-strat} we may assume that the dilation operators are equal. Then, there exist constants $K,L\geq 0$; Hilbert spaces with two orthogonal decompositions $\tsf{B}'=\bigoplus_{s\in\Z_2^2}{^i\tsf{B}_s}$ and $\tsf{C}'=\bigoplus_{s\in\Z_2^2}{^i\tsf{C}_s}$; isometries $V:\tsf{B}\rightarrow\tsf{B}'$ and $W:\tsf{C}\rightarrow\tsf{C}'$; and for each $s\in\Z_2^2$ vectors $\ket*{\psi^i_s}\in{^i\tsf{B}_s}\otimes{^i\tsf{C}_s}\otimes\tsf{R}$ such that
		\begin{align}
			\norm[\Big]{(V\otimes W)\ket{\psi}-\sum_{s\in\Z_2^2}\ket{\beta_s}\otimes\ket{\psi^i_s}}\leq K\sqrt{\delta}.
		\end{align}
		Further, for each $\theta\in\Z_2$, there exist unitary observables (and related PVMs) ${^{i\!}B'_\theta}\in\mc{U}(\tsf{B}')$ and ${^iC'_\theta}\in\mc{U}(\tsf{C}')$ such that
		\begin{align}
		\begin{split}
			&\norm*{(V{^{i\!}B_\theta}-{^{i\!}B_\theta'}V)\ket*{\psi}}\leq L\sqrt{\delta},\\ &\norm*{(W\,{^iC_\theta}-{^iC_\theta'}W)\ket{\psi}}\leq L\sqrt{\delta},\\ &{^{i\!}B'_\theta}\ket{\psi^i_s}={^iC'_\theta}\ket{\psi^i_s}=(-1)^{s_\theta}\ket{\psi^i_s},\\ &[{^{i\!}B'_0},{^{i\!}B'_1}]=0,\\
			&[{^iC'_0},{^iC'_1}]=0.
		\end{split}
		\end{align}
		First, using the triangle inequality, the distance between the two rigidity decompositions~is
		\begin{align}
			\norm[\Big]{\sum_{s\in\Z_2^2}\ket{\beta_s}\otimes\ket{\psi^0_s}-\sum_{s\in\Z_2^2}\ket{\beta_s}\otimes\ket{\psi^1_s}}\leq 2K\sqrt{\delta}.
		\end{align}
		Expanding $\tsf{A}$ in the basis $\{\ket*{\beta_{00}},\ket*{\beta_{11}}\}$, this gives that both
		\begin{align}
		\begin{split}
		&\norm*{\parens*{\ket{\psi^0_{00}}+\tfrac{1}{\sqrt{2}}\parens{\ket{\psi^0_{01}}+\ket{\psi^0_{10}}}}-\parens*{\ket{\psi^1_{00}}+\tfrac{1}{\sqrt{2}}\parens{\ket{\psi^1_{01}}+\ket{\psi^1_{10}}}}}\leq2K\sqrt{\delta}\\
		&\norm*{\parens*{\ket{\psi^0_{11}}+\tfrac{1}{\sqrt{2}}\parens{\ket{\psi^0_{01}}-\ket{\psi^0_{10}}}}-\parens*{\ket*{\psi^1_{11}}-\tfrac{1}{\sqrt{2}}\parens{\ket{\psi^1_{01}}-\ket{\psi^1_{10}}}}}\leq2K\sqrt{\delta}.
		\end{split}\label{eq:2-ineqs}
		\end{align}
		Again as in \cref{lem:2-strat}, we act by projectors of the form ${^i\Pi^B_s}={^{i\!}B'}^0_{s_0}{^{i\!}B'}^1_{s_1}$ and ${^i\Pi^C_s}={^iC'}^0_{s_0}{^iC'}^1_{s_1}$. Since the action of a projector cannot increase the norm, acting by ${^0\Pi^B_{00}}\otimes {^1\Pi^C_s}$ on the second inequality of \cref{eq:2-ineqs} gives
		\begin{align}
			\norm{{^0\Pi^B_{00}}\ket{\psi^1_{11}}}\leq 2K\sqrt{\delta}\text{ and }\norm{{^0\Pi^B_{00}}\ket{\psi^1_{01}}},\norm{{^0\Pi^B_{00}}\ket{\psi^1_{10}}}\leq 2\sqrt{2}K\sqrt{\delta}.
		\end{align}
		Then, acting by ${^0\Pi^B_{00}}$ on the first inequality of \cref{eq:2-ineqs} leads to $\norm{\ket{\psi^0_{00}}-{^0\Pi^B_{00}}\ket{\psi^1_{00}}}\leq6K\sqrt{\delta}$. Similar is true for other values of $s$ for the projectors, so
		\begin{align}
			\norm*{\ket{\psi^0_{00}}-\ket{\psi^1_{00}}}\leq\norm*{\ket{\psi^0_{00}}-{^0\Pi^B_{00}}\ket{\psi^1_{00}}}+\sum_{s\neq00}\norm*{{^0\Pi^B_{s}}\ket{\psi^1_{00}}}\leq 4(2+\sqrt{2})K\sqrt{\delta}.
		\end{align}
		The same thing holds in the same way for the other values of $s$ in the ket. Then
		\begin{align}
		\begin{split}
			\norm*{({^{0\!}B_\theta}-{^{1\!}B_\theta})\ket{\psi}}&=\norm*{(V\otimes W)({^{0\!}B_\theta}-{^{1\!}B_\theta})\ket{\psi}}\\
			&\leq\norm*{({^{0\!}B'_\theta}-{^{1\!}B'_\theta})(V\otimes W)\ket{\psi}}+2L\sqrt{\delta}\\
			&\leq\norm[\Big]{{^0B'_\theta}\sum_{s\in\Z_2^2}\ket{\beta_s}\otimes\ket{\psi^0_s}-{^1B'_\theta}\sum_{s\in\Z_2^2}\ket{\beta_s}\otimes\ket{\psi^1_s}}+2K\sqrt{\delta}+2L\sqrt{\delta}\\
			&=\norm[\Big]{\sum_{s\in\Z_2^2}(-1)^{s_\theta}\ket{\beta_s}\otimes(\ket{\psi^0_s}-\ket{\psi^1_s})}+2K\sqrt{\delta}+2L\sqrt{\delta}\\
			&\leq2\squ*{(17+8\sqrt{2})K+L}\sqrt{\delta}
		\end{split}
		\end{align}
		We can do the same with Charlie's observables.
	\end{proof}
	
	\begin{theorem}[robust parallel-repeated rigidity]\label{thm:rob-par}
		Let $n\in\N$ and let $\ttt{S}=\parens*{\tsf{B},\tsf{C},B,C,\rho=\ketbra{\psi}}$ be a purified strategy for $\ttt{TFKW}^n$. Suppose that for some $\varepsilon\geq 0$, for each $i\in[n]$, the $i$-th game wins with probability $\mfk{w}_{\ttt{TFKW}^n}^i(\ttt{S})\geq\cos^2\tfrac{\pi}{8}-\varepsilon$. Then, there exists a constant $K\geq 0$, Hilbert spaces $\tsf{B}'$ and $\tsf{C}'$, and isometries $V:\tsf{B}\rightarrow\tsf{B}'$ and $W:\tsf{C}\rightarrow\tsf{C}'$ such that the distance between quantum states
		\begin{align}
			\norm[\Big]{(V\otimes W)\ket{\psi}-\sum_{t\in(\Z_2^2)^n}X^{t_{10}}Z^{t_{11}}\ket{\beta}\otimes\cdots\otimes X^{t_{n0}}Z^{t_{n1}}\ket{\beta}\otimes\ket{\psi_t}}\leq Kn^3\sqrt{\varepsilon},
		\end{align}
		where the $\ket{\psi_t}\in\tsf{B}'\otimes\tsf{C}'\otimes\tsf{R}$ have orthogonal supports on both $\tsf{B}'$ and $\tsf{C}'$; and there exists a constant $L\geq 0$ and commuting observables ${^{\varphi,i\!}B_\theta'}\in\mc{U}(\tsf{B}')$ and ${^{\varphi,i}C_\theta'}\in\mc{U}(\tsf{C}')$ such that
		\begin{align}
		\begin{split}
			&\norm*{V\,{^{\varphi,i\!}B_\theta}\ket{\psi}-{^{\varphi,i\!}B_\theta'}V\ket{\psi}}\leq Ln^2\sqrt{\varepsilon}\\
			&\norm*{W\,{^{\varphi,i}C_\theta}\ket{\psi}-{^{\varphi,i}C_\theta'}W\ket{\psi}}\leq Ln^2\sqrt{\varepsilon}
		\end{split}\\
			&{^{\varphi,i\!}B_\theta'}\ket{\psi_t}={^{\varphi,i}C_\theta'}\ket{\psi_t}=(-1)^{t_{i\theta}}\ket{\psi_t},
		\end{align}
		for at least one value of $\varphi$ for each $i$.
	\end{theorem}
	
	The proof below gives that we may take values $L=230\;000$, and for large enough $n$, $K=320\;000$.
	
	We make use of the fact that, as in the exact case, the $i$-th winning probability is $\mfk{w}^i_{\ttt{TFKW}^n}(\ttt{S})=\mbb{E}_{\substack{\varphi\in\Z_2^n\\\varphi_i=0}}\mfk{w}_{\ttt{TFKW}}({^{\varphi,i}\ttt{S}})$. However, since the $i$-th winning probability is not quite optimal, showing that the ${^{\varphi,i}\ttt{S}}$ win near-optimally proves to be an obstacle. To get past this, we adapt a technique of \cite{Col17} for parallel repetition of CHSH games. It guarantees that there is a ``good set'' of strategies that win with only slightly relaxed probability, and the set is large enough to continue the proof as for the exact case.
	
	\begin{proof}
		Define $\varepsilon_{\varphi,i}\geq 0$ such that $\mfk{w}_{\ttt{TFKW}}({^{\varphi,i}\ttt{S}})=\cos^2{\tfrac{\pi}{8}}-\varepsilon_{\varphi,i}$. Then, we have that, for each $i$, $\varepsilon\geq \mbb{E}_{\substack{\varphi\in\Z_2^n\\\varphi_i=0}}\varepsilon_{\varphi,i}$ We want to collect a large enough number of terms where $\varepsilon_{\varphi,i}$ is not too large with respect to $\varepsilon$. To that effect, define the set of good values of $\varphi$ for $i$ as
		\begin{align}
		G_i=\set*{\varphi\in\Z_2^n,\varphi_i=0}{\varepsilon_{\varphi,i}\leq 5\varepsilon}.
		\end{align}
		As in \cite{Col17}, we claim that $|G_i|\geq 2^{n-2}+2^{n-3}+1$. In fact, suppose $|G_i|<2^{n-2}+2^{n-3}+1$. Then, there are at least $2^{n-3}$ values of $\varphi$ where $\varepsilon_{\varphi,i}>5\varepsilon$. This gives however that
		\begin{align}
			\varepsilon\geq\frac{1}{2^{n-1}}\sum_{\varphi}\varepsilon_{\varphi,i}>\frac{1}{2^{n-1}}2^{n-3}(5\varepsilon)=\frac{5}{4}\varepsilon>\varepsilon,
		\end{align}
		which is a contradiction. Now, as for the case of a single game, for $\varphi\in G_i$, the SOS decomposition implies
		\begin{align}
			&\norm*{(Z_i\otimes{^{\varphi,i\!}B_0}+X_i\otimes{^{\varphi,i\!}B_1}-\sqrt{2})\ket{\psi}}\leq2\sqrt{5}(2+8^{1/4})\sqrt{\varepsilon}\\
			&\norm*{{^{\varphi,i\!}B_\theta}\ket{\psi}-{^{\varphi,i}C_\theta}\ket{\psi}}\leq 4\sqrt{5}\sqrt{\varepsilon}.
		\end{align}
		
		This gives the commutation of ${^{\varphi,i\!}B_0}$ and ${^{\varphi,i\!}B_1}$ with respect to $\ket*{\psi}$ as
		\begin{align}
			&\norm*{[{^{\varphi,i\!}B_0},{^{\varphi,i\!}B_1}]\ket{\psi}}\leq 2\sqrt{5}(2+\sqrt{2})(2+8^{1/4})\sqrt{\varepsilon}=:K_0\sqrt{\varepsilon}.
		\end{align}
		Now, we need commutation between operators for different values of $i$. Let $i\neq i'\in[n]$, $\theta,\theta'\in\Z_2$ and $\varphi,\varphi'\in\Z_2^n$ such that $\varphi_i=\varphi'_{i'}=0$. By the pigeonhole principle, there exists a $\chi\in (G_i+\theta1^{i})\cap(G_{i'}+\theta'1^{i'})$, so using \cref{lem:2-strat-rob} with $\delta=5\varepsilon$, there exists $Q\geq 0$ such that
		\begin{align}
		\begin{split}
			&\norm*{({^{\varphi,i\!}B_\theta}-{^{\chi+\theta1^i,i\!}B_\theta})\ket{\psi}}\leq\sqrt{5}Q\sqrt{\varepsilon}\\
			&\norm*{({^{\varphi',i'\!}B_{\theta'}}-{^{\chi+\theta'1^{i'},i'\!}B_{\theta'}})\ket{\psi}}\leq\sqrt{5}Q\sqrt{\varepsilon},
		\end{split}
		\end{align}
		and identically for Charlie's observables. Thus, knowing $\norm*{[{^{\chi+\theta1^i,i\!}B_\theta},{^{\chi+\theta'1^{i'},i'\!}B_{\theta'}}]\ket{\psi}}\leq K_0\sqrt{\varepsilon}$, we have
		\begin{align}
		\begin{split}
			\norm*{[{^{\varphi,i\!}B_{\theta}},{^{\varphi',i'\!}B_{\theta'}}]\ket*{\psi}}&\leq\norm*{\parens{{^{\varphi,i\!}B_{\theta}}\otimes{^{\varphi',i'}C_{\theta'}}-{^{\varphi',i'\!}B_{\theta'}}\otimes{^{\varphi,i}C_{\theta}}}\ket*{\psi}}+8\sqrt{5}\sqrt{\varepsilon}\\
			&\leq\norm*{\parens{{^{\chi+\theta1^{i},i}B_{\theta}}\otimes{^{\chi+\theta'1^{i'},i'}C_{\theta'}}-{^{\chi+\theta'1^{i'},i'}B_{\theta'}}\otimes{^{\chi+\theta1^{i},i}C_{\theta}}}\ket*{\psi}}\\
		&\qquad+(4\sqrt{5}Q+8\sqrt{5})\sqrt{\varepsilon}\\
		&\leq(4\sqrt{5}(Q+4)+K_0)\sqrt{\varepsilon}.
		\end{split}
		\end{align}
		Now, for any $i$, we may pick some $\varphi\in G_i$, and define ${^{i\!}B_\theta}:={^{\varphi,i\!}B_\theta}$. We have
		\begin{align}
			\norm{[{^{i\!}B_\theta},{^{i'\!}B_{\theta'}}]\ket*{\psi}}\leq(4\sqrt{5}(Q+4)+K_0)\sqrt{\varepsilon}.
		\end{align}
		Then, we use \cref{lem:z2n-gh} with $U_{i\theta}={^{i\!}B_\theta}$ and $V_{i\theta}={^{i}C_\theta}$ so $\epsilon=4\sqrt{5}\sqrt{\varepsilon}$ and $\delta=(4\sqrt{5}(Q+4)+K_0)\sqrt{\varepsilon}$ to generate an $\parens*{Ln^2\sqrt{\varepsilon},\ket*{\psi}}$-representation of $(\Z_2^2)^n$, where $L=4(4\sqrt{5}(Q+7)+K_0)$. The same holds in the same way for Charlie's observables.
		
		So, this puts us in the right place to use the Gowers-Hatami theorem again. There exist Hilbert spaces with orthogonal decompositions $\tsf{B}'=\bigoplus_{t\in(\Z_2^2)^n}\tsf{B}_t$ and $\tsf{C}'=\bigoplus_{t\in(\Z_2^2)^n}\tsf{C}_t$; isometries $V:\tsf{B}\rightarrow\tsf{B}'$ and $W:\tsf{C}\rightarrow\tsf{C}'$; and unitary observables ${^{i\!}B_\theta'}=\sum_{t\in(\Z_2^2)^n}(-1)^{t_{i\theta}}\Id_{B,t}\in\mc{U}(\tsf{B}')$ and ${^{i}C_\theta'}=\sum_{t\in(\Z_2^2)^n}(-1)^{t_{i\theta}}\Id_{C,t}\in\mc{U}(\tsf{C}')$ such that
		\begin{align}
		\begin{split}
			&\norm*{(V\,{^{i\!}B_\theta}-{^{i\!}B'_\theta}V)\ket{\psi}}=Ln^2\sqrt{\varepsilon}\\
			&\norm*{(W\,{^iC_\theta}-{^iC'_\theta}W)\ket{\psi}}=Ln^2\sqrt{\varepsilon}.
		\end{split}
		\end{align}
		
		Let $\ket*{\psi'}=(V\otimes W)\ket*{\psi}$. We can put these observables back into the original inequalities to get
		\begin{align}
		&\norm*{(Z_i\otimes{^{i\!}B'_0}+X_i\otimes{^{i\!}B'_1}-\sqrt{2})\ket{\psi'}}\leq (2Ln^2+2\sqrt{5}(2+8^{1/4}))\sqrt{\varepsilon}\\
		&\norm*{({^{i\!}B'_\theta}-{^iC'_\theta})\ket{\psi'}}\leq(2Ln^2+4\sqrt{5})\sqrt{\varepsilon}.\label{eq:type-rob-par-prime}
		\end{align}
		Since the quantum state $\ket{\psi'}\in\bigoplus_{t,t'\in(\Z_2^2)^n}\tsf{A}_1\otimes\cdots\otimes \tsf{A}_n\otimes\tsf{B}_t\otimes\tsf{C}_{t'}\otimes\tsf{R}$, we can write it as $\ket{\psi'}=\sum_{t,t'\in(\Z_2^2)^n}\ket*{v_{t,t'}}$. Using
		\begin{align}
			Z_i\otimes {^{i\!}B_0'}+X_i\otimes {^{i\!}B_1'}-\sqrt{2}=2\sqrt{2}\sum_{t}(\ketbra*{\beta_{t_i}}_i-\Id)\otimes\Id_{B,t}
		\end{align}
		and defining $\ket{v_{\beta,i}}=\sum_{t,t'\in(\Z_2^2)^n}(\ketbra*{\beta_{t_1}}_1\otimes\cdots\otimes\ketbra{\beta_{t_i}}_i)\ket*{v_{t,t'}}$, we have
		\begin{align}
		\begin{split}
			\norm*{\ket{\psi'}-\ket{v_{\beta,n}}}&\leq\sum_{i=1}^{n-1}\norm*{\ket{v_{\beta,i}}-\ket{v_{\beta,i+1}}}\\
			&\leq\frac{1}{2\sqrt{2}}\sum_{i=1}^{n-1}\norm*{(Z_i\otimes{^{i\!}B'_0}+X_i\otimes{^{i\!}B'_1}-\sqrt{2})\ket{\psi'}}\\
			&\leq \tfrac{n}{\sqrt{2}}(Ln^2+\sqrt{5}(2+8^{1/4}))\sqrt{\varepsilon}.
		\end{split}
		\end{align}
		
		On the other hand, \cref{eq:type-rob-par-prime} implies
		\begin{align}
		(2Ln^2+4\sqrt{5})\sqrt{\varepsilon}&\geq\norm[\Big]{\sum_{t,t'}((-1)^{t_{i\theta}}-(-1)^{t_{i\theta}'})\ket{v_{t,t'}}}=2\norm[\Big]{\sum_{t_{i\theta}\neq t_{i\theta}'}\ket{v_{t,t'}}}
		\end{align}
		Writing $\ket{\psi'}=\ket{v_0}+\ket{v_1}$ where $\ket{v_0}=\sum_{t\in(\Z_2^2)^n}\ket{v_{t,t}}$ and $\ket{v_1}=\sum_{t\neq t'}\ket{v_{t,t'}}$. Then,
		\begin{align}
		\norm*{\ket{v_1}}^2&=\sum_{t\neq t'}\braket*{v_{t,t'}}\leq\sum_\theta\sum_{i=1}^n\sum_{t_{i\theta}\neq t'_{i\theta}}\braket{v_{t,t'}}\leq2n\squ*{(Ln^2+2\sqrt{5})\sqrt{\varepsilon}}^2.
		\end{align}
		Now, let $\ket{v_\beta}=\sum_{t\in(\Z_2^2)^n}(\ketbra{\beta_{t_1}}\otimes\cdots\ketbra{\beta_{t_n}})\ket{v_{t,t}}$, then $\norm*{\ket*{v_{\beta}}-\ket*{v_{\beta,n}}}\leq\norm{\ket*{v_1}},$
		so
		\begin{align}
		\begin{split}
		\norm{\ket*{\psi'}-\ket*{v_\beta}}&\leq\norm{\ket*{\psi'}-\ket*{v_{\beta,n}}}+\norm{\ket*{v_{\beta,n}}-\ket*{v_{\beta}}}\\
		&\leq\squ*{\tfrac{n}{\sqrt{2}}(Ln^2+\sqrt{5}(2+8^{1/4}))+\sqrt{2n}(Ln^2+2\sqrt{5})}\sqrt{\varepsilon}.
		\end{split}
		\end{align}
		Now, $\ket{v_\beta}$ has the form we want, but it may not be normalised. Define $\ket{\phi}=\frac{\ket{v_\beta}}{\norm*{\ket{v_\beta}}}$. We have
		\begin{align}
			\norm*{\ket{\phi}-\ket{v_\beta}}=\norm*{\ket{\psi'}}-\norm*{\ket{v_\beta}}\leq\norm*{\ket{\psi'}-\ket{v_\beta}},
		\end{align}
		giving
		\begin{align}
		\norm{\ket*{\psi'}-\ket*{\phi}}\leq 2\norm{\ket*{\psi'}-\ket*{v_\beta}}\leq\sqrt{2}\squ*{n(Ln^2+\sqrt{5}(2+8^{1/4}))+2\sqrt{n}(Ln^2+2\sqrt{5})}\sqrt{\varepsilon}
		\end{align}
	\end{proof}
	We can, as in the single-round case, generalise this result slightly to a general strategy.
	
	\begin{corollary} \label{cor:gen-rob-par}
		Let $n\in\N$ and let $\ttt{S}=\parens*{\tsf{B},\tsf{C},B,C,\rho}$ be an arbitrary strategy for $\ttt{TFKW}^n$. Suppose that for some $\varepsilon\geq 0$, for each $i\in[n]$, the $i$-th game wins with probability $\mfk{w}_{\ttt{TFKW}^n}^i(\ttt{S})\geq\cos^2\tfrac{\pi}{8}-\varepsilon$. Then there exists a constant $K\geq 0$ and isometries $V:\tsf{B}\rightarrow\tsf{B}'$ and $W:\tsf{C}\rightarrow\tsf{C}'$ such that
		\begin{align}
		\norm*{(V\otimes W)\rho(V\otimes W)^\dag-\Tr_{R}(\ketbra{\phi})}_{\Tr}\leq Kn^3\sqrt{\varepsilon},
		\end{align}
		where $R$ is an auxiliary register such that \begin{align}
		\ket{\phi}=\sum_{t\in(\Z_2^2)^n}X^{t_{10}}Z^{t_{11}}\ket{\beta}\otimes\cdots\otimes X^{t_{n0}}Z^{t_{n1}}\ket{\beta}\otimes\ket{\psi_t}
		\end{align}
		for some vectors $\ket{\psi_t}\in\tsf{B}'\otimes\tsf{C}'\otimes\tsf{R}$ with orthogonal supports on both $\tsf{B}'$ and $\tsf{C}'$.
	\end{corollary}
	The proof follows the same method as \cref{cor:gen-rob}.

\subsection{Observed Statistics} \label{sec:obs-correlations}

	In any self-testing scenario, the referee cannot actually query the winning probability of the adversaries' strategy. To get around this, she may play many rounds of the game in parallel and use the players' winning statistics to approximate their winning probability. The difficulty that arises, however, is that the players' strategies need not be independent for the different rounds of the game, and therefore the information Alice receives might not be meaningful. A technique of \cite{RUV13} allows us to get around this: first, we bound the probability of winning too many of the games if enough are too far from optimal, and then find good values of the bounding constants, depending on the application, so Alice may extract information about the state.
	
	\begin{lemma}\label{lem:obs}
		Let $0<\varepsilon,\eta<1$ and let $\delta\in\R$ such that $\delta\leq\eta\varepsilon$. Let $\ttt{S}$ be a strategy for $\ttt{TFKW}^n$. Let $E\subseteq\{0,\ldots,n\}$ be the set of rounds $i$ such that $\mfk{w}_{\ttt{TFKW}^n}^i(\ttt{S})\geq \cos^2\tfrac{\pi}{8}-\varepsilon$, and let $W\in\{0,\ldots,n\}$ be the number of rounds the adversaries win. Then, if $|E|<(1-\eta)n$,
		\begin{align}
			\Pr\parens*{W\geq(\cos^2\tfrac{\pi}{8}-\delta)n}\leq e^{-2n(\eta\varepsilon-\delta)^2}.
		\end{align}
	\end{lemma}

	We can make use of this in contrapositive. That is, other than with small probability, if the adversaries win at least $(\cos^2\tfrac{\pi}{8}-\delta)n$ games, then at least $(1-\eta)n$ of the games win with near-optimal winning probability. The proof proceeds in the same way as a similar result for sequentially repeated games in \cite{RUV13}.
	
	\begin{proof}
		Write $w^\ast=\cos^2\tfrac{\pi}{8}$ for convenience. Let $W_i$ be the random variable that is $1$ if the adversaries won round $i$ and $0$ if they lost. Then, we have that $W$ is the random variable $W=\sum_{i=1}^nW_i$. Since $\Pr(W_i=1)=\mfk{w}_{\ttt{TFKW}^n}^i(\ttt{S})$, if $i\in E$, we know that $\Pr(W_i=1)\leq w^\ast$, and if $i\notin E$, $\Pr(W_i=1)\leq w^\ast-\varepsilon.$ Let $\Gamma_1,\ldots,\Gamma_n,\Lambda_1,\ldots,\Lambda_n$ be independent Bernoulli variables such that the $\Pr(\Gamma_i=1)=w^\ast$ and $\Pr(\Lambda_i=1)=w^\ast-\varepsilon$. By the above, we can couple them to the $W_i$ so that $W_i\leq\Gamma_i$ if $i\in E$ and $W_i\leq\Lambda_i$ is $i\notin E$. This implies directly that $W\leq\sum_{i\in E}\Gamma_i+\sum_{i\notin E}\Lambda_i$, so
		\begin{align}
			\Pr(W\geq (w^\ast-\delta)n)\leq\Pr\parens[\Big]{\sum_{i\in E}\Gamma_i+\sum_{i\notin E}\Lambda_i\geq(w^\ast-\delta)n}.
		\end{align}
		Since $\mathbb{E}\parens[\Big]{\sum_{i\in E}\Gamma_i+\sum_{i\notin E}\Lambda_i}=|E|w^\ast+(n-|E|)(w^\ast-\varepsilon)$, Hoeffding's inequality implies
		\begin{align}
			Pr(W\geq (w^\ast-\delta)n)\leq e^{-\frac{2}{n}((n-|E|)\varepsilon-\delta n)^2}\leq e^{-2n(\eta\varepsilon-\delta)^2}
		\end{align}
	\end{proof}

	A simple canonical choice of variables for large $n$ is $\varepsilon=\eta=\delta=\frac{1}{n^{1/4}}$, which allows Alice to test sets of $k\sim n^{1/25}$ parallel rounds, where she is able say that the state is $\sim\frac{1}{k^{1/8}}$ near-optimal for those rounds with probability exponentially close to $1$ in $k$.
	
	In view of applications, we give in \cref{thm:prob-gen-rob-par} a version of \cref{cor:gen-rob-par} where Alice has less information about the winning probabilities of the strategy. Rather than assuming that she knows that they win each round near-optimally, we will assume that Alice only knows with high probability that each round wins near-optimally. Then, we are able to ascertain the behaviour of the shared state in expectation. This will allow us to directly apply the result of \cref{lem:obs} to get conclusions about the rigidity of the state.
	
	\begin{theorem}\label{thm:prob-gen-rob-par}
		Let $n\in\N$ and let $\ttt{S}=\parens*{\tsf{B},\tsf{C},B,C,\rho}$ be a strategy for $\ttt{TFKW}^n$. Suppose that for some $\varepsilon,\eta\in[0,1]$, for each $i\in[n]$, there is a probability $1-\eta$ that the $i$-th game wins with probability $\mfk{w}_{\ttt{TFKW}^n}^i(\ttt{S})\geq\cos^2\tfrac{\pi}{8}-\varepsilon$. Then, there exists a constant $K\geq 0$, Hilbert spaces $\tsf{B}'$ and $\tsf{C}'$, and isometries $V:\tsf{B}\rightarrow\tsf{B}'$ and $W:\tsf{C}\rightarrow\tsf{C}'$ such that the expected value of the distance between quantum states
		\begin{align}
		\mathbb{E}\norm*{(V\otimes W)\rho(V\otimes W)^\dag-\Tr_{R}(\ketbra{\phi})}\leq Kn^3\sqrt{\varepsilon}+n\eta,
		\end{align}
		where $\ket{\phi}=\sum_{t\in(\Z_2^2)^n}X^{t_{10}}Z^{t_{11}}\ket{\beta}\otimes\cdots\otimes X^{t_{n0}}Z^{t_{n1}}\ket{\beta}\otimes\ket{\psi_t}$ for some auxiliary register $R$ and $\ket{\psi_t}\in\tsf{B}'\otimes\tsf{C}'\otimes\tsf{R}$ with orthogonal supports on both $\tsf{B}'$ and $\tsf{C}'$.
	\end{theorem}
	
	\begin{proof}
		For each $i\in[n]$, let $H_i$ be the random variable indicating if $\mfk{w}_{\ttt{TFKW}^n}^i(\ttt{S})\geq\cos^2\tfrac{\pi}{8}-\varepsilon$. We have $\Pr(H_i=1)=1-\eta$. Let $H\subseteq[n]$ be the register-valued random variable such that $i\in H$ if and only if $H_i=1$; let $L=[n]\backslash H$ be the complement. Since for any round in $H$, $\mfk{w}_{\ttt{TFKW}^n}^i(\ttt{S})\geq\cos^2\tfrac{\pi}{8}-\varepsilon$, we can apply the rigidity of \cref{cor:gen-rob-par} to those rounds. Then, there exists a constant $K\geq 0$, Hilbert spaces $\tsf{B}'$ and $\tsf{C}'$, isometries $V:\tsf{B}\rightarrow\tsf{B}'$ and $W:\tsf{C}\rightarrow\tsf{C}'$, and a state $\ket{\phi}\in\tsf{A}_H\otimes\tsf{A}_L\otimes\tsf{B}'\otimes\tsf{C}'\otimes\tsf{R}$ of the form $\ket{\phi}=\sum_{t\in(\Z_2^2)^{|H|}}\ket{\beta_t}_{A_H}\otimes\ket{\psi_t}\in\tsf{A}\otimes\tsf{B}'\otimes\tsf{C}'\otimes\tsf{R}$
		where the supports of the $\ket{\psi_t}\in\tsf{A}_L\otimes\tsf{B}'\otimes\tsf{C}'\otimes\tsf{R}$ on both $\tsf{B}'$ and $\tsf{C}'$ are orthogonal such that
		\begin{align}
			\norm*{(V\otimes W)\rho(V\otimes W)^\dag-\Tr_R(\ketbra{\phi})}_{\Tr}\leq K|H|^3\sqrt{\varepsilon}\leq Kn^3\sqrt{\varepsilon}.
		\end{align}
		Let $\sigma=\ketbra{\beta}^{\otimes L}_{A_L}\otimes\Tr_{A_LR}(\ketbra{\phi})$. Then, $\sigma$ has the form we want and $\norm{\Tr_R(\ketbra{\phi})-\sigma}\leq n-|H|$, giving that $\mathbb{E}\norm{\Tr_R(\ketbra{\phi})-\sigma}\leq n-\mathbb\sum_iH_i=n\eta$. Using the triangle inequality, we get the wanted result.
	\end{proof}
	
	We give an example of the use of the results of this section by considering an explicit choice of parameters.
	
	\begin{example}\label{ex:numbers}
		Fix some large $n\in\N$. Take $\varepsilon=n^{-8}$, $\eta=n^{-2}$, and $\delta=\frac{1}{2}n^{-10}$. Suppose Alice plays $N=n^{21}$ rounds of the TFKW game in parallel with Bob and Charlie, and that the players are able to win at least $\cos^2\frac{\pi}{8}N-\frac{1}{2}n^{11}$ of them. Then \cref{lem:obs} implies that, other than with probability $e^{-\frac{n}{2}}$, there are at least $(1-n^{-2})N$ rounds that won with probability $\mfk{w}_{\ttt{TFKW}^n}^i(\ttt{S})\geq\cos^2\tfrac{\pi}{8}-n^{-8}$. Then Alice can check the rigidity on $n$ rounds chosen uniformly at random: call this register $A'$. Due to the uniform randomness, each of the $n$ has probability $1-n^{-2}$ of being within $n^{-8}$ of optimal. Then, we can use the rigidity of \cref{thm:prob-gen-rob-par} to say that there exists a constant $K\geq 0$, Hilbert spaces $\tsf{B}'$ and $\tsf{C}'$, and isometries $V:\tsf{B}\rightarrow\tsf{B}'$ and $W:\tsf{C}\rightarrow\tsf{C}'$ such that the expected value of the distance between quantum states
		\begin{align}
		\mathbb{E}\norm*{(V\otimes W)\rho_{A'BCR}(V\otimes W)^\dag-\Tr_{R}(\ketbra{\phi})}\leq\frac{K+1}{n},
		\end{align}
		where $\ket{\phi}=\sum_{t\in(\Z_2^2)^n}X^{t_{10}}Z^{t_{11}}\ket{\beta}\otimes\cdots\otimes X^{t_{n0}}Z^{t_{n1}}\ket{\beta}\otimes\ket{t}_{BCR}$ for some auxiliary register $R$ and $\ket{t}_{BCR}\in\tsf{B}'\otimes\tsf{C}'\otimes\tsf{R}$ with orthogonal supports on both $\tsf{B}'$ and $\tsf{C}'$.
	\end{example}

\section{Applications}
\label{sec:application}

In this section, we present applications of our rigidity result. In \cref{sec:appl-prel}, we introduce further definitions and techniques we will need in this section. In \cref{sec:wse}, we construct a three-party weak string erasure scheme. In \cref{sec:bc}, we discuss bit commitment constructed from this weak string erasure scheme, and contrast our model with prev0ious three-party models. In \cref{sec:rand}, we construct a everlasting randomness expansion protocol in a model closely following the model for MoE games, that requires temporary computational assumptions but no entanglement.

\subsection{Preliminaries and Notation}\label{sec:appl-prel}

A \emph{classical-quantum state} (cq) is a state $\rho_{XH}\in\mc{D}(\tsf{X}\otimes\tsf{H})$ that takes the form
\begin{align}
	\rho_{XH}=\sum_{x\in X}p_x\ketbra{x}_X\otimes\rho^x_H,
\end{align}
where $p_x\geq 0$, $\sum_xp_x=1$, and the $\rho^x_H\in\mc{D}(\tsf{H})$ are quantum  states. As the part of the state on the register $X$ is diagonal in the canonical basis, we consider that as a classical register. Since classical information may be cloned, we may write the ccq state with the classical part duplicated as
\begin{align}
	\rho_{XXH}=\sum_{x\in X}p_x\ketbra{x}\otimes\ketbra{x}\otimes\rho^x_H\in\mc{D}(\tsf{X}\otimes\tsf{X}\otimes\tsf{H}).
\end{align}
If a quantum register decomposes as a product $H=H_1\times \cdots\times H_n$, for any set $\iota=\{i_1,\ldots,i_k\}\subseteq[n]$ with $i_1<\ldots<i_k$ and state $\rho_H\in\mc{D}(\tsf{H})$, write \begin{align}
	\rho_{H_\iota}=\rho_{H_{i_1}\cdots H_{i_k}}\in\mc{D}(\tsf{H}_\iota)=\mc{D}(\tsf{H}_{i_1}\otimes\cdots\otimes\tsf{H}_{i_k}).
\end{align}
Finally, for a cq state $\rho_{IH}\in\mc{D}(\tsf{I}\otimes\tsf{H})$, where $I=P([n])$ the power set, write
\begin{align}
	\rho_{IH_I}=\sum_{\iota\in I}p_\iota\ketbra{\iota}\otimes\rho^{\iota}_{H_\iota}\in\mc{D}\parens*{\bigoplus_{\iota\in I}\ket{\iota}\otimes\tsf{H}_\iota}\subseteq\mc{D}(\tsf{I}\otimes\tsf{H}).
\end{align}

We also want to be able to estimate the uncertainty of a register given another. This is done using the \emph{conditional min-entropy}: for $\rho_{HK}\in\mc{D}(\tsf{H}\otimes\tsf{K})$, the uncertainty of $H$ knowing $K$ is
\begin{align}
	H_{\mathrm{min}}(H|K)_\rho=-\lg\inf\set*{\Tr(\sigma_K)}{\rho_{HK}\leq\Id_H\otimes\sigma_K,\sigma_K\in\mc{P}(\tsf{K})}.
\end{align}
Importantly, if $\rho$ is classical on $\tsf{H}$, $2^{-H_{\mathrm{min}}(H|K)_\rho}$ corresponds exactly to the probability of guessing $H$ when holding $K$. The robust version of this entropy is the \emph{smooth min-entropy}. For $\varepsilon>0$,
\begin{align}
	H^\varepsilon_{\mathrm{min}}(H|K)_\rho=\sup\set*{H_{\mathrm{min}}(H|K)_\sigma}{\sigma\in\mc{P}(\tsf{H}\otimes\tsf{K}),\norm{\rho-\sigma}_{\Tr}\leq\Tr(\rho)\varepsilon,\Tr(\sigma)\leq\Tr(\rho)}.
\end{align}
For more information, see \cite{Tom16}.

Generally, the evolution of a quantum system is given by a quantum channel, which subsumes both measurements and unitary evolution. A quantum channel is represented by a completely positive trace-preserving (CPTP) map, which is a linear map $\Phi:\mc{L}(\tsf{H})\rightarrow\mc{L}(\tsf{K})$ such that, for any Hilbert space $\tsf{Z}$ and any $P\in\mc{P}(\tsf{Z}\otimes\tsf{H})$, $(\Id_Z\otimes\Phi)(P)\geq0$; and $\Tr(\Phi(L))=\Tr(L)$ for all $L\in\mc{L}(\tsf{H})$. The partial trace provides an example of a quantum channel.

We write the orthonormal basis that diagonalises the Pauli $Y$ matrix as $\ket{0_\circlearrowleft}=\frac{1}{\sqrt{2}}(\ket{0}+i\ket{1})$ and $\ket{1_\circlearrowleft}=\frac{1}{\sqrt{2}}(\ket{0}-i\ket{1})$. Then, $Y=\ketbra{0_\circlearrowleft}-\ketbra{1_\circlearrowleft}$.

A function $f:\N\rightarrow[0,\infty)$ is called \emph{negligible} if, for every polynomial $p$, $p(n)f(n)\rightarrow 0$ as $n\rightarrow\infty$. Write $\ttt{negl}(n)$ for the set of negligible functions in $n$; abusing notation a bit, we will also write $\ttt{negl}(n)$ to represent some function taken from the set.

\subsection{Weak String Erasure}\label{sec:wse}

Weak string erasure (WSE) is a fundamental cryptographic primitive, introduced in \cite{KWW12}. It is a simple yet powerful way to share partial information between a sender Alice and a receiver Bob. A WSE protocol provides Alice with a random string $x\in\Z_2^n$, and Bob with a string $\hat{x}\in\Z_2^n$ and a subset $\iota\subseteq [n]$ such that $|\iota|$ is on average $n/2$. The strings satisfy the property that they are equal on the positions indexed by the elements of $\iota$: $x_\iota=\hat{x}_\iota$. Security for such a scheme consists of Alice being unable to guess which substring of $x$ Bob knows, while Bob is unable to guess the remaining bits of $x$, \emph{i.e.}~the substring $x_{\iota^c}$. WSE was used in \cite{KWW12} to create bit commitment and oblivious transfer schemes. In their construction of a BC scheme, the roles of the sender and the receiver are preserved: the string Alice commits to is the image of her WSE output~$x$ by a randomness extractor, and in the reveal phase, Bob uses the part of the string he knows $x_\iota$ to verify that Alice had in fact committed to this string. Hence, however, since bit commitment is impossible with no additional assumptions \cite{BS16}, WSE needs some assumptions about the model to hold. Accordingly, \cite{KWW12} used a quantum noisy-storage model to achieve it, generalising results on bounded quantum storage used to achieve oblivious transfer \cite{DFSS08}. We formally define security for WSE in the original two-party model.

\begin{definition}[\cite{KWW12}]\label{def:wse}
	A \emph{$(n,\lambda,\varepsilon)$-weak string erasure (WSE) scheme} is a protocol between two parties, Alice and Bob, that creates a state $\rho_{XAI\hat{X}B}$, where $X$, $I$, and $\hat{X}$ are classical registers such that $X=\Z_2^n$ holds string $x$, $\hat{X}=\Z_2^n$ holds Bob's guess of $x$, and $I=P([n])$ holds $\iota$; and $A$ and $B$ are optional quantum registers corresponding to Alice and Bob's remaining quantum states. The scheme must satisfy \emph{correctness}, and \emph{security} for both Alice and Bob:
	
	\textbf{Correctness}: If both Alice and Bob are honest then $\rho_{XI\hat{X}_I}=\rho_{XIX_I}$ and $\rho_{XI}=\mu_{X}\otimes\mu_{I}$.
	
	\textbf{Security for Alice}: If Alice is honest $H^\varepsilon_{\mathrm{min}}(X|B)_\rho\geq\lambda n$.
	
	\textbf{Security for Bob}: If Bob is honest, $\rho_{AI}=\rho_{A}\otimes\mu_{I}$ in the event that Alice does not abort.
\end{definition}

\noindent
We say that a protocol is a \emph{$(n,\lambda,\varepsilon)$-WSE scheme that fails with probability $p$} if any one of the three conditions does not hold with probability at most $p$.

Here, we show WSE in an alternative model. Instead of resorting to limitations on the quantum devices of parties, we add an additional dishonest prover, Charlie, who colludes with the receiver.\footnote{Note that oblivious transfer in yet another three-party model has been considered before. In \cite{YXTZ14}, they consider a model where an untrusted third party prepares entangled states for Alice and Bob to use. However, they make use of much stronger assumptions: the third party produces each state identically and independently, Alice and Bob need to cooperate to verify that these states are correct before running the protocol, and the third party does not collude with any of the other parties.} Instead of the storage limitation, we place restrictions on the communications: the prover is not allowed to communicate with the receiver, and the sender is required to communicate by publicly broadcasting. The former restriction helps an honest sender constrain the action of a dishonest receiver; the latter condition blocks a subtle cheating method of a dishonest sender, where she attempts to extract different information from the receiver and the prover.

\begin{definition}
	A WSE scheme in the \emph{three-party model} consists of a sender, Alice, a receiver, Bob, and a prover, Charlie. It satisfies the following:
	\begin{itemize}
		\item Charlie is dishonest if and only if Bob is dishonest.
		
		\item Alice communicates by publicly broadcasting.
		
		\item Bob and Charlie are isolated from each other once Alice starts broadcasting.
	\end{itemize}
\end{definition}

In this model, there is an additional prover Charlie, so the state takes the form $\rho_{XAI\hat{X}BC}$, where~$C$ is an additional register held by Charlie. If he is dishonest, he should not be able to get more information out of the protocol than his collaborator, Bob. Thus, we require that the security for Alice from \cref{def:wse} is satisfied with respect to either Bob or Charlie's registers. We state this formally.

\begin{definition}
	A \emph{$(n,\lambda,\varepsilon)$-WSE scheme in the three-party model} is a protocol that produces a shared state $\rho_{XAI\hat{X}BC}$ such that it is a two-party scheme for Alice and Bob, and the security for Alice is symmetric:
	
	\textbf{Two-party WSE:} $\rho_{XAI\hat{X}B}$, the state with Charlie's register traced out, satisfies \cref{def:wse}.
	
	\textbf{Symmetric security for Alice}: If Alice is honest $H^\varepsilon_{\mathrm{min}}(X|C)\geq\lambda n$.
\end{definition}

For our protocol in this model, as Bob colludes with Charlie but may not communicate with him, an honest sender exploits the rigidity of the TFKW game to constrain their actions. Note that Charlie needs to remain out of the reach of Bob's communications for as long as Bob is using his output data in order for it to stay secure. Since Alice must broadcast publicly, Bob and Charlie will receive the same TFKW game questions even if she is dishonest.

Now, we formally present our protocol.

\begin{mdframed}
\begin{protocol}[three-party weak string erasure]\label{prot:wse}\hphantom{ }
	\begin{enumerate}[1.]
		\item Bob prepares the shared state $\be^{\otimes N}\ket{x^\varphi}_A\otimes\ket{x\varphi}_B\otimes\ket{x\varphi}_C$ for $x,\varphi\in\Z_2^N$ chosen uniformly at random. Bob and Charlie are then no longer allowed to communicate.
		
		\item Alice chooses a set of $n$ indices $J\subseteq[N]$ and a string $\theta\in\Z_2^N$ uniformly at random. She measures each of her qubits $1\leq i\leq N$ in basis $\{\ket{0^{\theta_i}},\ket{1^{\theta_i}}\}$ if $i\notin J$ and in basis $\{\be\ket{0^{\theta_i}},\be\ket{1^{\theta_i}}\}$ if $i\in J$. This produces a string $y\in\Z_2^N$ that she keeps; and she broadcasts $J$ and $\theta$.
		
		\item Bob and Charlie, without communicating, each measure their subspaces to get strings corresponding to their optimal guess at the TFKW game on $J^c$, $y^B=y^C=x+1^{J^c}\land\theta\land\varphi\in\Z_2^N$, and they then send $y^B_{J^c}$ and $y^C_{J^c}$, respectively, to Alice.
		
		\item  Alice checks if her string everywhere but the index set, $y_{J^c}$, matches $y^B_{J^c}$ and $y^C_{J^c}$ simultaneously on at least $(\cos^2\frac{\pi}{8}-\delta)N$ bits. If it does not, she aborts.
		
		\item Alice takes as output $y_J$, and Bob takes as output the set $\iota(\theta,\varphi)\subseteq J$ where the bits of $\theta$ and $\varphi$ match, and the string $y^B_J$.
	\end{enumerate}
\end{protocol}
\end{mdframed}

In \cref{fig:wse}, we give a setup for our WSE scheme, where we see that the single round of communication makes it possible to devise a  way to run the protocol relativistically.

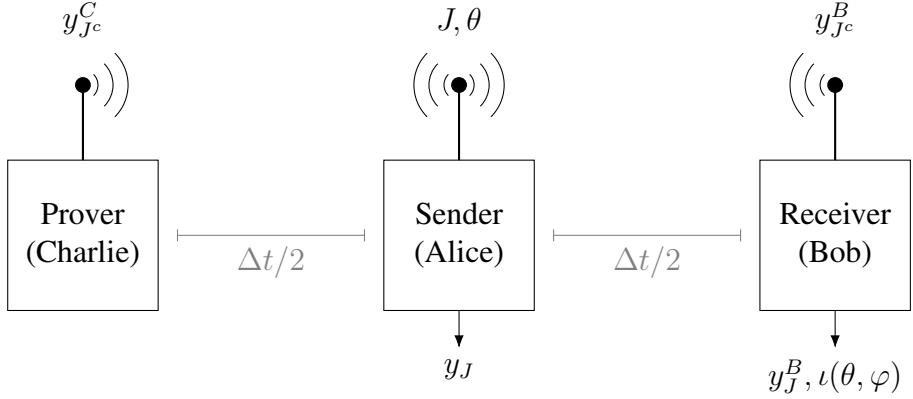
\begin{figure}
	\centering
	\begin{tikzpicture}
	\draw (-1,-1) rectangle (1,1);
	\node[align=center] at (0,0){Sender \\ (Alice)};
	\draw[thick] (0,1) -- (0,2);
	\fill (0,2) circle (3pt);
	\draw (0.2,2) arc (0:45:0.2);
	\draw (0.2,2) arc (0:-45:0.2);
	\draw (0.4,2) arc (0:45:0.4);
	\draw (0.4,2) arc (0:-45:0.4);
	\draw (0.6,2) arc (0:45:0.6);
	\draw (0.6,2) arc (0:-45:0.6);
	\draw (-0.2,2) arc (180:135:0.2);
	\draw (-0.2,2) arc (180:225:0.2);
	\draw (-0.4,2) arc (180:135:0.4);
	\draw (-0.4,2) arc (180:225:0.4);
	\draw (-0.6,2) arc (180:135:0.6);
	\draw (-0.6,2) arc (180:225:0.6);
	\node[above] at (0,2.5){$J,\theta$};
	
	\draw (-6,-1) rectangle (-4,1);
	\node[align=center] at (-5,0){Prover \\ (Charlie)};
	\draw[thick] (-5,1) -- (-5,2);
	\fill (-5,2) circle (3pt);
	\draw (-4.8,2) arc (0:45:0.2);
	\draw (-4.8,2) arc (0:-45:0.2);
	\draw (-4.6,2) arc (0:45:0.4);
	\draw (-4.6,2) arc (0:-45:0.4);
	\draw (-4.4,2) arc (0:45:0.6);
	\draw (-4.4,2) arc (0:-45:0.6);
	\node[above] at (-5,2.5){$y^C_{J^c}$};
	
	\draw (4,-1) rectangle (6,1);
	\node[align=center] at (5,0){Receiver \\ (Bob)};
	\draw[thick] (5,1) -- (5,2);
	\fill (5,2) circle (3pt);
	\draw (4.8,2) arc (180:135:0.2);
	\draw (4.8,2) arc (180:225:0.2);
	\draw (4.6,2) arc (180:135:0.4);
	\draw (4.6,2) arc (180:225:0.4);
	\draw (4.4,2) arc (180:135:0.6);
	\draw (4.4,2) arc (180:225:0.6);
	\node[above] at (5,2.5){$y^B_{J^c}$};
	
	\draw[-Latex] (0,-1) -- (0,-1.5) node[below]{$y_J$};
	\draw[-Latex] (5,-1) -- (5,-1.5) node[below]{$y^B_J,\iota(\theta,\varphi)$};
	
	\draw[gray, |-|] (-3.75, 0) -- (-2.5,0) node[below]{$\Delta t/2$} -- (-1.25,0);
	\draw[gray, |-|] (3.75, 0) -- (2.5,0) node[below]{$\Delta t/2$} -- (1.25,0);
	
	\end{tikzpicture}
	\caption{A setup for the three-party weak string erasure scheme \cref{prot:wse}. The parties communicate by broadcasting publicly, and as there is only one round of communication, the spacelike separation between Bob and Charlie ensures they do not communicate in the time $\Delta t$ needed to run the protocol.}
	\label{fig:wse}
\end{figure}

Note that the protocol requires no quantum storage to run honestly by considering it in a prepare-and-measure way. That is, Alice may come up with the random $\theta$ and $J$ before Bob prepares the state, measure her register one qubit at a time as soon she receives it from Bob, and only reveal $\theta$ and $J$ once she knows that Bob and Charlie are no longer communicating. Since her measurements are local, this has the same effect on the state as if she waited until Bob and Charlie finish communicating to make her measurements.

To illustrate the protocol, we consider in more detail the case where the players are honest. Bob should prepare some unentangled optimal state for the TFKW game uniformly at random:
\begin{align}
\rho_{ABC}&=\expec{x,\varphi\in\Z_2^{N}}\be^{\otimes N}\ketbra{x^\varphi}_A\be^{\otimes N}\otimes\ketbra{x\varphi}_B\otimes\ketbra{x\varphi}_C.
\end{align}
Alice then comes up with uniformly random $\theta$ and $J$ and makes her measurements, so the state becomes
\begin{align}
\rho_{YBC}&=\expec{x,\varphi,\theta,J}\sum_{y}\abs*{\braket{y^\theta}{\be^{1^{J^c}}}{x^\varphi}}^2\ketbra{y}_Y\otimes\ketbra{x\varphi\theta J}_B\otimes\ketbra{x\varphi\theta J}_C.
\end{align}
Now, the honest Bob measures his register in the computational basis and gets full information about $x$ and $\varphi$. Knowing $J$, Bob sends $y^B_{J^c}=x_{J^c}+\theta_{J^c}\land\varphi_{J^c}$ corresponding to his best guess of $y$ in the TFKW game, and keeps $y^B_{J}=x_J$ to himself. Charlie does the same and sends $y^C_{J^c}=y^B_{J^c}$ to Alice. Bob and Charlie win at each copy of $\ttt{TFKW}$ with probability exactly $\cos^2\frac{\pi}{8}$, so with overwhelming probability for large $N$, they do not cause Alice to abort. Assuming the protocol does not abort, Bob defines $\iota(\theta,\varphi)\subseteq(\Z_2^N)_J=\Z_2^n$ as the set $\iota(\theta,\varphi)=\set{i\in J}{\theta_i=\varphi_i}$. Alice and Bob have no use for $x_{J^c}$, $\varphi$, $y_{J^c}$, $\theta$, and $J$ and may forget them. Alice calls her remaining register~$X$ and Bob calls his remaining registers $I$ and $\hat{X}$, so the state becomes
\begin{align}
\rho_{XI\hat{X}C}=\expec{x,\varphi,\theta,J}\sum_{y_J}\abs*{\braket{y_J}{H^{\theta_J+\varphi_J}}{x_J}}^2\ketbra{y_J}_X\otimes\ketbra{\iota(\theta,\varphi)}_I\otimes\ketbra{x_J}_{\hat{X}}\otimes\rho_C^{x,\varphi,\theta,J},
\end{align}
where each $\rho_C^{x,\varphi,\theta,J}$ is a quantum state representing what Charlie continues to hold, but the structure of this state is unimportant. From the coefficients $\abs*{\braket{y_J}{H^{\theta_J+\varphi_J}}{x_J}}^2=\abs*{\braket{y_J}{H^{1^{\iota(\theta,\varphi)^c}}}{x_J}}^2$, we see that $(y_J)_i=(x_J)_i$ for $i\in\iota(\theta,\varphi)$ while for $i\notin\iota(\theta,\varphi)$, $(x_J)_i$ is uniformly random with respect to $(y_J)_i$. Therefore, Alice holds the string $y_J$, Bob has the substring $(y_J)_{\iota(\theta,\varphi)}$ and full information about where in the string they are found, but Bob has no information about the remaining bits. Formally, this gives correctness of the protocol.


\begin{lemma}
	\cref{prot:wse} is correct as a WSE scheme.
\end{lemma}

\begin{proof}

We need to show that $\rho_{XI}=\mu_X\otimes\mu_I$ and $\rho_{XIX_I}=\rho_{XI\hat{X}_I}$ for honest Alice and Bob. By the above argument,
\begin{align}
	\rho_{XI\hat{X}}=\expec{x,\varphi,\theta,J}\sum_{\substack{y_J\\(y_J)_{\iota(\theta,\varphi)}=\\(x_J)_{\iota(\theta,\varphi)}}}\frac{1}{2^{|\theta_J+\varphi_J|}}\ketbra{y_J}_X\otimes\ketbra{\iota(\theta,\varphi)}_I\otimes\ketbra{x_J}_{\hat{X}},
\end{align}
and therefore
\begin{align}
\rho_{XI}=\expec{y_J,\varphi,\theta,J}\ketbra{y_J}_X\otimes\ketbra{\iota(\theta,\varphi)}_I=\mu_X\otimes\mu_I.
\end{align}
This gives that the bit string $x$ and the subset $\iota$ are both uniformly random. We also want Bob's substring of $x$ to be correct. For this, padding Bob's space implicitly to keep every term the same dimension,
\begin{align}
\begin{split}
\rho_{XI\hat{X}_I}&=\expec{x,\varphi,\theta,J}\sum_{\substack{y_J\\(y_J)_{\iota(\theta,\varphi)}=\\(x_J)_{\iota(\theta,\varphi)}}}\frac{1}{2^{|\theta_J+\varphi_J|}}\ketbra{y_J}_X\otimes\ketbra{\iota(\theta,\varphi)}_I\otimes\ketbra{(x_J)_{\iota(\theta,\varphi)}}_{\hat{X}_I}\\
&=\expec{x,\varphi,\theta,J}\sum_{\substack{y_J\\(y_J)_{\iota(\theta,\varphi)}=\\(x_J)_{\iota(\theta,\varphi)}}}\frac{1}{2^{|\theta_J+\varphi_J|}}\ketbra{y_J}_X\otimes\ketbra{\iota(\theta,\varphi)}_I\otimes\ketbra{(y_J)_{\iota(\theta,\varphi)}}_{\hat{X}_I}\\
&=\rho_{XIX_I}. \qedhere
\end{split}
\end{align}
\end{proof}

We now show security.

\begin{theorem} \label{thm:wse-sec}
	Let $K$ be the constant from \cref{thm:prob-gen-rob-par}. For any $N,n\in\N$ and $\varepsilon,\eta,\delta\in(0,1)$ such that $\eta\varepsilon>\delta$, \cref{prot:wse} is a $(n,\lg(\frac{4}{3}),Kn^3\sqrt{\varepsilon}+n\eta)$-WSE scheme that fails with probability $e^{-2N(\eta\varepsilon-\delta)^2}$.
\end{theorem}

For example, taking the parameter values from \cref{ex:numbers} gives exponentially small failure probability, and requires only polynomially many qubits to run.

\begin{proof}
	First, we show security for Bob. This is essentially because an honest Bob provides Alice no information about any of his strings on $J$. Bob, as he is honest, prepares the shared state
	\begin{align}
		\rho_{ABC}=\expec{x,\varphi\in\Z_2^{N}}\be^{\otimes N}\ketbra{x^\varphi}_A\be^{\otimes N}\otimes\ketbra{x\varphi}_B\otimes\ketbra{x\varphi}_C.
	\end{align}
	Alice can do anything to her state but she must send Bob and Charlie $J$ and $\theta$. Note that Bob and Charlie must both receive the same pair by hypothesis. Therefore, as Alice must get $y,\theta,J$ by some channel $\Phi:\mc{L}(\tsf{A})\rightarrow\mc{L}(\tsf{Y}\otimes\tsf{A}')$,
	\begin{align}
	\begin{split}
		\rho_{YA'BC}=\expec{x,\varphi}\sum_{y,\theta,J}&\ketbra{y\theta J}_Y\otimes\braket{y\theta J}{_Y\Phi(\be^{\otimes N}\ketbra{x^\varphi}_A\be^{\otimes N})_{YA'}}{y\theta J}_Y\\
		&\otimes\ketbra{x\varphi\theta J}_B\otimes\ketbra{x\varphi\theta J}_C.
	\end{split}
	\end{align}
	Since Bob is honest, Alice knows that he must provide her with $y_{J^c}^B=x_{J^c}+\theta_{J^c}\land\varphi_{J^c}$ and Charlie provides her with the same. If Alice chooses not to abort, Bob produces $\iota(\theta,\varphi)$ so the state becomes
	\begin{align}
	\begin{split}
		\rho_{Y'A'I\hat{X}C}=\expec{x,\varphi}\sum_{y,\theta,J}&\ketbra{y\theta J y_{J^c}^B}_{Y'}\otimes\braket{y\theta J}{_Y\Phi(\be^{\otimes N}\ketbra{x^\varphi}_A\be^{\otimes N})_{YA'}}{y\theta J}_Y\\
		&\otimes\ketbra{\iota(\theta,\varphi)}_I\otimes\ketbra{x_J}_{\hat{X}}\otimes\rho^{x,\varphi,\theta,J}_C.
	\end{split}
	\end{align}
	From this state and the definition of $\iota(\theta,\varphi)$ as the set of indices where $\theta_J$ and $\varphi_J$ match for honest Bob, in order for Alice to guess $\iota$, she needs to guess $\varphi_J$. Since she has no information about $x_J$ she may not do better than uniformly random. Formally,
	\begin{align}
	\begin{split}
	\rho_{Y'A'I}=\expec{x,\varphi}\sum_{y,\theta,J}&\ketbra{y\theta J y^B_{J^c}}_{Y'}\otimes\braket{y\theta J}{_Y\Phi(\be^{\otimes N}\ketbra{x^\varphi}_A\be^{\otimes N})_{YA'}}{y\theta J}_Y\\
	&\otimes\ketbra{\iota(\theta,\varphi)}_I\\
	=\sum_{y,\theta,J}\expec{x_{J^c},\varphi_{J^c}}&\ketbra{y\theta J y^B_{J^c}}_{Y'}\otimes\braket{y\theta J}{_Y\Phi(\be^{\otimes|J^c|}\ketbra{x_{J^c}^{\varphi_{J^c}}}_{A_{J^c}}\be^{\otimes|J^c|}\otimes\mu_{A_J})_{YA'}}{y\theta J}_Y\\
	&\otimes\expec{\varphi_J}\ketbra{\iota(\theta,\varphi)}_I\\
	&=\rho_{Y'A'}\otimes\mu_I,
	\end{split}
	\end{align}
	which implies that, since Alice's actions are local, any action she may do on her space gives rise to an uncorrelated final state $\rho_{AI}=\rho_{A}\otimes\mu_{I}$.
	
	Now, we study security for Alice. That is, Alice is honest but Bob and Charlie are dishonest and colluding. We want to show that $H^{Kn^3\sqrt{\varepsilon}+n\eta}_{\mathrm{min}}(X|B)_\rho\geq-\ln\parens*{\frac{3}{4}}n$. As Bob is dishonest, for the first step of the protocol, he may produce any shared state $\rho_{ABC}$. The next three steps of the protocol consist of Alice playing $N-n$ TFKW games in parallel with Bob and Charlie, and verifying the rigidity condition. Therefore, if Alice does not abort, she knows by \cref{lem:obs} that, with probability $1-e^{-2N(\eta\varepsilon-\delta)^2}$, there are at least $(1-\eta)N$ games that win with probability greater than $\cos^2\frac{\pi}{8}-\varepsilon$. We can apply the rigidity from \cref{thm:prob-gen-rob-par} to get that there exists a constant $K\geq 0$, isometries $V:\tsf{B}\rightarrow\tsf{B}'$ and $W:\tsf{C}\rightarrow\tsf{C}'$, an auxiliary register $R$, and a state $\ket{\phi}=\sum_{x,\varphi\in\Z_2^n}\be^{\otimes n}\ket{x^\varphi}\otimes\ket{x,\varphi}_{BCR}$ where the $\ket{x,\varphi}_{BCR}\in\tsf{B}'\otimes\tsf{C}'\otimes\tsf{R}$ have orthogonal support on both $\tsf{B}'$ and $\tsf{C}'$ such that
	\begin{align}
		\expec{J}\norm*{(V\otimes W)\rho_{A_JBC}(V\otimes W)^\dag-\Tr_R(\ketbra{\phi})}_{\Tr}\leq Kn^3\sqrt{\varepsilon}+n\eta.
	\end{align}
	Let $\sigma_{A_JBC}=\Tr_R(\ketbra{\phi})$ and we study first what happens if the shared state is $\sigma$. Since Bob and Charlie may not communicate and Charlie provides no additional information in the protocol, we may safely trace out Charlie's state. However, we must include the copy of $\theta$ that Bob gets during the protocol. By the orthogonality of Charlie's state's support from the rigidity theorem,
	\begin{align}
		\sigma_{A_J\Theta B}=\expec{\theta}\sum_{x,\varphi}\be^{\otimes n}\ketbra{x^\varphi}\be^{\otimes n}\otimes\ketbra{\theta}_\Theta\otimes\Tr_{CR}(\ketbra{x,\varphi})_B.
	\end{align}
	Alice's measurement gives her $X$ and makes the state
	\begin{align}
	\begin{split}
		\sigma_{X\Theta B}&=\expec{\theta}\sum_{x,y,\varphi}\abs*{\braket{y}{H^{\theta_J+\varphi}}{x}}^2\ketbra{y}_X\otimes\ketbra{\theta}_\Theta\otimes\Tr_{CR}(\ketbra{x,\varphi})_B\\
		&=\expec{\theta}\sum_{\substack{x,y,\varphi\\x_{\iota(\theta,\varphi)=y_{\iota(\theta,\varphi)}}}}\frac{1}{2^{|\theta_J+\varphi|}}\ketbra{y}_X\otimes\ketbra{\theta}_\Theta\otimes\Tr_{CR}(\ketbra{x,\varphi})_B,
	\end{split}
	\end{align}
	where $\iota(\theta,\varphi)$ is defined as before. Noting that Bob's register is uncorrelated with part of Alice's register $X_{\iota(\theta,\varphi)^c}$, that gives that Bob's probability of guessing any bit in that register is $\frac{1}{2}$. So, for fixed $\theta,\varphi$, Bob's probability of guessing $X$ is $\frac{1}{2^{|\iota(\theta,\varphi)^c|}}=\frac{1}{2^{|\theta_J+\varphi|}}$. Since $\theta_J$ is uniformly random, Bob's average-case probability of guessing $X$ for fixed $\varphi$ is,
	\begin{align}
		\expec{\theta}\frac{1}{2^{|\theta_J+\varphi|}}=\frac{1}{2^n}\sum_{k=0}^n\binom{n}{k}\frac{1}{2^k}=\parens*{\frac{3}{4}}^n.
	\end{align}
	Since this has no dependence on $\varphi$, we see that this is Bob's probability of guessing $y$, and so the min-entropy is $H_{\mathrm{min}}(X|B)_\sigma\geq -\lg(\frac{3}{4})n$, where we consider $\Theta$ as part of Bob's register $B$. Now we relate this to the smoothed min-entropy of $\rho$. Since $V\otimes W$ is an isometry, $H^{Kn^3\sqrt{\varepsilon}+n\eta}_{\mathrm{min}}(X|B)_\rho=H^{Kn^3\sqrt{\varepsilon}+n\eta}_{\mathrm{min}}(X|B)_{(V\otimes W)\rho(V\otimes W)^\dag}$, and $\sigma$ belongs to a $Kn^3\sqrt{\varepsilon}+n\eta$-ball around $(V\otimes W)\rho(V\otimes W)^\dag$, so
	\begin{align}
		H^{Kn^3\sqrt{\varepsilon}+n\eta}_{\mathrm{min}}(X|B)_\rho\geq H_{\mathrm{min}}(X|B)_\sigma\geq-\lg\parens*{\frac{3}{4}}n.
	\end{align}
	Note that this holds in the same way for Charlie, so he cannot extract any more information that Bob can if he is dishonest.
\end{proof}

\subsection{Bit Commitment from WSE}\label{sec:bc}

Bit commitment (BC) is a cryptographic primitive where a sender, Alice, sends an encoded bit (or more generally a bit string) to a receiver, Bob, and may choose to reveal it at a later time. Accordingly, a scheme for BC consists of a $\ttt{commit}$ protocol and a $\ttt{reveal}$ protocol. Ideally, the scheme should be \emph{hiding} --- Bob is unable to learn the bit until Alice chooses to reveal it --- and \emph{binding} --- Alice must reveal the same bit that she originally chose. In \cite{KWW12}, they consider a BC scheme where Alice commits to a random bit string rather then one she chooses freely. We formally define such a scheme in essentially the same way they do.

\begin{definition}
	A \emph{$(\ell,\varepsilon)$-randomised bit string commitment (RBC) scheme} is a pair of protocols between two parties Alice and Bob: a protocol $\ttt{commit}$ that creates a state $\rho_{YAB}$ and a protocol $\ttt{reveal}$ that creates from this a state $\rho_{YA'\hat{Y}FB'}$. Here, $Y=\Z_2^\ell$ holds Alice's committed string; $\hat{Y}=\Z_2^\ell$ holds the string Alice reveals; $F=\Z_2$ indicates whether Bob accepts (1) or rejects (0) the reveal; and $A,A'$ and $B,B'$ are additional quantum registers for Alice and Bob, respectively. The scheme must be \emph{correct}, \emph{$\varepsilon$-hiding}, and \emph{$\varepsilon$-binding}:
	
	\textbf{Correctness}: If Alice and Bob are honest, for $\sigma_{YF}=\mu_{Y}\otimes\ketbra{1}$, $\norm{\rho_{Y\hat{Y}F}-\sigma_{YYF}}_{\Tr}\leq\varepsilon$.
	
	\textbf{$\varepsilon$-hiding}: If Alice is honest, $\norm{\rho_{YB}-\mu_Y\otimes\rho_B}\leq\varepsilon$.
	
	\textbf{$\varepsilon$-binding}: If Bob is honest, there exists a state $\sigma_{YAB}$ where $\norm{\rho_{YAB}-\sigma_{YAB}}_{\Tr}\leq\varepsilon$ such that, applying $\ttt{reveal}$ to it to get $\sigma_{YA'\hat{Y}FB'}$, $\Pr\squ{Y\neq\hat{Y}\land F=1}\leq\varepsilon$.
\end{definition}
\noindent As for WSE, we say this scheme \emph{fails with probability $p$} if any one of these conditions does not hold with probability at most $p$.

In \cite{KWW12}, they provide a way to construct an RBC scheme using a weak string erasure scheme and a linear code. The roles of the sender and the receiver from the WSE scheme are preserved. In particular, with an $(n,\lambda,\varepsilon)$-WSE scheme and an $(n,k,d)$-linear code, they construct a $\parens*{\lambda n-(n-k)-d,2\varepsilon+2^{-d/2}}$-randomised bit string commitment scheme. Using this recipe, our WSE scheme \cref{prot:wse} gives a form of bit commitment in a model with two receivers. In this model, Alice is a sender who is required to broadcast, and Bob and Charlie are colluding receivers who are isolated from each other. Similarly to WSE, we only require that Bob be able to read the revealed bit string, rather than both receivers. We call this the \emph{two-receiver model}.

\begin{corollary}
	Let $K,N,n,\varepsilon,\eta,\delta$ be constants that satisfy \cref{thm:wse-sec}, and let $k,d\in\N$ such that there exists an $(n,k,d)$-linear code. Then, for $\ell=(\lg\tfrac{4}{3})n-(n-k)-d$ and $\omega=2Kn^3\sqrt{\varepsilon}+2n\eta+2^{-d/2}$, in the two-receiver model, there exists a $\parens*{\ell,\omega}$-randomised bit string commitment scheme that fails with probability $e^{-2N(\eta\varepsilon-\delta)^2}$.
\end{corollary}

Using the construction of \cite{KWW12}, the correctness and $\varepsilon$-binding of the scheme between Alice and Bob follow immediately from the correctness and security for Bob of WSE. Also, due to the symmetry requirement on security for Alice in three-party WSE, this construction provides $\varepsilon$-hiding when Bob and Charlie are dishonest.

As mentioned before, a construction of \cite{BGKW88} provides classical bit commitment in a model with two senders who may not communicate. We observe that, in contrast, bit commitment is  classically impossible in our two-receiver model. The first step of a protocol in our model consists of the preparation of an initial shared state by Bob. If only classical operations are allowed, Bob is just sampling from a probability distribution and sharing the result. In particular, he can make sure that all three parties receive the same classical information. Next, for the remainder of the protocol, Bob and Charlie may not communicate. However, since Alice must communicate by publicly broadcasting, Bob and Charlie receive exactly the same information from her, and may respond to all the same challenges. As such, classically, our model becomes equivalent to the standard two-party model. In particular, bit commitment is impossible. The difference with the two-sender model, where bit commitment does exist classically, arises due to the additional communication restriction we imposed: the receiver of \cite{BGKW88} may share different information with each of the senders, rather than broadcasting publicly.

\subsection{Everlasting Randomness Expansion}\label{sec:rand}

The creation of fresh, trusted, uniform randomness is an important part of many computational and cryptographic tasks. Since quantum mechanics is inherently probabilistic, it stands to reason that quantum procedures prove useful for this task. A major theoretical hurdle in achieving this is that it is difficult to characterise the behaviour of an untrusted quantum device: one needs to verify that their source of randomness is truly random and not shared by an eavesdropper. Largely, the methods to bypass this difficulty use a nonlocal game to verify entanglement between two untrusted provers. However, this requires, in particular, that the provers are able to produce entangled states, and keep them from decohering throughout the running time of the protocol. This can be an impractical requirement.

In our contribution, we remove the need for long-distance entanglement, and instead make the assumption of a trusted but leaky measurement, as well as a standard computational assumption.
The protocol consists of two steps: first, Alice samples the output of a pseudorandom generator, allowing her to increase the size of her random string; then, she uses this as the source of randomness to play the TFKW game against computationally-bounded and isolated Bob and Charlie, where the rigidity allows her to extract a string that is uniformly random and unknown to either Bob or Charlie. First, we need to formalise the model we are working in, based on the structure of an MoE game.

\begin{definition}
	The \emph{MoE model for randomness expansion} consists of three quantum parties: a trusted verifier Alice, who interacts with two untrusted provers, Bob and Charlie. The model satisfies the following:
	\begin{itemize}
		\item Bob and Charlie are able to prepare a tripartite shared state but then are isolated.
		
		\item Alice can make trusted measurements on her register, which are leaky in the sense that Bob and Charlie can learn the measurement bases.
	\end{itemize}
\end{definition}

Now, we can define randomness expanders in this model.

\begin{definition}
	A \emph{$(s(n),\varepsilon)$-local randomness expander} is a protocol in the MoE model, where, given a uniformly random seed in $S=\Z_2^{s(n)}$, Alice, Bob, and Charlie construct a quantum state $\rho_{YSBC}$, where $Y=\Z_2^n$ and $S$ are classical registers that Alice holds and $B$ and $C$ are potentially quantum registers that Bob and Charlie hold, respectively, such that
	\begin{align}
	\begin{split}
	&\norm{\rho_{YSB}-\mu_Y\otimes\mu_S\otimes\rho_B}_{\Tr}\leq\varepsilon\\
	&\norm{\rho_{YSC}-\mu_Y\otimes\mu_S\otimes\rho_C}_{\Tr}\leq\varepsilon,
	\end{split}
	\end{align}
	if Alice does not abort during the execution. As before, we say this scheme \emph{fails with probability~$p$} if these conditions do not hold with probability at most $p$.
\end{definition}

The idea of this definition is that Alice's output needs to be approximately uniformly random in any case, but we can also get the additional guarantee that, as long as Bob and Charlie stay isolated, they cannot hold onto side information that allows them to guess the output. However, we do not constrain their ability to guess the output if they come back together: for example, the register $BC$ could be maximally entangled with Alice's register before she makes her final measurement, without either $B$ or $C$ being maximally entangled on their own.

The main computational tool we will be making use of is the idea of a pseudorandom generator against computationally-bounded adversaries.

\begin{definition}
	An algorithm $Q:\Z_2^\ast\rightarrow\Z_2^\ast$ is \emph{quantum polynomial time} (QPT) if there exists a Turing machine $T$ such that, for each $n\in\N$, $T(n)$ outputs in polynomial time the description of a quantum circuit that, on input $x\in\Z_2^n$, outputs $Q(x)$. Similarly, we can consider a family of states $\rho_n$, $n\in\N$, QPT if $T(n)$ outputs a quantum circuit that constructs $\rho_n$ from $\ket{0}$; a family of unitaries $U_n$ QPT if $T(n)$ provides a quantum circuit that acts as $U_n$; and a family of measurements $A_n$ QPT if the measurement $A_n$ can be undertaken by first acting by some QPT unitary $U_n$ and then measuring in the computational basis.
\end{definition}

Now, we can introduce pseudorandom generators as functions that take a uniformly random string to a longer string that no QPT algorithm can distinguish from uniform.

\begin{definition}
	A family of functions $G_n:\Z_2^{s(n)}\rightarrow\Z_2^{n}$ is a \emph{pseudorandom generator} (PRG) if, for uniform random variables $\Gamma$ in $\Z_2^{s(n)}$ and $\Delta$ in $\Z_2^n$, and for every QPT algorithm $Q:\Z_2^\ast\rightarrow\Z_2$,
	\begin{align}
		\abs[\Big]{\Pr\squ*{Q(G_n(\Gamma))=1}-\Pr\squ*{Q(\Delta)=1}}\in\ttt{negl}(n).
	\end{align}
	The input of $G_n$ is called the \emph{seed} and $s(n)$ is the seed length.
\end{definition}
Note that, in our context, the QPT algorithm need only be given classical access to the random variable, since Alice will be simply providing Bob and Charlie with strings sampled from this distribution. As such, that probability of $Q$ outputting $1$ takes its usual meaning as the probability measure of $Q^{-1}(\{1\})$.

Because of brute force attacks against $G_n$, $s(n)\in O(\lg n)$ is a strict lower bound on the seed length. Thus, we cannot hope for exponential randomness expansion with this method, but we can nevertheless expect large polynomial or even superpolynomial expansion. Now we can define a variant of the TFKW game that uses a pseudorandom rather than uniformly random question distribution.

\begin{definition}
	Let $G_n:\Z_2^{s(n)}\rightarrow\Z_2^n$ be a PRG and let $\Gamma$ be the uniform random variable on $\Z_2^{s(n)}$. The \emph{computational TFKW game} on $n$ qubits is the MoE game $\ttt{TFKW}^n_G=\parens*{\Z_2^n,\Z_2^n,\tsf{Q}^{\otimes n},\pi_G,A^n}$, where $(A^n)^\theta_x=\ketbra{x^\theta}$ as for the usual TFKW game and $\pi_G(\theta)=\Pr\squ*{G_n(\Gamma)=\theta}$.
\end{definition}

The set of strategies for the computational TFKW game is identical to that for the usual TFKW game, but where we restrict to families of strategies that can be modelled by QPT adversaries. As a warm-up to the main result of this section, we can see that against QPT strategies (shared state and measurements are all QPT), the usual and computational TFKW games behave essentially the same.

\begin{lemma}
	Let $n\mapsto\ttt{S}_n=\parens*{\tsf{B}_n,\tsf{C}_n,B_n,C_n,\rho_n}$ be a family of strategies with QPT adversaries. Then, assuming the existence of a PRG $G_n:\Z_2^{s(n)}\rightarrow\Z_2^n$, for every $i\in[n]$,
	\begin{align}
		\abs[\Big]{\mfk{w}^i_{\ttt{TFKW}^n_G}(\ttt{S}_n)-\mfk{w}^i_{\ttt{TFKW}^{n}}(\ttt{S}_n)}\in\ttt{negl}(n)
	\end{align}
\end{lemma}

\begin{proof}
	We will use $\ttt{S}_n$ to construct a QPT algorithm attempting to distinguish the variable $G_n(\Gamma)$ from uniformly random as follows. To compute $Q(\theta)$, measure the state $\rho_n$ with the POVM $(A^n)^\theta\otimes(B_n)^\theta\otimes(C_n)^\theta$. Output $1$ if the measurement result is some $(x,x^B,x^C)$ with $x_i=x^B_i=x^C_i$ and output $0$ otherwise. Then, for uniform random variables $\Gamma$ in $\Z_2^{s(n)}$ and $\Delta$ in $\Z_2^{n}$,
	\begin{align}
	\begin{split}
		&\abs[\Big]{\Pr\squ*{Q(G_{n}(\Gamma))=1}-\Pr\squ*{Q(\Delta)=1}}\\
		&=\abs[\Big]{\expec{\theta\leftarrow G_{n}(\Gamma)}\sum_{y\in\Z_2}\Tr\squ*{((A^n)^\theta_{y,i}\otimes(B_n)^\theta_{y,i}\otimes (C_n)^\theta_{y,i})\rho_n}-\expec{\theta\leftarrow\Delta}\sum_{y\in\Z_2}\Tr\squ*{((A^n)^\theta_{y,i}\otimes(B_n)^\theta_{y,i}\otimes (C_n)^\theta_{y,i})\rho_n}}\\
		&=\abs[\Big]{\mfk{w}^i_{\ttt{TFKW}^n_G}(\ttt{S}_n)-\mfk{w}^i_{\ttt{TFKW}^{n}}(\ttt{S}_n)}.
	\end{split}
	\end{align}
	We may conclude by noting that the left-hand side is contained in $\ttt{negl}(n)$ by hypothesis.
\end{proof}

\newpage 

Now, we can formally present the protocol.

\begin{mdframed}
\begin{protocol}[randomness expansion]\label{prot:re}\hphantom{}
	\begin{enumerate}
		
		\item Alice samples $(t,u)\in\Z_2^{s(N)+s\parens*{\ceil*{\lg\binom{N}{n}}}}$ uniformly at random. She computes $\theta=G_{N}(t)$ and $J=G_{\ceil*{\lg\binom{N}{n}}}(u)$, where she interprets $J$ as a subset of $[N]$ of cardinality $n$.
		
		\item Bob and Charlie prepare a shared state $\rho_{ABC}$ and then are isolated.
		
		\item Alice measures each of her qubits $i\in[N]$ in basis $\{\ket{0^{\theta_i}},\ket{1^{\theta_i}}\}$ if $i\notin J$ and in basis $\{\ket{0_\circlearrowleft},\ket{1_\circlearrowleft}\}$ if $i\in J$. This produces a string $y\in\Z_2^N$ that she keeps.
		
		\item Alice sends Bob and Charlie the key $\theta$ and $J$. Bob and Charlie each reply with a guess of $y$, $y^B$ and $y^C$ respectively.
		
		\item Alice verifies that they win the TFKW game $y_i=y^B_i=y^C_i$ for at least $(\cos^2\parens*{\frac{\pi}{8}}-\delta)N$ of the $i\in[N]\backslash J$, and then, if she accepts, takes $y_J$ to be her output.
	\end{enumerate}
\end{protocol}
\end{mdframed}

The protocol follows a very similar framework to \cref{prot:wse}, with the main differences being that Alice chooses her questions and test qubits only pseudorandomly, and measures always in the same basis to get her output. This basis is chosen to be mutually unbiased with all of the Breidbart states, and thus gives a uniformly random measurement result for any optimal strategy. Note that Bob and Charlie are able to make the protocol accept and provide randomness without using entanglement simply by preparing the Breidbart state $\ket{\beta}^{\otimes N}$, sending it to Alice, and guessing $0$ on all the TFKW game verification rounds.

Also, in this protocol, Alice shares $\theta$ and $J$ with Bob and Charlie immediately after she measures, so they have full information about her measurement bases. Thus, it doesn't affect the protocol if that information is leaked.

\begin{theorem}
	Let $K$ be the constant from \cref{thm:prob-gen-rob-par}, $\varepsilon,\eta,\delta\in(0,1)$ such that $\eta\varepsilon>\delta$, and $N\in\ttt{poly}(n)$. Assuming the existence of a pseudorandom generator, $G_n:\Z_2^{s(n)}\rightarrow\Z_2^n$, \cref{prot:re} is a $(s(N)+s(\lg\binom{N}{n}),2Kn^3\sqrt{\varepsilon}+2n\eta+\ttt{negl}(n))$-local randomness expander in the MoE model with QPT provers, that fails with probability $e^{-2N(\eta\varepsilon-\delta)^2}+\ttt{negl}(n)$.
\end{theorem}

The scenario in \cref{ex:numbers} allows us to take $N=n^{27}$, so provided that $s(n)\in o(n^{1/27})$ is possible, this yields randomness expansion.

\begin{proof}
	Write $b=\ceil*{\lg\binom{N}{n}}$. Let $U$ be the random variable representing the number of rounds Bob and Charlie win, let $V$ be the random variable representing the number of rounds they would have won if Alice chose $J$ uniformly at random (among the subsets of $[N]$ with cardinality $n$), and let~$W$ be the number of rounds they would have won if Alice chose both $J$ and $\theta$ uniformly at random.
	
	Take $Q(\theta,J)$ to be the QPT algorithm computed by running steps 2-5 of the randomness expansion protocol, and outputting $1$ if Alice accepts the verification of the TFKW games, and $0$ otherwise. Then, taking $\Gamma_1,\Gamma_2,\Delta_1,\Delta_2$ to be random variables in $\Z_2^{s(N)}$, $\Z_2^{s(b)}$, $\Z_2^{N}$, and $\Z_2^{b}$, respectively, we know that
	\begin{align}
	\begin{split}
		&\Pr\squ*{Q(G_N(\Gamma_1),G_{b}(\Gamma_2))=1}=\Pr\squ*{U\geq(\cos^2\parens*{\tfrac{\pi}{8}}-\delta)N},\\
		&\Pr\squ*{Q(G_N(\Gamma_1),\Delta_2)=1}=\Pr\squ*{V\geq(\cos^2\parens*{\tfrac{\pi}{8}}-\delta)N},\text{ and}\\ &\Pr\squ*{Q(\Delta_1,\Delta_2)=1}=\Pr\squ*{W\geq(\cos^2\parens*{\tfrac{\pi}{8}}-\delta)N},
	\end{split}
	\end{align}
	giving
	\begin{align}
	\begin{split}
		&\abs[\Big]{\Pr\squ*{U\geq(\cos^2\parens*{\tfrac{\pi}{8}}-\delta)N}-\Pr\squ*{W\geq(\cos^2\parens*{\tfrac{\pi}{8}}-\delta)N}}\\
		&\qquad\leq\abs[\Big]{\Pr\squ*{U\geq(\cos^2\parens*{\tfrac{\pi}{8}}-\delta)N}-\Pr\squ*{V\geq(\cos^2\parens*{\tfrac{\pi}{8}}-\delta)N}}\\
		&\qquad\qquad+\abs[\Big]{\Pr\squ*{V\geq(\cos^2\parens*{\tfrac{\pi}{8}}-\delta)N}-\Pr\squ*{W\geq(\cos^2\parens*{\tfrac{\pi}{8}}-\delta)N}}\\
		&\qquad\in\ttt{negl}(b)+\ttt{negl}(N)\subseteq\ttt{negl}(n),
	\end{split}
	\end{align}
	as $N\in\ttt{poly}(n)$ and $b\in O(n\lg n)$. Now, using \cref{lem:obs} as in \cref{thm:wse-sec}, if less than $(1-\eta)N$ of the rounds have winning probability greater than $\cos^2\parens*{\frac{\pi}{8}}-\varepsilon$, then \begin{align}
		\Pr\squ*{W\geq(\cos^2\parens*{\tfrac{\pi}{8}}-\delta)N}\leq e^{-2N(\eta\varepsilon-\delta)^2}.
	\end{align}
	By the above, then
	\begin{align}
		\Pr\squ*{U\geq(\cos^2\parens*{\tfrac{\pi}{8}}-\delta)N}\in e^{-2N(\eta\varepsilon-\delta)^2}+\ttt{negl}(n).
	\end{align}
	So, other than with negligible  failure probability, at least $(1-\eta)N$ of the rounds have winning probability greater than $\cos^2\parens*{\frac{\pi}{8}}-\varepsilon$.
	
	If we select $n$ rounds uniformly at random, each of the rounds has probability $1-\eta$ of winning with probability greater than $\cos^2\parens*{\frac{\pi}{8}}-\varepsilon$. Of course, the rounds are actually selected pseudorandomly: we claim that Bob and Charlie have a negligible probability of distinguishing the two cases. Let $E\subseteq[N]$ be the set of rounds that win with probability greater than $\cos^2\parens*{\frac{\pi}{8}}-\varepsilon$, and write $J=\{j_1(J),\ldots,j_n(J)\}$, where $j_1(J)<\ldots< j_n(J)$. Then, as Bob and Charlie's strategy is QPT, it is possible, for each $i\in[n]$, by using their strategy to play $\ttt{TFKW}$ a polynomial number of times, to get a QPT algorithm that, on input $J$, outputs whether $j_i(J)\in E$ correctly with $1-\ttt{negl}(n)$ probability. Thus, using pseudorandomness, we know
	\begin{align}
		\abs*{\Pr\squ*{j_i\parens*{G_{b}(\Gamma_2)}\in E}-\Pr\squ*{j_i(\Delta_2)\in E}}\in\ttt{negl}(n).
	\end{align}
	As $\Pr\squ*{j_i(\Delta_2)\in E}\geq 1-\eta$, each of the rounds chosen pseudorandomly has probability at least $1-\eta-\ttt{negl}(n)$ of having winning probability greater than $\cos^2\parens*{\frac{\pi}{8}}-\varepsilon$. So, by \cref{thm:prob-gen-rob-par} there exists a constant $K\geq 0$, isometries $V:\tsf{B}\rightarrow\tsf{B}'$ and $W:\tsf{C}\rightarrow\tsf{C}'$, an auxiliary register $R$, and a state $\ket{\phi}=\sum_{x,\varphi\in\Z_2^n}\be^{\otimes n}\ket{x^\varphi}\otimes\ket{x,\varphi}_{BCR}$ where the $\ket{x,\varphi}_{BCR}\in\tsf{B}'\otimes\tsf{C}'\otimes\tsf{R}$ have orthogonal support on both $\tsf{B}'$ and $\tsf{C}'$ such that
	\begin{align}
		\expec{J\leftarrow G_{b}(\Gamma_2)}\norm*{(V\otimes W)\rho_{A_JBC}(V\otimes W)^\dag-\Tr_R(\ketbra{\phi})}_{\Tr}\leq Kn^3\sqrt{\varepsilon}+n\eta+\ttt{negl}(n).
	\end{align}
	Let $\sigma_{A_JBC}=\Tr_R(\ketbra{\phi})$. If Alice measures her register in the basis $\set*{\ket{y_\circlearrowleft}}{y\in\Z_2^n}$, she gets
	\begin{align}
	\begin{split}
		\sigma_{YB}&=\sum_{y,x,\varphi\in\Z_2^n}\abs*{\braket{y_\circlearrowleft}{\be^{\otimes n}}{x^\varphi}}^2\ketbra{y}_Y\otimes\Tr_{CR}\parens*{\ketbra{x,\varphi}}_{B}\\
		&=\sum_{y,x,\varphi\in\Z_2^n}\frac{1}{2^n}\ketbra{y}_Y\otimes\Tr_{CR}\parens*{\ketbra{x,\varphi}}_{B}=\mu_Y\otimes\sigma_B.
	\end{split}
	\end{align}
	So, following the protocol, Alice measures $A_J$ of $\rho$ in this basis, giving
	\begin{align}
		\expec{J\leftarrow G_{b}(\Gamma_2)}\norm{V\rho_{Y_JB}V^\dag-\mu_Y\otimes\sigma_B}_{\Tr}\leq Kn^3\sqrt{\varepsilon}+n\eta+\ttt{negl}(n).
	\end{align}
	Acting with the trace non-increasing channel $\rho\mapsto V^\dag\rho V$,
	\begin{align}
		\expec{J\leftarrow G_{b}(\Gamma_2)}\norm{\rho_{Y_JB}-\mu_Y\otimes V^\dag\sigma_BV}_{\Tr}\leq Kn^3\sqrt{\varepsilon}+n\eta+\ttt{negl}(n),
	\end{align}
	 where in particular,
	 \begin{align}
	 	\expec{J\leftarrow G_{b}(\Gamma_2)}\norm{\rho_{B}-V^\dag\sigma_BV}_{\Tr}\leq Kn^3\sqrt{\varepsilon}+n\eta+\ttt{negl}(n),
	\end{align}
	 so, using the triangle inequality,
	\begin{align}
		\expec{J\leftarrow G_{b}(\Gamma_2)}\norm{\rho_{Y_JB}-\mu_Y\otimes\rho_B}_{\Tr}\leq 2Kn^3\sqrt{\varepsilon}+2n\eta+\ttt{negl}(n).
	\end{align}
	Let $S$ be a classical register holding the seed, and let $I=G_b(S)$ be the register that holds $J$. Then,
	\begin{align}
	\begin{split}
		\norm{\rho_{YSB}-\mu_Y\otimes\mu_S\otimes\rho_B}_{\Tr}&\leq\norm[\Big]{\expec{t\in\Z_2^{s(b)}}\ketbra{t}_S\otimes\ketbra{G_b(t)}_I\otimes\parens*{\rho_{Y_{G_b(t)}B}-\mu_Y\otimes\rho_B}}_{\Tr}\\
		&=\expec{t\in\Z_2^{s(b)}}\norm*{\rho_{Y_{G_b(t)}B}-\mu_Y\otimes\rho_B}_{\Tr}\\
		&\leq 2Kn^3\sqrt{\varepsilon}+2n\eta+\ttt{negl}(n)
	\end{split}
	\end{align}
	The same proof holds for $\rho_{YSC}$.
\end{proof}

\appendix

\section{Preliminary Lemmas} \label{sec:a-prelimlems}

\begin{lemma}[State purification]\label{lem:pure-state}
	Let $\tsf{H}$ be a Hilbert space. For any mixed state $\rho\in\mc{D}(\tsf{H})$, there exists a Hilbert space $\tsf{R}$ and a pure state $\ket{\psi}\in\tsf{H}\otimes\tsf{R}$ such that $\rho=\Tr_R(\ketbra{\psi})$.
\end{lemma}

\begin{lemma}[Measurement purification]\label{lem:pure-pvm}
	Let $\tsf{H}$ be a Hilbert space. Any POVM $P:X\rightarrow\mc{P}(\tsf{H})$ can be simulated by a PVM by enlarging the state space.
\end{lemma}

This lemma appeared in \cite{TFKW13} and was used in the same way as it is here. We see from the proof below that the enlargement of the state space is an isometry.

\begin{proof}
	Consider the linear operator
	\begin{align}
	\begin{matrix}V:&\tsf{H}&\rightarrow&\tsf{H}\otimes\C^X\\&\ket{\psi}&\mapsto&\sum_x\sqrt{P_x}\ket{\psi}\otimes\ket{x}.\end{matrix}
	\end{align}
	Since $\braket{\varphi}{V^\dag V}{\psi}=\sum_x\braket{\varphi}{P_x}{\psi}=\braket{\varphi}{\psi}$, we have that $V$ is an isometry. Fix some $x_0\in X$ and, identifying $\tsf{H}$ with the subspace $\tsf{H}\otimes\ket{x_0}$ of $\tsf{H}\otimes\tsf{X}$, we can extend $V$ to a unitary operator $U:\tsf{H}\otimes\tsf{X}\rightarrow \tsf{H}\otimes\tsf{X}$. Define now $\tilde{P}:X\rightarrow\mc{P}(\tsf{H}\otimes\tsf{X})$ by $\tilde{P}_x=U^\dag(\Id_H\otimes\ketbra*{x})U$. $\tilde{P}$ is a projective measurement as $\tilde{P}_x\tilde{P}_y=\braket{x}{y}U^\dag(\Id_H\otimes\ketbra*{x}{y})U=\delta_{x,y}\tilde{P}_x$, and $\sum_x\tilde{P}_x=U^\dag(\Id_H\otimes\Id_X)U=\Id_{HX}$. Finally, $\tilde{P}$ acts as $P$ on $\tsf{H}\otimes\ket{x_0}$ as
	\begin{align}
		(\bra{\varphi}\otimes\bra{x_0})\tilde{P}_x(\ket{\psi}\otimes\ket{x_0})=\sum_{y,y'}\braket{\varphi}{\sqrt{P_{y'}}\sqrt{P_y}}{\psi}\braket*{y'}{x}\braket*{x}{y}=\braket{\varphi}{P_x}{\psi}.
	\end{align}
\end{proof}

\begin{lemma}[Properties of the trace distance]\label{lem:tr-dist}
	Let $\tsf{H}$ and $\tsf{K}$ be Hilbert spaces.
	\begin{itemize}
		\item $d_{\Tr}:\mc{D}(\tsf{H})\times\mc{D}(\tsf{H})\rightarrow\R_+$ is a metric.
		\item For $\rho,\sigma\in\mc{D}(\tsf{H})$, $\norm{\rho-\sigma}_{\Tr}\leq 1$ with equality iff $\rho$ and $\sigma$ have orthogonal supports.
		\item For any quantum channel $\Phi:\mc{L}(\tsf{H})\rightarrow\mc{L}(\tsf{K})$, $\norm{\Phi(\rho)-\Phi(\sigma)}_{\Tr}\leq\norm{\rho-\sigma}_{\Tr}$.
		\item For pure states $\ket{\psi},\ket{\phi}\in\tsf{H}$, $\norm{\ketbra{\psi}-\ketbra{\phi}}_{\Tr}\leq\norm{\ket{\psi}-\ket{\phi}}$.
	\end{itemize}
\end{lemma}

\begin{lemma}[Properties of the operator norm]
	Let $\tsf{H}$ and $\tsf{K}$ be Hilbert spaces.
	\begin{itemize}
		\item The operator norm is a norm.
		\item The operator norm of a normal (e.g.,~Hermitian or unitary) operator is maximum of the moduli of the eigenvalues.
		\item For $A\in\mc{L}(\tsf{H})$ and $\ket{v}\in\tsf{H}$, $\norm{A\ket{v}}\leq\norm{A}\norm{\ket{v}}$.
		\item For $A,B\in\mc{L}(\tsf{H})$, $\norm{AB}\leq\norm{A}\norm{B}$.
		\item For any isometry $V:\tsf{H}\rightarrow\tsf{K}$, $\norm{VAV^\dag}=\norm{A}$.
	\end{itemize}
\end{lemma}

\fi

\newpage


\bibliographystyle{bibtex/bst/alphaarxiv.bst}
\bibliography{bibtex/bib/full.bib,bibtex/bib/quantum.bib,bibtex/bib/quantum_new.bib}

\end{document}